\documentclass[pra,onecolumn,superscriptaddress,nofootinbib]{revtex4}
\usepackage[a4paper, left=0.85in, right=0.85in, top=0.9in, bottom=0.9in]{geometry}

\usepackage{cmap} 
\usepackage[utf8]{inputenc}
\usepackage[english]{babel}
\usepackage[T1]{fontenc}
\usepackage{amsmath}
\usepackage{appendix}
\usepackage{amssymb}
\usepackage{mathtools}
\usepackage{amsfonts}
\usepackage{bm}
\usepackage{xcolor}
\definecolor{blueviolet}{rgb}{0.2, 0.2, 0.6}
\definecolor{webgreen}{rgb}{0,.5,0}
\definecolor{webbrown}{rgb}{.6,0,0}
\usepackage[pdftex,
	bookmarks=false,
	colorlinks=true, 
	urlcolor=webbrown,
	linkcolor=blueviolet,
	citecolor=webgreen,
	pdfstartpage=1,
	pdfstartview={FitH},  
	bookmarksopen=false
	]{hyperref}
\allowdisplaybreaks
\usepackage{tikz}
\usepackage{dsfont}
\usepackage{braket}
\usepackage{natbib}
\usepackage{amsthm}
\usepackage{listings}

\DeclareFixedFont{\ttb}{T1}{txtt}{bx}{n}{9} 
\DeclareFixedFont{\ttm}{T1}{txtt}{m}{n}{9}  

\usepackage{color}
\definecolor{deepblue}{rgb}{0,0,0.5}
\definecolor{deepred}{rgb}{0.6,0,0}
\definecolor{deepgreen}{rgb}{0,0.5,0}

\newcommand\pythonstyle{\lstset{
language=Python,
basicstyle=\ttm,
morekeywords={self},              
keywordstyle=\ttb\color{deepblue},
emph={MyClass,__init__},          
emphstyle=\ttb\color{deepred},    
stringstyle=\color{deepgreen},
frame=tb,                         
showstringspaces=false
}}

\lstnewenvironment{python}[1][]
{
\pythonstyle
\lstset{#1}
}
{}

\newcommand\pythoninline[1]{{\pythonstyle\lstinline!#1!}}

\newtheorem{theorem}{Theorem}
\newtheorem{task}{Task}
\newtheorem{corollary}{Corollary}

\newtheorem{definition}{Definition}
\newtheorem{lemma}{Lemma}
\newtheorem{fact}{Fact}

\newcommand{\indicator}{\mathds{1}}

\DeclareMathOperator*{\argmax}{arg\,max}
\DeclareMathOperator*{\argmin}{arg\,min}

\newtheorem{proposition}{Proposition}
\DeclareMathOperator{\Tr}{tr}
\DeclareMathOperator*{\E}{{\mathbb{E}}}

\newcommand{\ketbra}[2]{\lvert #1 \rangle \! \langle #2 \rvert}
\newcommand{\norm}[1]{\left\lVert#1\right\rVert}

\usepackage[noend]{algpseudocode}
\usepackage{algorithm,algorithmicx}

\algrenewcommand\alglinenumber[1]{\sf\scriptsize\color{blue}{#1}}
\algrenewcommand\algorithmicrequire{\textbf{Input:}}
\algrenewcommand\algorithmicensure{\textbf{Output:}}

\newcommand{\pol}{\mathsf{pol}}
\newcommand{\calS}{\mathcal{S}}

\newcommand{\wt}{\widetilde}

\newcommand{\calA}{\mathcal{A}}

\begin{document}

\title{Learning to predict arbitrary quantum processes}
\author{Hsin-Yuan Huang}
\affiliation{Institute for Quantum Information and Matter and \\ Department of Computing and Mathematical Sciences, Caltech, Pasadena, CA, USA}
\author{Sitan Chen}
\affiliation{Department of Electrical Engineering and Computer Sciences, UC Berkeley, Berkeley, CA, USA}
\author{John Preskill}
\affiliation{Institute for Quantum Information and Matter and \\ Department of Computing and Mathematical Sciences, Caltech, Pasadena, CA, USA}
\affiliation{AWS Center for Quantum Computing, Pasadena, CA, USA}
\date{\today}

\begin{abstract}
We present an efficient machine learning (ML) algorithm for predicting any unknown quantum process $\mathcal{E}$ over $n$ qubits.
For a wide range of distributions $\mathcal{D}$ on arbitrary $n$-qubit states, we show that this ML algorithm can learn to predict any local property of the output from the unknown process~$\mathcal{E}$, with a small average error over input states drawn from $\mathcal{D}$.
The ML algorithm is computationally efficient even when the unknown process is a quantum circuit with exponentially many gates. Our algorithm combines efficient procedures for learning properties of an unknown state and for learning a low-degree approximation to an unknown observable. The analysis hinges on proving new norm inequalities, including a quantum analogue of the classical Bohnenblust-Hille inequality, which we derive by giving an improved algorithm for optimizing local Hamiltonians.
Numerical experiments on predicting quantum dynamics with evolution time up to $10^6$ and system size up to $50$ qubits corroborate our proof.
Overall, our results highlight the potential for ML models to predict the output of complex quantum dynamics much faster than the time needed to run the process itself.
\end{abstract}

\maketitle

\vspace{-2em}
{\renewcommand\addcontentsline[3]{} \section{Introduction}}

Learning complex quantum dynamics is a fundamental problem at the intersection of machine learning (ML) and quantum physics.
Given an unknown $n$-qubit completely positive trace preserving (CPTP) map~$\mathcal{E}$ that represents a physical process happening in nature or in a laboratory,
we consider the task of learning to predict functions of the form
\begin{equation} \label{eq:frhoO}
    f(\rho, O) = \Tr(O \mathcal{E} (\rho)),
\end{equation}
where $\rho$ is an $n$-qubit state and $O$ is an $n$-qubit observable.
Related problems arise in many fields of research, including quantum machine learning \cite{biamonte2017quantum, schuld2019quantum, havlivcek2019supervised, caro2022generalization, schreiber2022classical, mcclean2018barren, caro2022out, huang2022quantum, farhi2018classification, arunachalam2017guest}, variational quantum algorithms \cite{gibbs2022dynamical, cirstoiu2020variational, peruzzo2014variational, kandala2017hardware, kokail2019self, cerezo2021variational, grimsley2019adaptive}, machine learning for quantum physics \cite{carleo2017solving, sharir2020deep, van2017learning, zhou2017optimizing, carrasquilla2017machine, parr1980density, car1985unified, becke1993new, white1993density, gilmer2017neural, huang2021provably, huang2020power}, and quantum benchmarking \cite{mohseni2008quantum, Scott08, o2004quantum, levy2021classical, huang2022foundations, merkel2013self, blume2017demonstration}.
As an example, for predicting outcomes of quantum experiments \cite{huang2021information, melnikov2018active, huang2022quantum}, we consider $\rho$ to be parameterized by a classical input $x$, $\mathcal{E}$ is an unknown process happening in the lab, and $O$ is an observable measured at the end of the experiment.
Another example is when we want to use a quantum ML algorithm to learn a model of a complex quantum evolution with the hope that the learned model can be faster \cite{cirstoiu2020variational, gibbs2022dynamical, caro2022out}.

As an $n$-qubit CPTP map $\mathcal{E}$ consists of exponentially many parameters, prior works, including those based on covering number bounds \cite{caro2022generalization, caro2022out, huang2021information, huang2022quantum}, classical shadow tomography \cite{levy2021classical, kunjummen2021shadow}, or quantum process tomography \cite{mohseni2008quantum, Scott08, o2004quantum}, require an exponential number of data samples to guarantee a small constant error for predicting outcomes of an arbitrary evolution~$\mathcal{E}$ under a general input state $\rho$.
To improve upon this, recent works \cite{chung2019sample, caro2022generalization, caro2022out, huang2021information, huang2022quantum} have considered quantum processes~$\mathcal{E}$ that can be generated in polynomial-time and shown that a polynomial amount of data samples suffices to learn $\Tr(O \mathcal{E} (\rho))$ in this restricted class.
However, these results still require exponential computation time.

In this work, we present a computationally-efficient ML algorithm that can learn a model of an arbitrary unknown $n$-qubit process $\mathcal{E}$, such that given $\rho$ sampled from a wide range of distributions over arbitrary $n$-qubit states and any $O$ in a large physically-relevant class of observables, the ML algorithm can accurately predict $f(\rho, O) = \Tr(O \mathcal{E}(\rho))$.
The ML model can predict outcomes for highly entangled states $\rho$ after learning from a training set that only contains data for random product input states and randomized Pauli measurements on the corresponding output states.
The training and prediction of the proposed ML model are both efficient even if the unknown process $\mathcal{E}$ is a Hamiltonian evolution over an exponentially long time, a quantum circuit with exponentially many gates, or a quantum process arising from contact with an infinitely large environment for an arbitrarily long time.
Furthermore, given few-body reduced density matrices (RDMs) of the input state $\rho$, the ML algorithm uses only classical computation to predict output properties $\Tr(O \mathcal{E}(\rho))$.

The proposed ML model is a combination of efficient ML algorithms for two learning problems:
(1) predicting $\Tr(O \rho)$ given a known observable $O$ and an unknown state $\rho$, and (2) predicting $\Tr(O \rho)$ given an unknown observable $O$ and a known state $\rho$.
We give sample- and computationally-efficient learning algorithms for both problems.
Then we show how to combine the two learning algorithms to address the problem of learning to predict $\Tr(O \mathcal{E}(\rho))$ for an arbitrary unknown $n$-qubit quantum process $\mathcal{E}$.
Together, the sample and computational efficiency of the two learning algorithms implies the efficiency of the combined ML algorithm.

In order to establish the rigorous guarantee for the proposed ML algorithms,
we consider a different task: optimizing a $k$-local Hamiltonian $H = \sum_{P \in \{I, X, Y, Z\}^{\otimes n}} \alpha_P P$.
We present an improved approximate optimization algorithm that finds either a maximizing/minimizing state $\ket{\psi}$ with a rigorous lower/upper bound guarantee on the energy $\bra{\psi} H \ket{\psi}$ in terms of the Pauli coefficients $\alpha_P$ of $H$.
The rigorous bounds improve upon existing results on optimizing $k$-local Hamiltonians \cite{Dinur2006OnTF, barak_et_al:LIPIcs:2015:5298, harrow2017extremal, anshu2021improved}.
We then use the improved optimization algorithm to give a constructive proof of several useful norm inequalities relating the spectral norm $\norm{O}$ of an observable $O$ and the $\ell_p$-norm of the Pauli coefficients $\alpha_P$ associated to the observable~$O$.
The proof resolves a recent conjecture in \cite{RWZ2022KKL} about the existence of quantum Bohnenblust-Hille inequalities.
These norm inequalities are then used to establish the efficiency of the proposed ML algorithms.

\begin{figure*}[t]
\centering
\includegraphics[width=0.88\textwidth]{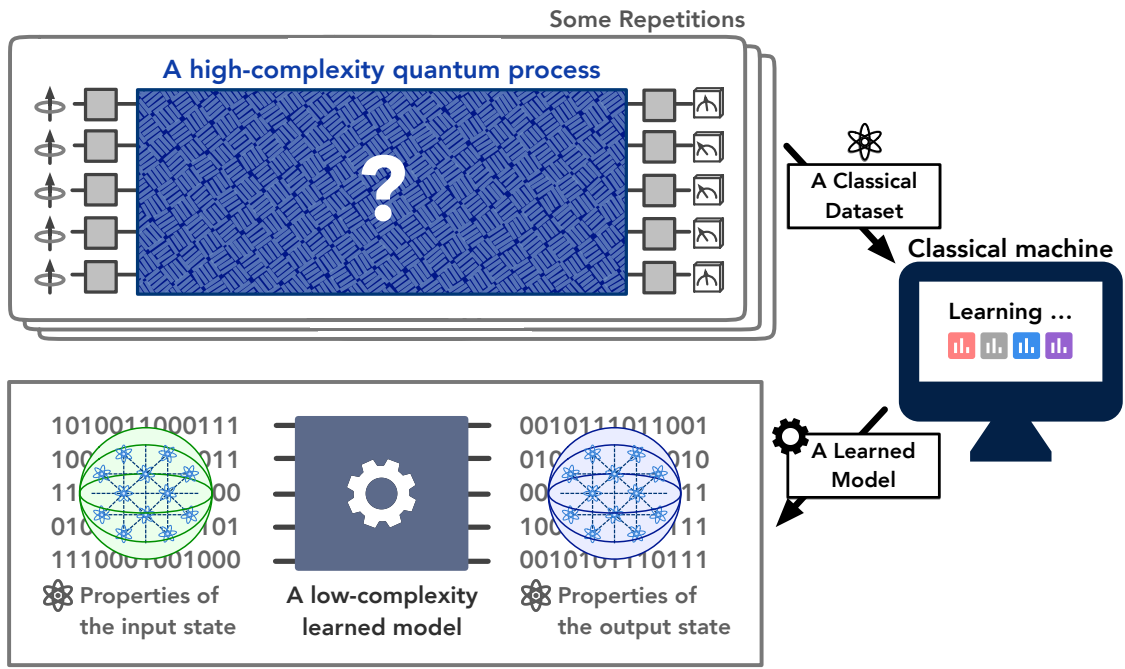}
    \caption{
    \emph{Learning to predict an arbitrary unknown quantum process $\mathcal{E}$.}
    Given an unknown quantum process $\mathcal{E}$ with arbitrarily high complexity, and a classical dataset obtained from evolving random product states under $\mathcal{E}$ and performing randomized Pauli measurements on the output states.
    We give an algorithm that can learn a low-complexity model for predicting the local properties of the output states given the local properties of the input states.
    \label{fig:LearnProcess}}
\end{figure*}

\vspace{2em}
{\renewcommand\addcontentsline[3]{} \section{Learning quantum states, observables, and processes}\label{sec:learn-intro}}

Before proceeding to state our main results in greater detail, we describe informally the learning tasks discussed in this paper: what do we mean 
by learning a quantum state, observable, and process?

\vspace{2em}
{\renewcommand\addcontentsline[3]{} \subsection{Learning an unknown state} \label{subsec:state-learn}}
It is possible in principle to provide a complete classical description of an $n$-qubit quantum state~$\rho$.
However this would require an exponential number of experiments, which is not at all practical.
Therefore, we set a more modest goal: to learn enough about $\rho$ to predict many of its physically relevant properties.
We specify a family of target observables $\{O_i\}$ and a small target accuracy $\epsilon$. The learning procedure is judged to be successful if we can predict the expectation value $\Tr (O_i \rho)$ of every observable in the family with error no larger than $\epsilon$.

Suppose that $\rho$ is an arbitrary and unknown $n$-qubit quantum state, and suppose that we have access to $N$ identical copies of $\rho$. We acquire information about $\rho$ by measuring these copies. In principle, we could consider performing collective measurements across many copies at once. Or we might perform single-copy measurements sequentially and \textit{adaptively}; that is, the choice of measurement performed on copy $j$ could depend on the outcomes obtained in measurements on copies $1,2, 3, \dots j{-}1.$
The target observables we consider are \textit{bounded-degree observables}. A bounded-degree $n$-qubit observable $O$ is a sum of local observables (each with support on a constant number of qubits independent of $n$) such that only a constant number (independent of $n$) of terms in the sum act on each qubit.
Most thermodynamic quantities that arise in quantum many-body physics can be written as a bounded-degree observable $O$, such as a  geometrically-local Hamiltonian or the average magnetization.

In the learning protocols discussed in this paper, the measurements are neither collective nor adaptive.
Instead, we fix an ensemble of possible single-copy measurements, and for each copy of $\rho$ we independently sample from this ensemble and perform the selected measurement on that copy.
Thus there are two sources of randomness in the protocol --- the randomly chosen measurement on each copy, and the intrinsic randomness of the quantum measurement outcomes. If we are unlucky, the chosen measurements and/or the measurement outcomes might not be sufficiently informative to allow accurate predictions. We will settle for a protocol that achieves the desired prediction task with a high success probability.

For the protocol to be practical, it is highly advantageous for the sampled measurements to be easy to perform in the laboratory, and easy to describe in classical language.
The measurements we consider, \textit{random Pauli measurements}, meet both of these criteria. For each copy of $\rho$ and for each of the $n$ qubits, we choose uniformly at random to measure one of the three single-qubit Pauli observables $X$, $Y$, or $Z$.
This learning method, called \textit{classical shadow tomography}, was analyzed in \cite{huang2020predicting}, where an upper bound on the sample complexity (the number $N$ of copies of $\rho$ needed to achieve the task) was expressed in terms of a quantity called the \textit{shadow norm} of the target observables.

In this work, using a new norm inequality derived here, we improve on the result in \cite{huang2020predicting} by obtaining a tighter upper bound on the shadow norm for bounded degree observables.
The upshot is that, for a fixed target accuracy $\epsilon$, we can predict all bounded-degree observables with spectral norm less than $B$ by performing random Pauli measurement on
\begin{equation}
N = \mathcal{O}\left( \log(n) B^2 / \epsilon^2 \right)
\end{equation}
copies of $\rho$.
This result improves upon the previously known bound of
$\mathcal{O}(n \log(n) B^2 / \epsilon^2)$.
Furthermore, we derive a matching lower bound on the number of copies required for this task, which applies even if collective measurements across many copies are allowed.

\vspace{2em}
{\renewcommand\addcontentsline[3]{} \subsection{Learning an unknown observable} \label{subsec:observable-learn}}

Now suppose that $O$ is an arbitrary and unknown $n$-qubit observable. We also consider a distribution $\mathcal{D}$ on $n$-qubit quantum states. This distribution, too, need not be known, and it may include highly entangled states. Our goal is to find a function $h(\rho)$ which predicts the expectation value $\Tr (O\rho)$ of the observable $O$ on the state $\rho$ with a small mean squared error:
\begin{equation*}
    \E_{\rho \sim \mathcal{D}} \left| h(\rho) - \Tr(O \rho) \right|^2 \leq \epsilon.
\end{equation*}
To define this learning task, it is convenient to assume that we can access training data of the form
\begin{equation}\label{eq:observable-training}
    \left\{ \rho_\ell, \Tr\left(O \rho_\ell\right) \right\}_{\ell = 1}^N,
\end{equation}
where $\rho_\ell$ is sampled from the distribution $\mathcal{D}$. In practice, though, we cannot directly access the exact value of the expectation value $\Tr\left(O \rho_\ell\right)$; instead, we might measure $O$ multiple times in the state $\rho_\ell$ to obtain an accurate estimate of the expectation value. Furthermore, we don't necessarily need to sample states from $\mathcal{D}$ to achieve the task. We might prefer to learn about $O$ by accessing its expectation value in states drawn from a different ensemble.

A crucial idea of this work is that we can learn $O$ efficiently if the distribution $\mathcal{D}$ has suitably nice features. Specifically, we consider distributions that are invariant under single-qubit Clifford gates applied to any one of the $n$ qubits. We say that such distributions are \textit{locally flat}, meaning that the probability weight assigned to an $n$-qubit state is unmodified (i.e., the distribution appears flat) when we locally rotate any one of the qubits.

An arbitrary observable $O$ can be expanded in terms of the Pauli operator basis:
\begin{equation}
    O = \sum_{P \in \{I, X, Y, Z\}^{\otimes n}} \alpha_P P.
    \end{equation}
Though there are $4^n$ Pauli operators, if the distribution $\mathcal{D}$ is locally flat and $O$ has a constant spectral norm, we can approximate the sum over $P$ by a truncated sum
\begin{equation}
    O^{(k)} = \sum_{P \in \{I, X, Y, Z\}^{\otimes n}: |P| \leq k} \alpha_P P.
\end{equation}
including only the Pauli operators $P$ with weight $|P|$ up to $k$, those acting nontrivially on no more than $k$ qubits.
The mean squared error incurred by this truncation decays exponentially with $k$.
Therefore, to learn $O$ with mean squared error $\epsilon$ it suffices to learn this truncated approximation to $O$, where $k=\mathcal{O}(\log(1/\epsilon))$.
Furthermore, using norm inequalities derived in this paper, we show that for the purpose of predicting the expectation value of this truncated operator it suffices to learn only a few relatively large coefficients $\alpha_P$, while setting the rest to zero. The upshot is that, for a fixed target error $\epsilon$, an observable with constant spectral norm can be learned from training data with size $\mathcal{O}(\log n)$, where the classical computational cost of training and predicting is $n^{\mathcal{O}(k)}$.

Usually, in machine learning, after learning from a training set sampled from a distribution $\mathcal{D}$, we can only predict new instances sampled from the same distribution $\mathcal{D}$. We find, though, that for the purpose of learning an unknown observable, there is a particular locally flat distribution $\mathcal{D}'$ such that learning to predict under $\mathcal{D}'$ suffices for predicting under any other locally flat distribution. Namely, we samples from the $n$-qubit state distribution $\mathcal{D}'$ by preparing each one of the $n$ qubits in one of the six Pauli operator eigenstates $\{\ket{0}, \ket{1}, \ket{+}, \ket{-}, \ket{y+}, \ket{y-}\}$, chosen uniformly at random. Pleasingly, preparing samples from $\mathcal{D}'$ is not only sufficient for our task, but also easy to do with existing quantum devices.

After training is completed, to predict $\Tr (O\rho)$ for a new state $\rho$ drawn from the distribution $\mathcal{D}$, we need to know some information about $\rho$.
The state $\rho$, like the operator $O$, can be expanded in terms of Pauli operators, and when we replace $O$ by its weight-$k$ truncation, only the truncated part of $\rho$ contributes to its expectation value.
Thus if the $k$-body reduced density matrices (RDMs) for states drawn from $\mathcal{D}$ are known classically, then the predictions can be computed classically.
If the states drawn from $\mathcal{D}$ are presented as unknown quantum states, then we can learn these $k$-body RDMs efficiently (for small $k$) using classical shadow tomography, and then proceed with the classical computation to obtain a predicted value of $\Tr (O\rho)$.

\vspace{2em}
{\renewcommand\addcontentsline[3]{} \subsection{Learning an unknown process} \label{subsec:process-learn}}

Now suppose that $\mathcal{E}$ is an arbitrary and unknown quantum process mapping $n$ qubits to $n$ qubits.
Let $\{O_i\}$ be a family of target observables, and $\mathcal{D}$ be a distribution on quantum states.
We assume the ability to repeatedly access $\mathcal{E}$ for a total of $N$ times. Each time we can apply $\mathcal{E}$ to an input state of our choice, and perform the measurement of our choice on the resulting output.
In principle we could allow input states that are entangled across the $N$ channel uses, and allow collective measurements across the $N$ channel outputs.
But here we confine our attention to the case where the $N$ inputs are unentangled, and the channel outputs are measured individually.
Our goal is to find a function $h(\rho,O)$ which predicts, with a small mean squared error, the expectation value of $O_i$ in the output state $\mathcal{E}(\rho)$ for every observable $O_i$ in the family $\{O_i\}$:
\begin{equation}
        \E_{\rho \sim \mathcal{D}} \left| h(\rho, O_i) - \Tr(O_i \mathcal{E}(\rho)) \right|^2 \leq \epsilon.
    \end{equation}
Our main result is that this task can be achieved efficiently if $O_i$ is a bounded-degree observable and $\mathcal{D}$ is locally flat.  That is, $N$, the number of times we access $\mathcal{E}$, and the computational complexity of training and prediction, scale reasonably with the system size $n$ and the target accuracy $\epsilon$.

To prove this result, we observe that the task of learning an unknown quantum process can be reduced to learning unknown states and learning unknown observables. If $\rho_\ell$ is sampled from the distribution $\mathcal{D}$, then, since $\mathcal{E}$ is unknown, $\mathcal{E}(\rho_\ell)$ should be regarded as an unknown quantum state. Suppose we learn this state; that is, after preparing and measuring $\mathcal{E}(\rho_\ell)$ sufficiently many times we can accurately predict the expectation value $\Tr(O_i \mathcal{E}(\rho_\ell))$ for each target observable $O_i$.

Now notice that $\Tr(O_i \mathcal{E}(\rho_\ell))= \Tr(\mathcal{E}^\dagger(O_i)\rho_\ell)$, where $\mathcal{E}^\dagger$ is the (Heisenberg-picture) map dual to $\mathcal{E}$. Since $\mathcal{E}^\dagger$ is unknown, $\mathcal{E}^\dagger(O_i)$ should be regarded as an unknown observable. Suppose we learn this observable; that is, using the dataset $\{ \rho_\ell,\Tr(\mathcal{E}^\dagger(O_i)\rho_\ell) \}$ as training data, we can predict $\Tr(\mathcal{E}^\dagger(O_i)\rho)$ for $\rho$ drawn from $\mathcal{D}$ with a small mean squared error.
This achieves the task of learning the process $\mathcal{E}$ for state distribution $\mathcal{D}$ and target observable $O_i$.

Having already shown that arbitrary quantum states can be learned efficiently for the purpose of predicting expectation values of bounded-degree observables, and that arbitrary observables can be learned efficiently for locally flat input state distributions, we obtain our main result. Since the distribution $\mathcal{D}$ is locally flat, it suffices to learn the low-degree truncated approximation to the unknown operator $\mathcal{E}^\dagger(O_i)$, incurring only a small mean squared error. To predict $\Tr(\mathcal{E}^\dagger(O_i)\rho)$, then, it suffices to know only the few-body RDMs of the input state $\rho$. For any input state $\rho$, these few-body density matrices can be learned efficiently using classical shadow tomography.

As noted above in the discussion of learning observables, the states $\rho_\ell$ in the training data need not be sampled from $\mathcal{D}$. To learn a low-degree approximation to $\mathcal{E}^\dagger(O_i)$, it suffices to sample from a locally flat distribution on product states. Even if we sample only product states during training, we can make accurate predictions for highly entangled input states.
We also emphasize again that the unknown process $\mathcal{E}$ is arbitrary. Even if $\mathcal{E}$ has quantum computational complexity exponential in $n$, we can learn to predict $\Tr(O \mathcal{E}(\rho) )$ accurately and efficiently, for bounded-degree observables $O$ and for any locally flat distribution on the input state $\rho$.



\vspace{2em}
{\renewcommand\addcontentsline[3]{} \section{Algorithm for learning an unknown quantum process}}

Consider an unknown $n$-qubit quantum process $\mathcal{E}$ (a CPTP map).
Suppose we have obtained a classical dataset by performing $N$ randomized experiments on $\mathcal{E}$.
Each experiment prepares a random product state $\ket{\psi^{\mathrm{(in)}}} = \bigotimes_{i=1}^n \ket{s^{\mathrm{(in)}}_i}$, passes through $\mathcal{E}$, and performs a randomized Pauli measurement \cite{huang2020predicting, elben2022randomized} on the output state.
Recall that a randomized Pauli measurement measures each qubit of a state in a random Pauli basis ($X$, $Y$ or $Z$) and produces a measurement outcome of $\ket{\psi^{\mathrm{(out)}}} = \bigotimes_{i=1}^n \ket{s^{\mathrm{(out)}}_i}$, where $\ket{s^{\mathrm{(out)}}_i} \in \mathrm{stab}_1 \triangleq \{\ket{0}, \ket{1}, \ket{+}, \ket{-}, \ket{y+}, \ket{y-}\}$.
We denote the classical dataset of size $N$ to be
\begin{equation} \label{eq:classicalSn-E}
    S_N(\mathcal{E}) \triangleq \left\{ \ket{\psi^{\mathrm{(in)}}_\ell} = \bigotimes_{i=1}^n \ket{s^{\mathrm{(in)}}_{\ell, i}}, \,\, \ket{\psi^{\mathrm{(out)}}_\ell} = \bigotimes_{i=1}^n \ket{s^{\mathrm{(out)}}_{\ell, i}} \right\}_{\ell = 1}^N,
\end{equation}
where $\ket{s^{\mathrm{(in)}}_{\ell, i}}, \ket{s^{\mathrm{(out)}}_{\ell, i}} \in \mathrm{stab}_1$.
Each product state is represented classically with $\mathcal{O}(n)$ bits.
Hence, the classical dataset $S_N(\mathcal{E})$ is of size $\mathcal{O}(n N)$ bits.
The classical dataset can be seen as one way to generalize the notion of classical shadows of quantum states~\cite{huang2020predicting} to quantum processes.
Our goal is to design an ML algorithm that can learn an approximate model of $\mathcal{E}$ from the classical dataset $S_N(\mathcal{E})$, such that for a wide range of states~$\rho$ and observables~$O$, the ML model can predict a real value $h(\rho, O)$ that is approximately equal to $\Tr(O \mathcal{E}(\rho))$.

\vspace{2em}
{\renewcommand\addcontentsline[3]{} \subsection{ML algorithm} \label{sec:ML-algo}}

We are now ready to state the proposed ML algorithm.
At a high level, the ML algorithm learns a low-degree approximation to the unknown $n$-qubit CPTP map $\mathcal{E}$.
Despite the simplicity of the ML algorithm, several ideas go into the design of the ML algorithm and the proof of the rigorous performance guarantee. These ideas are presented in Section~\ref{sec:idea-proof}.

Let $O$ be an observable with $\norm{O} \leq 1$ that is written as a sum of few-body observables, where each qubit is acted by $\mathcal{O}(1)$ of the few-body observables.
We denote the Pauli representation of $O$ as $\sum_{Q \in \{I, X, Y, Z\}^{\otimes n}} a_Q Q$.
By definition of $O$, there are $\mathcal{O}(n)$ nonzero Pauli coefficients $a_Q$.
We consider a hyperparameter $\tilde{\epsilon} > 0$; roughly speaking $\tilde{\epsilon}$ will scale inverse polynomially in the dataset size $N$ from Eq.~\eqref{eq:N-samp-final}.
For every Pauli observable $P \in \{I, X, Y, Z\}^{\otimes n}$ with $|P| \leq k = \Theta( \log(1 / \epsilon) )$, the algorithm computes an empirical estimate for the corresponding Pauli coefficient $\alpha_P$ via
\begin{align}
    \hat{x}_P(O) &= \frac{1}{N} \sum_{\ell=1}^{N} \Tr\left(P \bigotimes_{i=1}^n \ketbra{s^{\mathrm{(in)}}_{\ell, i}}{s^{\mathrm{(in)}}_{\ell, i}} \right) \Tr\left( O \bigotimes_{i=1}^n\left(3 \ketbra{s^{\mathrm{(out)}}_{\ell, i}}{s^{\mathrm{(out)}}_{\ell, i}} - I \right) \right), \label{eq:hatxP}\\
    \hat{\alpha}_P(O) &=
    \begin{cases}
        3^{|P|} \hat{x}_P(O), & \left(\frac{1}{3}\right)^{|P|} > 2 \tilde{\epsilon} \,\,\,\, \mbox{and} \,\,\, |\hat{x}_P(O)| > 2 \cdot 3^{|P| / 2} \sqrt{\tilde{\epsilon}} \sum_{Q: a_Q \neq 0} |a_Q|,\\
        0, & \mbox{otherwise}.
    \end{cases}
     \label{eq:hatalphaP}
\end{align}
The computation of $\hat{x}_P(O)$ and $\hat{\alpha}_P(O)$ can both be done classically.
The basic idea of $\hat{\alpha}_P(O)$ is to set the coefficient $3^{|P|} \hat{x}_P(O)$ to zero when the influence of Pauli observable $P$ is negligible.
Given an $n$-qubit state $\rho$, the algorithm outputs
\begin{equation} \label{eq:h-rhoO}
    h(\rho, O) = \sum_{P: |P| \leq k} \hat{\alpha}_P(O) \Tr(P\rho).
\end{equation}
With a proper implementation, the computational time is $\mathcal{O}(k n^k N)$.
Note that, to make predictions, the ML algorithm only needs the $k$-body reduced density matrices ($k$-RDMs) of $\rho$.
The $k$-RDMs of $\rho$ can be efficiently obtained by performing randomized Pauli measurement on $\rho$ and using the classical shadow formalism \cite{huang2020predicting, elben2022randomized}.
Except for this step, which may require quantum computation, all other steps of the ML algorithm only requires classical computation.
Hence, if the $k$-RDMs of $\rho$ can be computed classically, then we have a classical ML algorithm that can predict an arbitrary quantum process~$\mathcal{E}$ after learning from data.

\vspace{2em}
{\renewcommand\addcontentsline[3]{} \subsection{Rigorous guarantee}}

To measure the prediction error of the ML model, we consider the average-case prediction performance under an arbitrary $n$-qubit state distribution $\mathcal{D}$ invariant under single-qubit Clifford gates, which means that the probability distribution $f_\mathcal{D}(\rho)$ of sampling a state $\rho$ is equal to $f_\mathcal{D}(U \rho U^\dagger)$ of sampling $U \rho U^\dagger$ for any single-qubit Clifford gate $U$.
We call such a distribution locally flat.

\vspace{0.5em}
\begin{theorem}[Learning an unknown quantum process] \label{thm:main-res}
    Given $\epsilon, \epsilon' = \Theta(1)$ and a training set $S_N(\mathcal{E})$ of size $N = \mathcal{O}(\log n)$ as specified in Eq.~\eqref{eq:classicalSn-E}.
    With high probability, the ML model can learn a function $h(\rho, O)$ from $S_N(\mathcal{E})$ such that for any distribution $\mathcal{D}$ over $n$-qubit states invariant under single-qubit Clifford gates, and for any bounded-degree observable $O$ with $\norm{O} \leq 1$,
    \begin{equation}
        \E_{\rho \sim \mathcal{D}} \left| h(\rho, O) - \Tr(O \mathcal{E}(\rho)) \right|^2 \leq \epsilon + \max\left(\norm{O'}^2, 1\right) \epsilon',
    \end{equation}
    where $O'$ is the low-degree truncation (of degree $k = \lceil \log_{1.5}(1 / \epsilon) \rceil$) of the observable $O$ after the Heisenberg evolution under~$\mathcal{E}$.
    The training and prediction time of $h(\rho, O)$ are both polynomial in $n$.
    When $\epsilon$ is small and $\epsilon' = 0$, the data size $N$ and computational time scale as $2^{\mathcal{O}(\log(\frac{1}{\epsilon}) \log(n))}$.
\end{theorem}

The detailed theorem statement and the proof of the theorem are given in Appendix~\ref{sec:learning-qevo}.
An interesting aspect of the above theorem is that the states sampled from the distribution $\mathcal{D}$ can be highly entangled, even though the training data $S_N(\mathcal{E})$ only contains information about random product states.
From the theorem, we can see that if $\norm{O'} = \mathcal{O}(1)$, then we only need $\mathcal{O}(\log(n))$ samples to obtain a constant prediction error.
Otherwise, $\mathcal{O}(\log(n))$ samples is still enough to guarantee a constant prediction error relative to $\norm{O'}^2$.
The precise scaling is given as follows. Consider data size
\begin{equation} \label{eq:N-samp-final}
    N = \log(n) \, \min\left( 2^{\mathcal{O}\left(\log(\frac{1}{\epsilon}) \left( \log\log(\frac{1}{\epsilon}) +  \log(\frac{1}{\epsilon'})\right)\right)}, 2^{\mathcal{O}(\log(\frac{1}{\epsilon}) \log(n))} \right).
\end{equation}
The computational time to learn and predict $h(\rho, O)$ is bounded above by $\mathcal{O}(k n^k N)$ and the prediction error is bounded as
\begin{equation} \label{eq:pred-error-final}
    \E_{\rho \sim \mathcal{D}} \left| h(\rho, O) - \Tr(O \mathcal{E}(\rho) ) \right|^2 \leq \epsilon + \max\left(\norm{O'}^2, 1\right) \epsilon'.
\end{equation}
As we take $\epsilon'$ to be zero, we can remove the dependence on the low-degree truncation $O'$. In this setting, $N$ and computation time both become $2^{\mathcal{O}(\log(\frac{1}{\epsilon}) \log(n))}$, which is polynomial in $n$ if $\epsilon = \Theta(1)$.

{\renewcommand\addcontentsline[3]{} \section{Proof ideas} \label{sec:idea-proof}}

The proof of the rigorous performance guarantee for the proposed ML algorithm consists of five parts.
The first two parts presented in Appendix~\ref{sec:optimize-klocal} and Appendix~\ref{sec:norm-ineq-Pauli} are a detour to establish a few fundamental and useful norm inequalities about Hamiltonians/observables.
The latter three parts given in Appendix~\ref{sec:samp-opt-bd-local-H}, Appendix~\ref{sec:learn-unk-obs}, and Appendix~\ref{sec:learning-qevo} 
apply the newly-established norm inequalities to 
three learning tasks.
In the following, we present the basic ideas in each part.

\vspace{2em}
{\renewcommand\addcontentsline[3]{} \subsection{Improved approximation algorithms for optimizing local Hamiltonians}}

We begin with a different task, namely optimizing local Hamiltonians.
We are given an $n$-qubit $k$-local Hamiltonian
\begin{equation}
H = \sum_{P \in \{I, X, Y, Z\}^{\otimes n}: |P|\leq k} \alpha_P P,
\end{equation}
where $|P|$ is the weight of the Pauli operator $P$, the number of qubits upon which $P$ acts nontrivially.
Our goal is to find a state $\ket{\psi}$ that maximizes/minimizes $\bra{\psi} H \ket{\psi}$.
This task is related to solving ground states \cite{kempe2006complexity, sakurai_napolitano_2017} when we consider minimizing $\bra{\psi} H \ket{\psi}$ and quantum optimization \cite{farhi2014quantum,Farhi2014AQA,harrow2017extremal, parekh2020beating, Hallgren2020AnAA, anshu2021improved, hastings2022optimizing} when we consider maximizing $\bra{\psi} H \ket{\psi}$.

We give a general randomized approximation algorithm in Appendix~\ref{sec:optimize-klocal} for producing a random product state $\ket{\psi}$ that either approximately minimizes or approximately maximizes a $k$-local Hamiltonian $H$ with a rigorous upper/lower bound based on the Pauli coefficients $\alpha_P$ of $H$.
The proposed optimization algorithm applies to various classes of Hamiltonians and is inspired by the proofs of Littlewood's 4/3 inequality \cite{littlewood1930bounded} and the Bohnenblust-Hille inequality \cite{BHineq1931}.
For classes that have been studied previously \cite{Dinur2006OnTF, barak_et_al:LIPIcs:2015:5298, harrow2017extremal, anshu2021improved}, the proposed algorithm obtains an improved bound.
Our improvement crucially stems from our construction for the random state $\ket{\psi}$.
\cite{Dinur2006OnTF, barak_et_al:LIPIcs:2015:5298, harrow2017extremal} utilize a random restriction approach, where some random subset of qubits are fixed with some random values and the rest of the qubits are optimized.
On the other hand, we utilize a polarization approach, where we replicate each qubit many times, randomly fix all except the last replica, optimize the last replica, and combine using a random-signed averaging.
A detailed comparison is given in Appendix~\ref{sec:alt-main-theorem}~and~\ref{sec:cor-main-opt-thm}.

Two classes of Hamiltonians used in our learning applications are general $k$-local Hamiltonians and bounded-degree $k$-local Hamiltonians.
A $k$-local Hamiltonian with degree at most $d$ is a Hermitian operator that can be written as a sum of $k$-qubit observables, where each qubit is acted on by at most $d$ of the $k$-qubit observables.

\vspace{0.8em}
\begin{corollary}[Optimizing general $k$-local Hamiltonian]
    Given an $n$-qubit $k$-local Hamiltonian
    \begin{equation}
        H = \sum_{P: |P| \leq k} \alpha_P P.
    \end{equation}
    There is a randomized algorithm that runs in time $\mathcal{O}(n^k)$ and produces either a random maximizing state $\ket{\psi} = \ket{\psi_1} \otimes \ldots \otimes \ket{\psi_n}$ satisfying
    \begin{equation}
        \E_{\ket{\psi}} \big[ \bra{\psi} H \ket{\psi} \big] \geq \E_{\ket{\phi}: \mathrm{Haar}} \big[ \bra{\phi} H \ket{\phi} \big] + C(k) \left(\sum_{P \neq I} |\alpha_P|^{2k / (k+1)}\right)^{(k+1) / (2k)},
    \end{equation}
    or a random minimizing state $\ket{\psi} = \ket{\psi_1} \otimes \ldots \otimes \ket{\psi_n}$ satisfying
    \begin{equation}
        \E_{\ket{\psi}} \big[ \bra{\psi} H \ket{\psi} \big] \leq \E_{\ket{\phi}: \mathrm{Haar}} \big[ \bra{\phi} H \ket{\phi} \big] - C(k) \left(\sum_{P \neq I} |\alpha_P|^{2k / (k+1)}\right)^{(k+1) / (2k)},
    \end{equation}
    where $C(k) = 1 / \exp(\Theta(k \log k))$.
\end{corollary}

\begin{corollary}[Optimizing bounded-degree $k$-local Hamiltonian]
    Given an $n$-qubit $k$-local Hamiltonian $H = \sum_{P: |P| \leq k} \alpha_P P$ with bounded degree $d$, $|\alpha_P| \leq 1$ for all $P$, and $k = \mathcal{O}(1)$.
    There is a randomized algorithm that runs in time $\mathcal{O}(nd)$ and produces either a random maximizing state $\ket{\psi} = \ket{\psi_1} \otimes \ldots \otimes \ket{\psi_n}$ satisfying
    \begin{equation}
        \E_{\ket{\psi}} \big[ \bra{\psi} H \ket{\psi} \big] \geq \E_{\ket{\phi}: \mathrm{Haar}} \big[ \bra{\phi} H \ket{\phi} \big] + \frac{C}{\sqrt{d}} \sum_{P \neq I} |\alpha_P|,
    \end{equation}
    or a random minimizing state $\ket{\psi} = \ket{\psi_1} \otimes \ldots \otimes \ket{\psi_n}$ satisfying
    \begin{equation}
        \E_{\ket{\psi}} \big[ \bra{\psi} H \ket{\psi} \big] \leq \E_{\ket{\phi}: \mathrm{Haar}} \big[ \bra{\phi} H \ket{\phi} \big] - \frac{C}{\sqrt{d}} \sum_{P \neq I} |\alpha_P|
    \end{equation}
    for some constant $C$.
\end{corollary}

We note that in the above results, we cannot control whether our algorithm outputs an approximate maximizer or minimizer.
This caveat stems from the use of polarization, where the random-signed averaging only guarantees improvement in one of the two directions.
Modifying our approach to address this issue is an interesting direction for future work.

{\renewcommand\addcontentsline[3]{} \subsection{Norm inequalities from approximate optimization algorithms}}

The bridge that connects the optimization of $k$-local Hamiltonians and efficient learning of quantum states and processes is a set of norm inequalites.
A norm that characterizes the efficiency of learning is the Pauli-$p$ norm, defined as the $\ell_p$-norm on the Pauli coefficients of an observable/Hamiltonian $H = \sum_{P} \alpha_P P$,
\begin{equation}
    \norm{H}_{\mathrm{Pauli}, p} \triangleq \left(\sum_{P \in \{I, X, Y, Z\}^{\otimes n}} |\alpha_{P}|^p \right)^{1/p}.
\end{equation}
The rigorous guarantees from the previous section, namely on finding a state $\ket{\psi}$ whose energy is higher/lower than a Haar-random state by a margin that depends on the Pauli coefficients $\alpha_P$, give an algorithmic proof that the spectral norm $\norm{H}$ and the Pauli coefficients $\alpha_P$ are related.
The proof of this relation is given in Appendix~\ref{sec:norm-ineq-Pauli}.
In particular, for general and bounded-degree $k$-local Hamiltonian, we can use the rigorous guarantee from the approximation algorithms to obtain the following norm inequalites.
Corollary~\ref{cor:norm-ineq-klocal-main} proves the conjecture given in \cite{RWZ2022KKL}.

\vspace{0.8em}
\begin{corollary}[Norm inequality for general $k$-local Hamiltonian] \label{cor:norm-ineq-klocal-main}
    Given an $n$-qubit $k$-local Hamiltonian $H$. We have
    \begin{equation}
        \frac{1}{3} C(k) \norm{H}_{\mathrm{Pauli}, \frac{2k}{k + 1}} \leq \norm{H},
    \end{equation}
    where $C(k) = 1 / \exp(\Theta(k \log k))$.
\end{corollary}

\begin{corollary}[Norm inequality for bounded-degree local Hamiltonian] \label{cor:norm-ineq-bd-klocal-main}
    Given an $n$-qubit $k$-local Hamiltonian $H$ with a bounded degree $d$. We have
    \begin{equation}
        \frac{1}{3} C(k, d) \norm{H}_{\mathrm{Pauli}, 1} \leq \norm{H},
    \end{equation}
    where $C(k, d) = 1 / (\sqrt{d}\exp(\Theta(k \log k)))$.
\end{corollary}

{\renewcommand\addcontentsline[3]{} \subsection{Sample-optimal algorithm for predicting bounded-degree observables}}

As the first application of the above norm inequalities to learning, we consider the basic problem of predicting many properties of an unknown $n$-qubit state $\rho$.
Given $M$ observables $O_1, \ldots, O_M$, after performing measurements on multiple copies of $\rho$, we would like to predict $\Tr(O_i \rho)$ to $\epsilon$ error for all $i \in \{1, \ldots, M\}$.
This is the task known as shadow tomography \cite{aaronson2018shadow, aaronson2019gentle, huang2020predicting}.
One approach for obtaining practically-efficient algorithms for shadow tomography is via the classical shadow formalism \cite{huang2020predicting}.

We consider a physically-relevant class of observables, where the observable $O_i = \sum_{j} O_{ij}$ is a sum of few-body observables $O_{ij}$ and each qubit is acted on by $\mathcal{O}(1)$ of the few-body observables.
Despite significant recent progress in shadow tomography \cite{levy2021classical, zhao2021fermionic, hu2021hamiltonian, koh2020classical, chen2020robust, hadfield2020measurements, aaronson2018shadow, struchalin2020experimental, huang2022learning, o2022fermionic, wan2022matchgate, bu2022classical, huang2022quantum, chen2022exponential, coopmans2022predicting}, the sample complexity (number of copies of $\rho$) for predicting this class of observables has not been established.
The central challenge is the appearance of the Pauli-$1$ norm $\norm{O_i}_{\mathrm{Pauli}, 1}$ when characterizing the sample complexity.
In particular, one can bound the shadow norm $\norm{O_i}_{\mathrm{shadow}}$ \cite{huang2020predicting}, which gives an upper bound on the sample complexity in terms of the Pauli-$1$ norm $\norm{O_i}_{\mathrm{Pauli}, 1}$ up to a constant factor.
Using the new norm inequality established in this work, we give a sample-optimal algorithm for predicting bounded-degree observables.

The sample-optimal algorithm is equivalent to performing classical shadow tomography based on randomized Pauli measurements \cite{huang2020predicting, elben2022randomized}, and is essentially the ML algorithm given in Section~\ref{sec:ML-algo} with a fixed input state.
Consider an unknown $n$-qubit state $\rho$.
After performing $N$ randomized Pauli measurements on $N$ copies of $\rho$, we have a classical dataset denoted as
\begin{equation}
    S_N(\rho) \triangleq \left\{ \ket{\psi^{\mathrm{(out)}}_\ell} = \bigotimes_{i=1}^n \ket{s^{\mathrm{(out)}}_{\ell, i}} \right\}_{\ell = 1}^N,
\end{equation}
where $\ket{s^{\mathrm{(out)}}_{\ell, i}} \in \mathrm{stab}_1$ is a single-qubit stabilizer state.
Given an observable $O$, the algorithm predicts
\begin{equation}
    h(O) = \frac{1}{N} \sum_{\ell = 1}^N \Tr\left( O \bigotimes_{i=1}^n\left(3 \ketbra{s^{\mathrm{(out)}}_{\ell, i}}{s^{\mathrm{(out)}}_{\ell, i}} - I \right) \right).
\end{equation}
It is not hard to see that computing $h(O)$ only requires $\mathcal{O}(n N)$ classical computation time.
Hence, as we show later that $N = \mathcal{O}(\log(n) / \epsilon^2)$, the learning algorithm is very efficient.
Using the norm inequality for bounded-degree local Hamiltonian $ \norm{H}_{\mathrm{Pauli}, 1} \leq C \norm{H}$ for a constant $C$ in Corollary~\ref{cor:norm-ineq-bd-klocal-main}, and the classical shadow formalism \cite{huang2020predicting, elben2022randomized},
we obtain the following performance guarantee.

\vspace{0.8em}
\begin{theorem}[Sample complexity upper bound]
    Given an unknown $n$-qubit state $\rho$ and any $n$-qubit observables $O_1, \ldots, O_M$ with $\norm{O_i} \leq B_\infty$. Suppose each observable $O_i$ is a sum of few-body observables, where each qubit is acted on by $\mathcal{O}(1)$ of the few-body observables.
    Using a classical dataset $S_N(\rho)$ of size
    \begin{equation}
        N = \mathcal{O}\left( \frac{\log\big( \min( M, n) \big) B_\infty^2}{\epsilon^2} \right),
    \end{equation}
    we have $|h(O_i) - \Tr(O_i \rho)| \leq \epsilon, \forall i \in \{1, \ldots, M\}$ with high probability. The constant factor in the $\mathcal{O}(\cdot)$ notation above scales polynomially in the degree and exponentially in the locality of the observables.
\end{theorem}

The following theorem shows that the above algorithm achieves the optimal sample complexity among any algorithms that can perform collective measurement on many copies of $\rho$.

\begin{theorem}[Sample complexity lower bound]
    Consider the following task.
    There is an unknown $n$-qubit state $\rho$, and we are given $M$ observables $O_1,\ldots,O_M$ with $\max_{i} \norm{O_i} \leq B_\infty$.
    Each observable $O_i$ is a sum of few-body observables, where every qubit is acted on by $\mathcal{O}(1)$ of the few-body observables.
    We would like to estimate $\Tr(O_i\rho)$ to $\epsilon$ error for all $i\in[M]$ with high probability by performing arbitrary collective measurements on $N$ copies of $\rho$.
    The number of copies $N$ must be at least
    \begin{equation}
        N = \Omega\left( \frac{\log\big( \min( M, n) \big) B_\infty^2}{\epsilon^2} \right)
    \end{equation}
    for any algorithm to succeed in this task.
\end{theorem}

\noindent The detailed proofs of the sample complexity stated in the above theorems are given in Appendix~\ref{sec:samp-opt-bd-local-H}.

{\renewcommand\addcontentsline[3]{} \subsection{Efficient algorithms for learning an unknown observable from $\log(n)$ samples}}

As a second learning application of the norm inequalities, we consider the task of learning an unknown $n$-qubit observable $O^{\mathrm{(unk)}} = \sum_{P \in \{I, X, Y, Z\}^{\otimes n}} \alpha_P P$.
We can think of this unknown observable as $\mathcal{E}^\dagger(O)$, i.e., the observable $O$ after Heisenberg evolution under the unknown process $\mathcal{E}$.
Suppose we are given a training dataset of $\{ \rho_\ell, \Tr\left(O^{\mathrm{(unk)}} \rho_\ell\right) \}_{\ell = 1}^N$, where $\rho_\ell$ is sampled from an arbitrary distribution $\mathcal{D}$ over $n$-qubit states that is invariant under single-qubit Clifford gates.
Given an integer $k > 0$, we define the weight-$k$ truncation of $O^{\mathrm{(unk)}}$ to be the following Hermitian operator
\begin{equation}
    O^{(\mathrm{unk}, k)} \triangleq \sum_{P \in \{I, X, Y, Z\}^{\otimes n}: |P| \leq k} \alpha_P P,
\end{equation}
where $|P|$ is the number of qubits upon which $P$ acts nontrivially.
For a small $k$, we can think of $O^{(\mathrm{unk}, k)}$ as a low-weight approximation of the unknown observable $O^{(\mathrm{unk})}$.
By definition, $O^{(\mathrm{unk}, k)}$ is a $k$-local Hamiltonian, hence the norm inequality in Corollary~\ref{cor:norm-ineq-klocal-main} shows that
\begin{equation} \label{eq:norm-ineq-bound}
    \frac{1}{3} C(k) \norm{O^{(\mathrm{unk}, k)} }_{\mathrm{Pauli}, \frac{2k}{k + 1}} = \frac{1}{3} C(k) \left(\sum_{P \in \{I, X, Y, Z\}^{\otimes n}: |P|\le k} |\alpha_{P}|^r \right)^{1/r} \leq \norm{O^{(\mathrm{unk}, k)}},
\end{equation}
where $r = 2 k / (k + 1) \in [1, 2)$.
An $\ell_r$ norm bound ($r < 2$) on the Pauli coefficients implies that we can remove most of the small Pauli coefficients without incurring too much change under the $\ell_2$ norm.
As an example, consider an $M$-dimensional vector $x$ with $\norm{x}_r \leq 1$. Given $\wt{\epsilon} > 0$,
let $\wt{x}$ be the $M$-dimensional vector with
$\wt{x}_i = x_i$ if $|x_i| > \wt{\epsilon}$ and $\wt{x}_i = 0$ if $|x_i| \leq \wt{\epsilon}$.
We have
\begin{equation}
    \norm{x - \wt{x}}_2^2 = \sum_{i: |x_i| \leq \wt{\epsilon}} \left| x_i \right|^2 \leq \wt{\epsilon}^{\,2-r} \sum_{i: |x_i| \leq \wt{\epsilon}} \left| x_i \right|^r \leq \wt{\epsilon}^{\,2-r} \sum_{i} \left| x_i \right|^r \leq \wt{\epsilon}^{\,2-r}.
\end{equation}
In Appendix~\ref{sec:low-deg-MSE}, we show that the average error (both the mean squared error and the mean absolute error) is characterized by the $\ell_2$ norm.
Hence, Eq.~\eqref{eq:norm-ineq-bound} implies that we can set most of the Pauli coefficients in $O^{(\mathrm{unk}, k)}$ to zero without incurring too much error on average.

Using the above reasoning, learning the low-weight truncation $O^{(\mathrm{unk}, k)}$ amounts to learning the large Pauli coefficients of $O^{(\mathrm{unk}, k)}$ and setting all small Pauli coefficients to zero.
This ensures that the learning can be done very efficiently.
This approach is presented in Appendix~\ref{sec:Pauli-tools} with the main result stated in Lemma~\ref{lem:filtering-full}.
It is inspired by the learning algorithm of \cite{eskenazis2022learning} that achieves a logarithmic sample complexity for learning classical low-degree functions.

The last step in the proof is to argue that the low-weight truncation $O^{(\mathrm{unk}, k)}$ is a good surrogate for the unknown observable $O^{(\mathrm{unk})}$ when the goal is to predict $\Tr(O^{(\mathrm{unk})} \rho)$. The key insight here is that for distributions $\mathcal{D}$ that are invariant under single-Clifford gates, the contribution of any Pauli term $P$ in $O^{(\mathrm{unk})}$ to $\E_{\rho\sim \mathcal{D}}[\Tr(O^{(\mathrm{unk})}\rho)^2]$ is \emph{exponentially decaying} in the 
the weight $|P|$. This allows us to prove that $\E_{\rho\sim \mathcal{D}}[\Tr((O^{(\mathrm{unk})} - O^{(\mathrm{unk},k)})\rho)^2]$ is small.

Putting these ingredients together, we arrive at the following theorem.
As stated in the theorem, the learning algorithm is computationally efficient.

\vspace{0.8em}
\begin{theorem}[Learning an unknown observable]
    Given $\epsilon, \epsilon', \delta > 0$.
    Let $k = \lceil \log_{1.5}(1 / \epsilon)\rceil$ and $r = \tfrac{2k}{k+1} \in [1, 2)$.
    From training data $\{ \rho_\ell, \Tr\left(O^{\mathrm{(unk)}} \rho_\ell\right) \}_{\ell = 1}^N$ of size
    \begin{equation}
        N = \log(n / \delta) \, \min\left( 2^{\mathcal{O}\left(\log(\frac{1}{\epsilon}) \left( \log\log(\frac{1}{\epsilon}) +  \log(\frac{1}{\epsilon'})\right)\right)}, 2^{\mathcal{O}(\log(\frac{1}{\epsilon}) \log(n))} \right),
    \end{equation}
    where $\rho_\ell$ is sampled from $\mathcal{D}$, we can learn a function $h(\rho)$ such that
    \begin{equation}
        \E_{\rho \sim \mathcal{D}} \left| h(\rho) - \Tr(O^{\mathrm{(unk)}} \rho) \right|^2 \leq \left(\epsilon + \epsilon'\right) \norm{O^{(\mathrm{unk})}}^2 + \epsilon' \norm{O^{(\mathrm{unk}, k)}}^r \norm{O^{(\mathrm{unk})}}^{2-r}
    \end{equation}
    with probability at least $1 - \delta$.
    The training and prediction time of $h(\rho)$ are $\mathcal{O}(N n^k)$.
\end{theorem}

The factor of $\norm{O^{(\mathrm{unk})}}^2$ in the prediction error is the natural scale of the squared error.
From the theorem, we can see that we only need $\mathcal{O}(\log(n))$ samples to obtain a constant prediction error relative to $\norm{O^{(\mathrm{unk})}}^2 + \norm{O^{(\mathrm{unk}, k)}}^{r} \norm{O^{\mathrm{(unk)}}}^{2-r}$.
The proof the the theorem and the detailed description of the ML algorithm are given in Appendix~\ref{sec:learn-unk-obs}.

{\renewcommand\addcontentsline[3]{} \subsection{Learning an unknown quantum process}}

The ML algorithm for learning an unknown $n$-qubit quantum process $\mathcal{E}$ is essentially the combination of the two learning applications described above with a few modifications.
At a high level, we consider the following.
There is an $n$-qubit state $\rho$ sampled from an unknown distribution $\mathcal{D}$, as well as an observable $O$ that can be written as a sum of few-body observables, where each qubit is acted on by a constant number of the few-body observables.
In the first stage, we use the sample-optimal algorithm for predicting the bounded-degree observable $O$, where $\mathcal{E}(\rho_\ell)$ is an unknown quantum state, thus transforming the classical dataset $S_N(\mathcal{E})$ in Eq.~\eqref{eq:classicalSn-E} into a dataset,
\begin{equation}
    \left\{ \rho_\ell \triangleq \ketbra{\psi^{\mathrm{(in)}}_\ell}{\psi^{\mathrm{(in)}}_\ell}, \,\, \Tr\left(O \mathcal{E}\left( \rho_\ell \right)\right) \right\}_{\ell = 1}^N
    \label{eq:second-stage-training}
\end{equation}
that maps quantum states to real numbers. In the second stage, we apply the efficient algorithm for learning an unknown observable $O^{\mathrm{(unk)}} = \mathcal{E}^\dagger(O)$, regarding  Eq.~\eqref{eq:second-stage-training} as the training data for this task, thus predicting $\Tr\left(\mathcal{E}^\dagger(O) \rho\right)=\Tr\left(O \mathcal{E}\left( \rho \right)\right)$ for the state $\rho$ drawn from the distribution $\mathcal{D}$.
Because both stages of the algorithm run in time polynomial in $n$, the overall runtime for this procedure is polynomial in $n$.

In our actual proofs, there are a few deviations from the above high-level design, stemming from the fact that the input states $\rho_\ell$ are tensor products of random single-qubit stabilizer states.
This specific setting allows a few simplifications to be made.
With the simplifications, we can remove an additive factor of $\epsilon'$ in the prediction error.
Furthermore, a surprising fact is that learning from random product states is sufficient to predict highly-entangled states sampled from any distribution $\mathcal{D}$ invariant under single-qubit Clifford unitaries.
This surprising fact is a result of the characterization of the prediction error given in Lemma~\ref{lem:mse-D} based on a modified purity on subsystems of an input quantum state $\rho \sim \mathcal{D}$.

By combining the five parts together, we can establish Theorem~\ref{thm:main-res}, the precise sample complexity scaling in Eq.~\eqref{eq:N-samp-final}, and the prediction error bound in Eq.~\eqref{eq:pred-error-final}.
The full proof is given in Appendix~\ref{sec:learning-qevo}.

\begin{figure*}[t]
\centering
\includegraphics[width=0.97\textwidth]{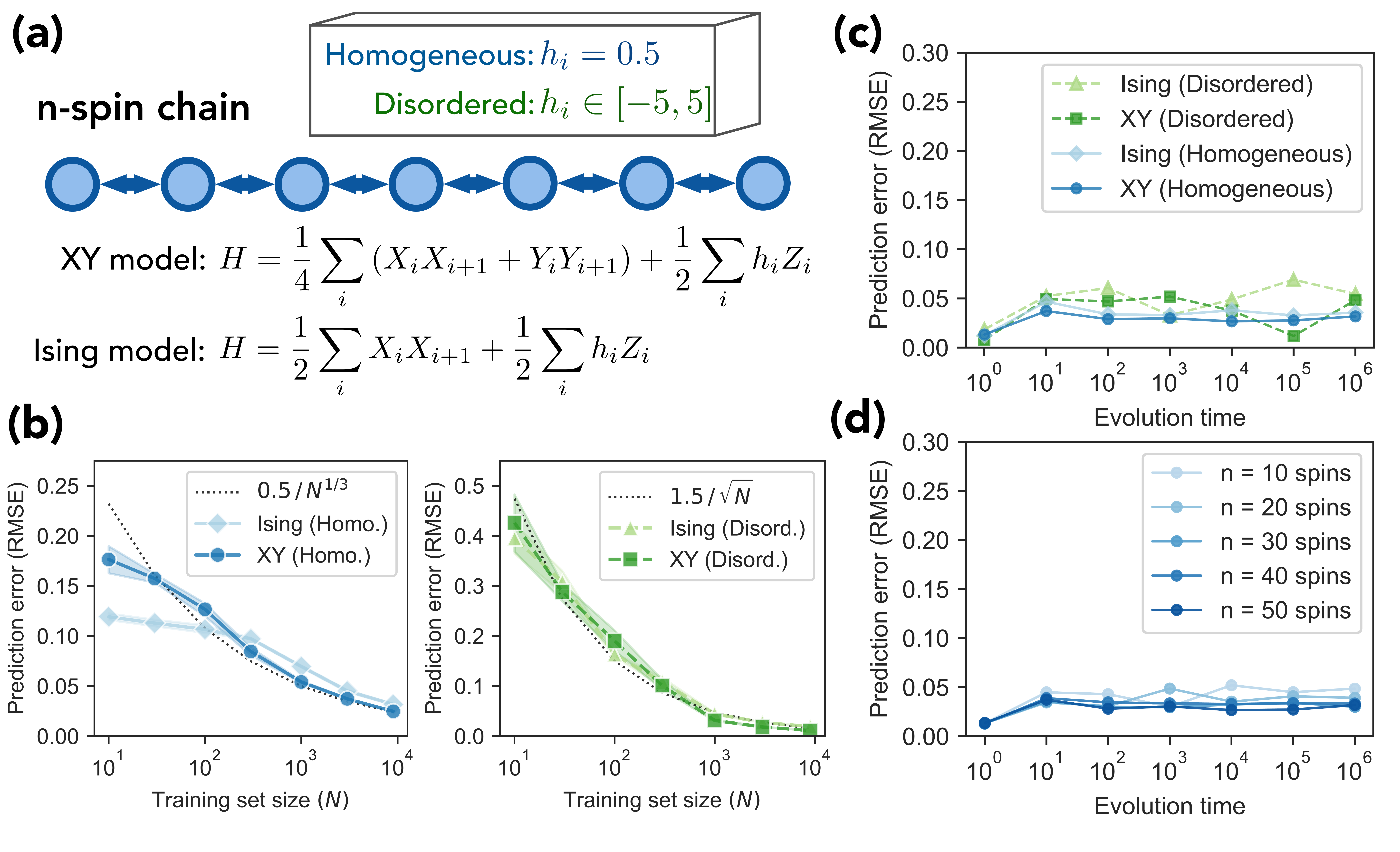}
    \caption{
    \emph{Prediction performance of ML models for learning $\mathcal{E}(\rho) = e^{-i t H} \rho e^{i t H}$ for a large time $t$.}
    \textsc{(a) Hamiltonians.} We consider XY/Ising model with a homogeneous/disordered $Z$ field on an $n$-spin open chain.
    \textsc{(b) Error scaling with training set size ($N$).} We show the root-mean-square error (RMSE) for predicting the Pauli-Z operator $Z_i$ on the output state $\mathcal{E}(\rho)$ for random product states~$\rho$.
    \textsc{(c, d) Error scaling with evolution time ($t$) and system size ($n$).} (d) shows the RMSE for the XY model with a homogeneous $Z$ field. The prediction error remains similar as we exponentially increase $t$ and Hilbert space dimension $2^n$.
    \label{fig:MLstats}}
\end{figure*}

\vspace{2em}
{\renewcommand\addcontentsline[3]{} \section{Numerical experiments}}

We have conducted numerical experiments to assess the performance of ML models in learning the dynamics of several physical systems.
The results corroborate our theoretical claims that long-time evolution over a many-body system can be learned efficiently. While our theorem only guarantees good performance for randomly sampled input states, we also find that the ML models work very well for structured input states that could be of practical interest. The source code can be found on a public GitHub repository\footnote{\url{https://github.com/hsinyuan-huang/learning-quantum-process}}.

We focus on training ML models to predict output state properties after the time dynamics of 1D $n$-spin XY/Ising chains with homogeneous/disordered $Z$ fields.
Let $H$ be the many-body Hamiltonian. The quantum processes $\mathcal{E}$ is given by $\mathcal{E}(\rho) = e^{-i t H} \rho e^{i t H}$ for a significantly long evolution time $t = 10^6$.
We consider the ML models described by Eq.~\eqref{eq:h-rhoO}.
While we utilize the very simple sparsity-enforcing strategy of setting small values to zero to prove Theorem~\ref{thm:main-res}, the standard sparsity-enforcing approach is through $\ell_1$ regularization \cite{tibshirani1996regression}.
A detailed description of applying $\ell_1$ regularization to enforce sparsity in $\alpha_{P}(O)$ is given in Appendix~\ref{sec:num-details}.
We find the best hyperparameters using four-fold cross-validation to minimize root-mean-square error (RMSE) and report the predictions on a test set.

\begin{figure*}[t]
\centering
\includegraphics[width=1.0\textwidth]{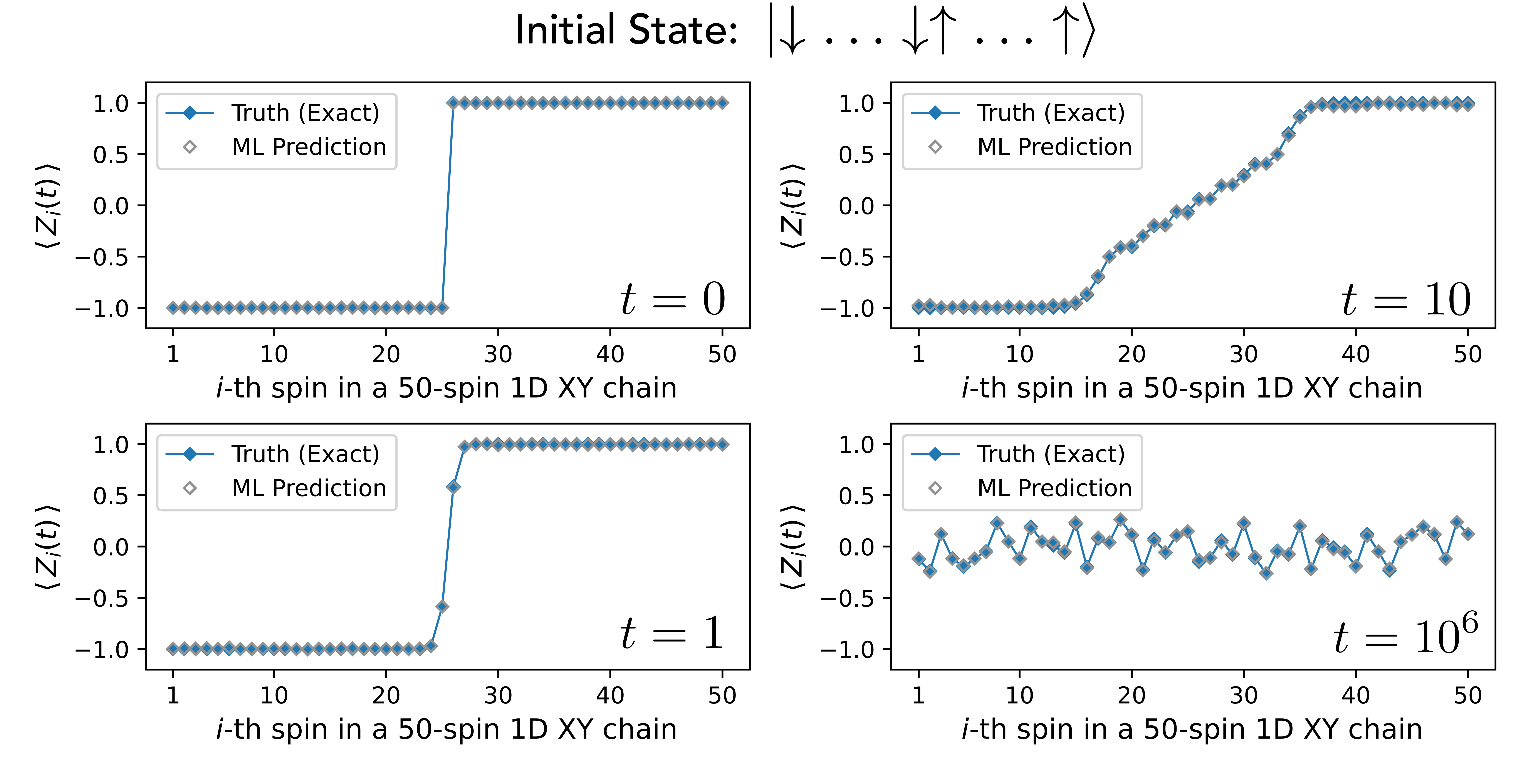}
    \caption{
    \emph{Visualization of ML model's prediction for an initial state $\rho = \ketbra{\psi}{\psi}$ with a domain wall.}
    We consider the 1D 50-spin XY chain with a homogeneous $Z$ field.
    We show the expectation value of $Z_i(t) = e^{i t H} Z_i e^{-i t H}$ for all the $50$ spins on the initial state $\ket{\psi} = \ket{\downarrow \ldots \downarrow \uparrow \ldots \uparrow}$. The ML model is trained on $10000$ random product states.
    We see that the ML model performs accurately for a significantly large range of time $t$.
    \label{fig:domainwall}}
\end{figure*}

Fig.~\ref{fig:MLstats} considers the performance for predicting the expectation of the Pauli-Z operator $Z_i$ on the output state for randomly sampled product input states not in the training data.
Fig.~\ref{fig:MLstats}(a) illustrates the many-body Hamiltonian $H$.
Fig.~\ref{fig:MLstats}(b) shows the dependence of the error on training set size $N$.
We can clearly see that as training set size $N$ increases, the prediction error notably decreases. This observation confirms our theoretical claim that long-time quantum dynamics could be efficiently learned.
In Fig.~\ref{fig:MLstats}(c), we consider how evolution time $t$ affects prediction performance. From the figure, we can see that even when we exponentially increase $t$, the prediction performance remains similar.
This matches with our theorem stating that no matter what the quantum process $\mathcal{E}$ is, even if $\mathcal{E}$ is an exponentially long-time dynamics, the ML model can still predict accurately and efficiently.
In Fig.~\ref{fig:MLstats}(d), we consider the dependence on system size $n$.
As $n$ increases linearly, the Hilbert space dimension $2^n$ grows exponentially.
Despite the exponential growth, even for $50$-spin systems, the ML model still predicts well. This matches with the logarithmic scaling on $n$ given in Theorem~\ref{thm:main-res}.

In Fig.~\ref{fig:domainwall}, we consider predicting properties of the final state after long-time dynamics for a highly structured input product state:
\begin{equation}
\ket{\psi} = \ket{\downarrow \ldots \downarrow \uparrow \ldots \uparrow},
\end{equation}
which has a single domain wall in the middle.
We focus on predicting the expected value for $Z_i(t) = e^{i t H} Z_i e^{-i t H}$ on every spin in the 1D 50-spin XY chain with a homogeneous $Z$ field $h_i = 0.5$ and consider evolution time $t$ from $0$ to $10^6$.
We train the ML model using $N = 10000$ random input product states.
We can see that the ML model predicts very well for this highly structured product state.
The collapse of the domain wall is accurately predicted by the ML model despite only seeing outcomes from random unstructured product states.
This numerical experiment suggests that the performance of the ML model goes beyond Theorem~\ref{thm:main-res}, which only guarantees accurate prediction on average.

Theorem~\ref{thm:main-res} states that the ML model can predict well on highly-entangled input states after learning only from random product state inputs.
We test this claim in Fig.~\ref{fig:GHZ} by considering an entangled input state
\begin{equation}
    \ket{\psi_e} = \sum_{\substack{s \in \{\leftarrow, \rightarrow\}^{n/2} \\ \mathrm{w/}\,\mathrm{even}\,\#\,\mathrm{of}\,\rightarrow}} \frac{1}{\sqrt{2^{(n/2) - 1}}} \ket{s} \otimes \ket{\rightarrow \downarrow \leftarrow \uparrow \rightarrow \downarrow \leftarrow \uparrow \ldots}.
\end{equation}
The left $n/2$ spins of the state $\ket{\psi_e}$ exhibit GHZ-like entanglement, which requires a linear-depth 1D quantum circuit to prepare.
The right $n/2$ spins of $\ket{\psi_e}$ form a product state with spins rotating clockwise from left to right.
Combining the left and right spins together, the state $\ket{\psi_e}$ cannot be generated by a short-depth 1D quantum circuit.
We can see that for this entangled input state, the ML model trained on random product states still predicts very well across a broad range of the evolution time $t$.

\begin{figure*}[t]
\centering
\includegraphics[width=1.0\textwidth]{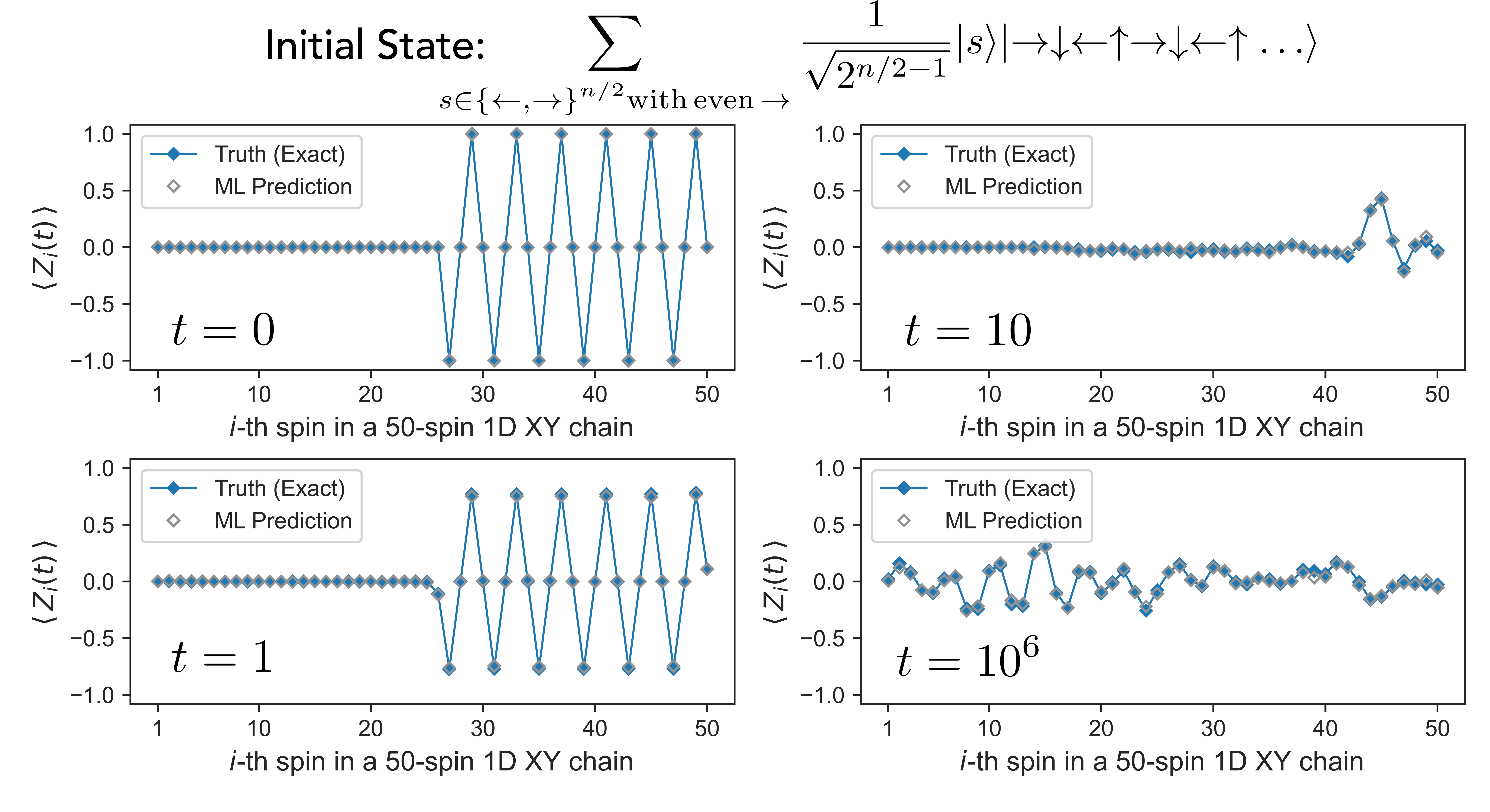}
    \caption{
    \emph{Visualization of ML model's prediction for a highly-entangled initial state $\rho = \ketbra{\psi}{\psi}$.}
    We consider the expected value of $Z_i(t) = e^{i t H} Z_i e^{-i t H}$, where $H$ corresponds to the 1D 50-spin XY chain with a homogeneous $Z$ field.
    The initial state $\ket{\psi}$ has a GHZ-like entanglement over the left-half chain and is a product state with spins rotating clockwise over the right-half chain.
    To prepare $\ket{\psi}$ with 1D circuits, a depth of at least $\Omega(n)$ is required.
    Even though the ML model is trained only on random product states (a total of $N = 10000$), it still performs accurately in predicting the highly-entangled state over a wide range of evolution time $t$.
    \label{fig:GHZ}}
\end{figure*}

\vspace{2em}
{\renewcommand\addcontentsline[3]{} \section{Outlook}}

The theorem established in this work shows that learning to predict a complex quantum process can be achieved with computationally-efficient ML algorithms.
Once we have obtained training data by accessing the unknown process $\mathcal{E}$ sufficiently many times, the proposed ML algorithm is entirely classical except for the step of obtaining the reduced density matrix (RDM) of the input state $\rho$, which may require quantum computation.
This algorithm is reminiscent of recent proposals for quantum ML based on kernel methods \cite{havlivcek2019supervised, schuld2019quantum, huang2020power}, in particular the projected quantum kernel \cite{huang2020power}.
This result highlights the potential for using hybrid quantum-classical ML algorithms to learn to model exotic quantum dynamics occurring in nature.

The results presented in this work also have implications in several previously-studied problems.
Prior works \cite{cirstoiu2020variational, gibbs2022dynamical, caro2022out} have proposed to train quantum ML algorithms on a given quantum process with the hope that the learned model can be faster than the process itself.
Our proof that even an exponential-time quantum dynamics can be predicted in quasi-polynomial time provides rigorous support for such a hope.
Furthermore, the proposed ML algorithm can be efficiently run on a classical computer when the few-body RDMs of the input state $\rho$ is easy to compute classically.
Hence, this result provides a rigorous foundation for empirical works using classical ML to learn and simulate quantum dynamics \cite{torlai2020quantum, gentile2021learning, banchi2018modelling, gilmer2017neural}.
When $\mathcal{E}$ is a parameterized quantum circuit $U_\theta$, such as a quantum neural network \cite{mcclean2018barren, caro2022generalization, havlivcek2019supervised, farhi2018classification, huang2020power, huang2021information},
the existence of a classical ML model that can efficiently predict the output of $U_\theta$ implies that the function $\Tr(O U_\theta \rho U_\theta^\dagger)$ is easy to represent and learn on a classical computer.
This finding shows that quantum circuits do not have strong representational power for a wide range of distributions over quantum state input $\rho$ with easy-to-compute RDMs.

Several open problems remain to be answered.
While we only focus on locally flat distributions $\mathcal{D}$, we believe that efficient ML algorithms also exist for other general classes of distributions.
An important open problem is hence the following: Can we obtain computationally efficient learning algorithms for any ``smooth'' distribution over quantum state space?
If not, how general can the class of distributions be?
Similar questions can be asked about the class of observables that we predict.
For what general classes of observables $O$ can one predict efficiently, in terms of both sample size and computation time?
This problem is closely related to the problem of when shadow tomography \cite{aaronson2018shadow, aaronson2019gentle, aaronson22shadow} can be made computationally efficient.
Other important questions include: If we restrict the quantum process $\mathcal{E}$ to be generated in polynomial time, can we obtain improved efficiency? What efficiency guarantees apply to fermionic or bosonic systems?
A better understanding of these problems would illuminate the ultimate power of classical and quantum ML algorithms for learning about physical dynamics.

\vspace{2em}
{\renewcommand\addcontentsline[3]{} \acknowledgements}
{ The authors thank Victor V. Albert, Chi-Fang (Anthony) Chen, Bryan Clark, Richard Kueng, and Spiros Michalakis for valuable input and inspiring discussions.
After learning about our proof of the quantum Bohnenblust-Hille inequality, Alexander Volberg and Haonan Zhang found a very different proof with a better Bohnenblust-Hille constant. We thank them for sharing their results with us.
HH is supported by a Google PhD fellowship. SC is supported by NSF Award 2103300. JP acknowledges funding from  the U.S. Department of Energy Office of Science, Office of Advanced Scientific Computing Research, (DE-NA0003525, DE-SC0020290), and the National Science Foundation (PHY-1733907). The Institute for Quantum Information and Matter is an NSF Physics Frontiers Center. }

\bibliographystyle{unsrt}
{\renewcommand\addcontentsline[3]{} \bibliography{references}}

\newpage
\vspace{4em}
\appendix

\renewcommand{\appendixname}{APPENDIX}
\renewcommand{\thesubsection}{\MakeUppercase{\alph{section}}.\arabic{subsection}}
\makeatletter
\renewcommand{\p@subsection}{}
\makeatother

\noindent
\textbf{ \LARGE{}Appendices }
\vspace{-0.5em}

\tableofcontents

\vspace{1em}

\section{Optimizing $k$-local Hamiltonian with random product states}
\label{sec:optimize-klocal}

While our goal is to design a good machine learning (ML) algorithm with low sample complexity, this section is a detour to a different task on the optimization of a $k$-local Hamiltonian.
We present an improved approximation algorithm for optimizing any $k$-local Hamiltonian.
The central result in this detour will become useful for showing the low sample complexity of several ML algorithms.

\subsection{Task description and main theorem}

\begin{task}[Optimizing quantum Hamiltonian]
    Given $n, k \geq 1$ and an $n$-qubit $k$-local Hamiltonian
    \begin{equation}
    H = \sum_{P \in \{I, X, Y, Z\}^{\otimes n}: |P|\leq k} \alpha_P P,
    \end{equation}
    where $|P|$ is the number of non-identity components in $P$.
    Find a state $\ket{\psi}$ that maximizes/minimizes $\bra{\psi} H \ket{\psi}$.
\end{task}

\noindent The task given above is related to solving ground states \cite{kempe2006complexity, sakurai_napolitano_2017} when we consider minimizing $\bra{\psi} H \ket{\psi}$ and quantum optimization \cite{farhi2014quantum,Farhi2014AQA,harrow2017extremal, parekh2020beating, Hallgren2020AnAA, anshu2021improved, hastings2022optimizing} when we consider maximizing $\bra{\psi} H \ket{\psi}$.
The maximization and minimization are often the same problem since maximizing $\bra{\psi} H \ket{\psi}$ is the same as minimizing $\bra{\psi} (-H) \ket{\psi}$.
Without further constraints, even for $k = 2$, finding the optimal state $\ket{\psi^*}$ maximizing $\bra{\psi} H \ket{\psi}$ is known to be QMA-hard \cite{piddock2015complexity},
hence it is expected to have no polynomial-time algorithm even on a quantum computer.
Most existing works consider deterministic or randomized constructions of $\ket{\psi}$ with rigorous upper/lower bound guarantees on $\bra{\psi} H \ket{\psi}$ for minimization/maximization.
Some of these lower bounds \cite{parekh2020beating, Hallgren2020AnAA, hastings2022optimizing} are based on the optimal value $\mathrm{OPT} = \sup_{\ket{\psi}} \bra{\psi} H \ket{\psi}$, while some \cite{Farhi2014AQA,harrow2017extremal, anshu2021improved} are based on the Pauli coefficients $\alpha_P$.

\subsubsection{Definition of expansion}

In this section, we present a random product state construction for the optimization problem, where the rigorous upper/lower bound is based on the Pauli coefficients $\alpha_P$ and the expansion property defined below.
The expansion property is defined for any Hamiltonian $H$.

\begin{definition}[Expansion property]
    Given an $n$-qubit Hamiltonian $H = \sum_{P} \alpha_P P$. We say $H$ has an expansion coefficient $c_e$ and expansion dimension $d_e$ if for any $\Upsilon \subseteq \{1, \ldots, n\}$ with $|\Upsilon| = d_e$,
    \begin{equation}
        \sum_{P \in \{I, X, Y, Z\}^{\otimes n}} \indicator\Big[\alpha_P \neq 0 \,\, \mathrm{and} \,\,  \big(\Upsilon \subseteq \mathsf{dom}(P) \,\, \mathrm{or} \,\, \mathsf{dom}(P) \subseteq \Upsilon \big) \, \Big] \leq c_e,
    \end{equation}
    where $\mathsf{dom}(P)$ is the set of qubits that $P$ acts nontrivially on.
\end{definition}

\noindent The expansion property captures the connectivity of the Hamiltonian.
We give two examples, general $k$-local Hamiltonian and geometrically-local Hamiltonian, to provide more intuition on the expansion property.

\begin{fact}[Expansion property for general $k$-local Hamiltonian] \label{fact:exp-klocal}
    Any Hamiltonian given by a sum of $k$-qubit observables has expansion coefficient $4^{k}$ and expansion dimension $k$.
\end{fact}
\begin{proof}
    Let $H = \sum_P \alpha_P P$. All the Pauli observables $P$ with nonzero $\alpha_P$ act at most on $k$ qubits.
    For any $\Upsilon$ with $|\Upsilon| = k$, all the Pauli observables with nonzero $\alpha_P$ must have a domain contained in $\Upsilon$.
    There are at most $4^k$ such Pauli observables.
    Hence, the claim follows.
\end{proof}

\begin{fact}[Expansion property for bounded-degree $k$-local Hamiltonian] \label{fact:exp-bounded}
    Any Hamiltonian given by a sum of $k$-qubit observables $H = \sum_{j} h_j$, where each qubit is acted on by at most $d$ of the $k$-qubit observables $h_j$, has expansion coefficient $c_e = 4^{k} d$ and expansion dimension $d_e = 1$.
\end{fact}
\begin{proof}
    For every $\Upsilon$ with $|\Upsilon|$, $\Upsilon = \{i\}$ for some qubit $i$.
    For each qubit $i$ (corresponding to $\Upsilon = \{i\}$), we have at most $d$ $k$-qubit observables acting on $i$.
    Each of the $k$-qubit observables can be expanded into at most $4^k$ Pauli terms.
    Hence we can set $c_e = 4^{k} d$ and $d_e = 1$.
\end{proof}

\begin{fact}[Expansion property for geometrically-local Hamiltonian]
    Any Hamiltonian given by a sum of geometrically-local observables has expansion coefficient $c_e = \mathcal{O}(1)$ and expansion dimension $1$.
\end{fact}
\begin{proof}
    For a geometrically-local Hamiltonian $H = \sum_{P} \alpha_P P$, each qubit $i$ is acted by at most a constant number $c_i = \mathcal{O}(1)$ of $P$ with non-zero $\alpha_P$. Hence for any qubit $i$, $\sum_{P \in \{I, X, Y, Z\}^{\otimes n}} \indicator[\alpha_P \neq 0 \,\, \mathrm{and} \,\,  (\Upsilon \subseteq \mathsf{dom}(P) \,\, \mathrm{or} \,\, \mathsf{dom}(P) \subseteq \Upsilon ) \, ] = c_i.$
    Thus, we can set $d_e = 1$ and $c_e = \max_i c_i = \mathcal{O}(1)$.
\end{proof}

\subsubsection{Main theorem}

With the expansion property defined, we can state the rigorous guarantee on the performance of the proposed randomized approximation algorithm on optimizing an $n$-qubit $k$-local Hamiltonian $H$.
We compare with the average energy $\E_{\ket{\phi}: \mathrm{Haar}} \big[ \bra{\phi} H \ket{\phi} \big] = \alpha_I$ over Haar random state.
The randomized approximation algorithm uses an optimization over a single-variable polynomial that guarantees improvement in at least one direction (minimization or maximization).

\begin{theorem}[Random product states for optimizing $k$-local Hamiltonian]\label{thm:main}
    Given an $n$-qubit $k$-local Hamiltonian $H = \sum_{P: |P| \leq k} \alpha_P P$ with expansion coefficient/dimension $c_e, d_e$.
    Let $r = 2 d_e / (d_e + 1) \in [1, 2)$ and $\mathsf{nnz}(H) \triangleq |\{P: \alpha_P \neq 0 \}|$.
    There is a randomized algorithm that runs in time $\mathcal{O}(nk + \mathsf{nnz}(H) 2^k)$ and produces either a random maximizing state $\ket{\psi} = \ket{\psi_1} \otimes \ldots \otimes \ket{\psi_n}$ satisfying
    \begin{equation} \label{eq:max-state-opt}
        \E_{\ket{\psi}} \big[ \bra{\psi} H \ket{\psi} \big] \geq \E_{\ket{\phi}: \mathrm{Haar}} \big[ \bra{\phi} H \ket{\phi} \big] + C(c_e, d_e, k) \left(\sum_{P \neq I} |\alpha_P|^r\right)^{1/r},
    \end{equation}
    or a random minimizing state $\ket{\psi} = \ket{\psi_1} \otimes \ldots \otimes \ket{\psi_n}$ satisfying
    \begin{equation} \label{eq:min-state-opt}
        \E_{\ket{\psi}} \big[ \bra{\psi} H \ket{\psi} \big] \leq \E_{\ket{\phi}: \mathrm{Haar}} \big[ \bra{\phi} H \ket{\phi} \big] - C(c_e, d_e, k) \left(\sum_{P \neq I} |\alpha_P|^r\right)^{1/r}.
    \end{equation}
    The constant $C(c_e, d_e, k)$ is given by
    \begin{equation}
        C(c_e, d_e, k) = \frac{\sqrt{2(k!)}}{c_e^{1/(2d_e)} k^{k+1.5+1/r} (\sqrt{6} + 2 \sqrt{3} )^k} = \Theta_k\left(\frac{1}{c_e^{1 / (2 d_e)}}\right),
    \end{equation}
    where $\Theta_k$ considers the asymptotic scaling when $k$ is a constant.
\end{theorem}

Some observations can be made. First, the improvement over Haar random states in Theorem~\ref{thm:main} becomes larger when the expansion coefficient $c_e$ is smaller.
Second, $(\sum_{P \neq I} |\alpha_P|^r)^{1/r}$ is the $\ell_r$-norm on the non-identity Pauli coefficients, so by monotonicity of $\ell_r$-norms, $(\sum_{P \neq I} |\alpha_P|^r)^{1/r}$ becomes smaller as $r$ becomes larger (corresponding to larger $d_e$).
Hence, the improvement is greater for smaller expansion dimension $d_e$.
In particular, it is helpful to contrast Eqs.~\eqref{eq:max-state-opt} and \eqref{eq:min-state-opt} with the following basic estimate corresponding to $r = 2$ which holds regardless of $c_e,d_e,k$:
\begin{equation}
    \sup_{\ket{\psi}} \Big|\bra{\psi} H \ket{\psi} - \E_{\ket{\phi}: \mathrm{Haar}}\big[ \bra{\phi} H \ket{\phi} \big]\Big| \ge \left(\sum_{P\neq I} |\alpha_P|^2\right)^{1/2}.
\end{equation}
This holds for any Hamiltonian $H = \sum_P \alpha_P P$ because $\left(\sum_{P\neq I} |\alpha_P|^2\right)^{1/2} = \frac{1}{2^{n/2}}\norm{H - \alpha_I I}_F \le \norm{H - \alpha_I I}_{\infty},$
where $\norm{\cdot}$ denotes spectral norm, and $\alpha_I = \E_{\ket{\psi}: \mathrm{Haar}}\big[ \bra{\phi} H \ket{\phi}\big]$.
This basic estimate shows that we can always find a state that improves by at least the $\ell_2$-norm of $\alpha_P$, although the optimization process can be computationally hard.

\subsubsection{An alternative version of the main theorem}\label{sec:alt-main-theorem}

By following the proof of Theorem~\ref{thm:main} and replacing the use of Corollary~\ref{cor:final-char} by Lemma~\ref{lem:key-char}, we can establish the following alternative theorem statement that does not utilize the expansion property.

\begin{theorem}[Random product states for optimizing $k$-local Hamiltonian; alternative]\label{thm:main-alt}
    Given an $n$-qubit $k$-local Hamiltonian $H = \sum_{P: |P| \leq k} \alpha_P P$ with $k = \mathcal{O}(1)$.
    Let $\mathsf{nnz}(H) \triangleq |\{P: \alpha_P \neq 0 \}|$.
    There is a randomized algorithm that runs in time $\mathcal{O}(nk + \mathsf{nnz}(H) 2^k)$ and produces a random state $\ket{\psi} = \ket{\psi_1} \otimes \ldots \otimes \ket{\psi_n}$ satisfying
    \begin{equation} \label{eq:maxmin-state-opt-alt}
        \left| \E_{\ket{\psi}} \big[ \bra{\psi} H \ket{\psi} \big] - \E_{\ket{\phi}: \mathrm{Haar}} \big[ \bra{\phi} H \ket{\phi} \big] \right| \geq D \sum_{i \in [n], p \in \{X, Y, Z\}} \sqrt{\sum_{\substack{P: P_i = p}} \alpha_P^2},
    \end{equation}
    for some constant $D$.
\end{theorem}

We can compare the above theorem with a closely related result in \cite{harrow2017extremal}.
The following is a restatement of the approximation guarantee from Theorem 2 and Lemma 3 in \cite{harrow2017extremal}, which is a corollary of a powerful result in Boolean function analysis \cite{Dinur2006OnTF, barak_et_al:LIPIcs:2015:5298} relating the maximum influence and the ability to sample a bitstring from the Boolean hypercube with a large magnitude in the function value.
We can define the influence of qubit $i$ under Pauli matrix $p \in \{X, Y, Z\}$ as $I(i, p) = \sum_{\substack{P: P_i = p}} \alpha_P^2$.

\begin{theorem}[Approximation guarantee from \cite{harrow2017extremal} for optimizing $k$-local Hamiltonian] \label{thm:main-harrow}
    Given an $n$-qubit $k$-local Hamiltonian $H = \sum_{P: |P| \leq k} \alpha_P P$ with $k = \mathcal{O}(1)$.
    There is a polynomial-time randomized algorithm that produces a random state $\ket{\psi} = \ket{\psi_1} \otimes \ldots \otimes \ket{\psi_n}$ satisfying
    \begin{equation}
        \left| \E_{\ket{\psi}} \big[ \bra{\psi} H \ket{\psi} \big] - \E_{\ket{\phi}: \mathrm{Haar}} \big[ \bra{\phi} H \ket{\phi} \big] \right| \geq D \sum_{i \in [n], p \in \{X, Y, Z\}} \frac{\sum_{\substack{P: P_i = p}} \alpha_P^2}{\max_{j, q} \sqrt{\sum_{\substack{P: P_j = q}} \alpha_P^2}},
    \end{equation}
    for some constant $D$.
\end{theorem}

The guarantee from \cite{harrow2017extremal} is asymptotically optimal when the influence $I(i, p)$ are of a similar magnitude for different qubit $i$ and Pauli matrix $p$.
However, the approximation guarantee can be far from optimal when there is a large variation in the influence $I(i, p)$ over different qubits $i, p$.
As an example, consider a 1D $n$-qubit nearest-neighbor chain, where $|\alpha_P| = 1$ for only a constant number of Pauli observables $P$ and $|\alpha_P| = 1 / \sqrt{n}$ for the rest of the Pauli observables.
The improvements over Haar random state by our algorithm and the algorithm in \cite{harrow2017extremal} are respectively given by
\begin{align}
    \Theta\left(\sum_{i \in [n], p \in \{X, Y, Z\}} \sqrt{\sum_{\substack{P: P_i = p}} \alpha_P^2}\right) &= \Theta\left(\sqrt{n}\right),\\
    \Theta\left( \sum_{i \in [n], p \in \{X, Y, Z\}} \frac{\sum_{\substack{P: P_i = p}} \alpha_P^2}{\max_{j, q} \sqrt{\sum_{\substack{P: P_j = q}} \alpha_P^2}} \right) &= \Theta\left( 1 \right).
\end{align}
Hence, when there is large variation in the influence, our guarantee improves over that of \cite{harrow2017extremal}.
For our machine learning applications, the removal of the dependence on the maximum influence is central.
By removing the ratio $\sqrt{I(i, p)} / \max_{j, q} \sqrt{I(j, q)}$, we can obtain the $\ell_r$ norm dependence for an $r < 2$ as given in Theorem~\ref{thm:main}.
We will later see that having the $\ell_r$ norm bound (for $r < 2$) allows a substantial reduction in the sample complexity in training machine learning models for predicting properties.

We do want to mention that the improvement comes at a cost of a slightly worse dependence on $k = \mathcal{O}(1)$.
In Theorem~\ref{thm:main-harrow} from \cite{harrow2017extremal} based on Boolean function analysis \cite{Dinur2006OnTF, barak_et_al:LIPIcs:2015:5298}, the dependence on $D$ is $1 / 2^{\Theta(k)}$.
However, our result in Theorem~\ref{thm:main-alt} is $D = 1 / 2^{\Theta(k \log k)}$.
This difference stems from the construction for the random state $\ket{\psi}$.
\cite{Dinur2006OnTF, barak_et_al:LIPIcs:2015:5298, harrow2017extremal} utilize a random restriction approach, where some random subset of variables are fixed with some random values and the rest of the variables are optimized.
On the other hand, we utilize a polarization approach, where we replicate each variable many times, randomly fix all except the last replica, optimize the last replica, and combine using a random-signed averaging.

\subsection{Corollaries of the main theorem}
\label{sec:cor-main-opt-thm}

Here, we consider how the main theorem applies to certain classes of $k$-local Hamiltonians and discuss the relations of the corollaries to related works.

\subsubsection{Optimizing arbitrary $k$-local Hamiltonians}

The first corollary considers a general $k$-local Hamiltonian $H = \sum_{P: |P| \leq k} \alpha_P P$. We can combine Fact~\ref{fact:exp-klocal} and the main theorem to obtain the following corollary.

\begin{corollary}[Optimizing arbitrary $k$-local Hamiltonian] \label{cor:opt-any}
    Given an $n$-qubit $k$-local Hamiltonian $H = \sum_{P: |P| \leq k} \alpha_P P$.
    There is a randomized algorithm that runs in time $\mathcal{O}(n^k)$ and produces a random product state $\ket{\psi} = \ket{\psi_1} \otimes \ldots \otimes \ket{\psi_n}$ with
    \begin{equation}
        \left| \E_{\ket{\psi}} \big[ \bra{\psi} H \ket{\psi} \big] - \E_{\ket{\phi}: \mathrm{Haar}} \big[ \bra{\phi} H \ket{\phi} \big] \right| \geq C(k) \left(\sum_{P \neq I} |\alpha_P|^{2k / (k+1)}\right)^{(k+1) / (2k)},
    \end{equation}
    where $C(k) = \frac{\sqrt{2(k!)}}{2 k^{k+1.5+ (k+1)/(2k)} (\sqrt{6} + 2 \sqrt{3} )^k}$.
\end{corollary}

For $k = 2$, we have $2k / (k + 1) = 4/3$ and the above result resembles Littlewood's $4/3$ inequality.
Recall that Littlewood's $4/3$ inequality states that given $\{\beta_{i, j} \in \mathbb{C}\}_{i, j}$,
\begin{equation}
    \sup\left\{ \left|\sum_{i, j} \beta_{i, j} x^{(1)}_i x^{(2)}_j \right| : x^{(k)}_i \in \mathbb{C}, \left|x^{(k)}_i\right| \leq 1, \forall i \in \mathbb{N}, k \in \{1, 2\} \right\} \geq \frac{1}{\sqrt{2}} \left(\sum_{i, j} |\beta_{i, j}|^{4/3}\right)^{3/4}.
\end{equation}
For $k > 2$, the above result resembles Bohnenblust-Hille inequality, which states that given $\{\beta_{i_1, \ldots, i_k} \in \mathbb{C}\}_{i_1, \ldots, i_k}$,
\begin{align}
    &\sup\left\{ \left|\sum_{i_1, \ldots, i_k} \beta_{i_1, \ldots, i_k} x^{(1)}_{i_1} \ldots x^{(k)}_{i_k} \right| : x^{(\kappa)}_{i_{\kappa}} \in \mathbb{C}, \left|x^{(\kappa)}_{i_{\kappa}} \right| \leq 1, \forall i_{\kappa} \in \mathbb{N}, \kappa \in [k] \right\}\\
    &\geq D_k \left(\sum_{i_1, \ldots, i_k} |\beta_{i_1, \ldots, i_k}|^{2k / (k+1)}\right)^{(k+1) / (2k)},
\end{align}
for some constant $D_k$ that depends on $k$.
For optimizing general $k$-local Hamiltonian, the design of the randomized approximation algorithm is inspired by the original proof \cite{BHineq1931} of Bohnenblust-Hille inequality from 1931, which is used to study the absolute convergence of Dirichlet series.

\subsubsection{Optimizing bounded-degree $k$-local Hamiltonians}

Here, we consider a Hamiltonian given by a sum of $k$-qubit observables, where each qubit is acted on by at most $d$ of the $k$-qubit observables. This is often referred to as a $k$-local Hamiltonian with a bounded degree $d$.
We can combine Fact~\ref{fact:exp-bounded} and the main theorem to obtain the following corollary.

\begin{corollary}[Optimizing bounded-degree $k$-local Hamiltonian] \label{cor:opt-bounded}
    Given an $n$-qubit $k$-local Hamiltonian $H = \sum_{P: |P| \leq k} \alpha_P P$ with bounded degree $d$, $|\alpha_P| \leq 1$ for all $P$, and $k = \mathcal{O}(1)$.
    There is a randomized algorithm that runs in time $\mathcal{O}(nd)$ and produces either a random maximizing state $\ket{\psi} = \ket{\psi_1} \otimes \ldots \otimes \ket{\psi_n}$ satisfying
    \begin{equation}
        \E_{\ket{\psi}} \big[ \bra{\psi} H \ket{\psi} \big] \geq \E_{\ket{\phi}: \mathrm{Haar}} \big[ \bra{\phi} H \ket{\phi} \big] + \frac{C}{\sqrt{d}} \sum_{P \neq I} |\alpha_P|,
    \end{equation}
    or a random minimizing state $\ket{\psi} = \ket{\psi_1} \otimes \ldots \otimes \ket{\psi_n}$ satisfying
    \begin{equation}
        \E_{\ket{\psi}} \big[ \bra{\psi} H \ket{\psi} \big] \leq \E_{\ket{\phi}: \mathrm{Haar}} \big[ \bra{\phi} H \ket{\phi} \big] - \frac{C}{\sqrt{d}} \sum_{P \neq I} |\alpha_P|
    \end{equation}
    for some constant $C$.
\end{corollary}

The task of optimizing bounded-degree $k$-local Hamiltonians has been considered in previous work \cite{anshu2021improved}.

\begin{theorem}[Approximation guarantee from \cite{anshu2021improved}] \label{thm:anshu-opt}
    Given an $n$-qubit $2$-local Hamiltonian $H = \sum_{P: |P| \leq 2} \alpha_P P$ with bounded degree $d$, and $|\alpha_P| \leq 1$ for all $P$.
    There is a polynomial-time randomized algorithm that produces a quantum circuit that generates a random maximizing state $\ket{\psi}$ satisfying
    \begin{equation}
        \E_{\ket{\psi}} \big[ \bra{\psi} H \ket{\psi} \big] \geq \E_{\ket{\phi}: \mathrm{Haar}} \big[ \bra{\phi} H \ket{\phi} \big] + \frac{C}{d} \left(\sum_{P \neq I} |\alpha_P|^2 \right) \cdot \frac{\sum_{P \neq I} |\alpha_P|^2}{\sum_{P \neq I} \indicator[\alpha_P \neq 0]},
    \end{equation}
    as well as a random minimizing state $\ket{\psi}$ satisfying
    \begin{equation}
        \E_{\ket{\psi}} \big[ \bra{\psi} H \ket{\psi} \big] \leq \E_{\ket{\phi}: \mathrm{Haar}} \big[ \bra{\phi} H \ket{\phi} \big] - \frac{C}{d} \left(\sum_{P \neq I} |\alpha_P|^2 \right) \cdot \frac{\sum_{P \neq I} |\alpha_P|^2}{\sum_{P \neq I} \indicator[\alpha_P \neq 0]}
    \end{equation}
    for some constant $C$.
\end{theorem}

The result from \cite{anshu2021improved} considers a single-step gradient descent using a shallow quantum circuit on an initial random product state.
Because $\sum_{P \neq I} |\alpha_P|^2 \leq \sum_{P \neq I} \indicator[\alpha_P \neq 0]$ and $\sum_{P \neq I} |\alpha_P| \geq \sum_{P \neq I} |\alpha_P|^2$,
our result in Corollary~\ref{cor:opt-bounded} improves either the maximization problem or the minimization problem over Theorem~\ref{thm:anshu-opt}.
For example, if we consider $\alpha_P = \Theta(1 / d)$, which sets the total interaction strength on each qubit to be $\Theta(1)$, then the improvement over Haar random state by our algorithm and the algorithm in \cite{anshu2021improved} is given by
\begin{equation}
    \Theta\left(\frac{1}{\sqrt{d}} \sum_{P \neq I} |\alpha_P|\right) = \Theta\left(\frac{n}{d^{1.5}}\right), \quad\quad \Theta\left(\frac{1}{\sqrt{d}} \left(\sum_{P \neq I} |\alpha_P|^2 \right) \cdot \frac{\sum_{P \neq I} |\alpha_P|^2}{\sum_{P \neq I} \indicator[\alpha_P \neq 0]} \right) = \Theta\left(\frac{n}{d^{4.5}}\right).
\end{equation}
We can see that our algorithm gives a larger improvement for the scaling with the degree $d$.
As another example, consider a 1D $n$-qubit nearest-neighbor chain (hence $d = 2$), where $|\alpha_P| = 1$ for only a constant number of Pauli observables $P$ and $|\alpha_P| = 1 / \sqrt{n}$ for the rest of the Pauli observables.
The improvement over Haar random state by our algorithm and the algorithm in \cite{anshu2021improved} is given by
\begin{equation}
    \Theta\left(\frac{1}{\sqrt{d}} \sum_{P \neq I} |\alpha_P|\right) = \Theta\left( \sqrt{n} \right), \quad\quad \Theta\left(\frac{1}{\sqrt{d}} \left(\sum_{P \neq I} |\alpha_P|^2 \right) \cdot \frac{\sum_{P \neq I} |\alpha_P|^2}{\sum_{P \neq I} \indicator[\alpha_P \neq 0]} \right) = \Theta\left(\frac{1}{n}\right).
\end{equation}
We can see that our algorithm gives a larger improvement for the scaling with the number $n$ of qubits.

\subsection{Description of the randomized approximation algorithm}
\label{sec:desc-algo}

There are a few steps in the proposed randomized algorithm.
The first step is to choose the best slice of the $k$-local Hamiltonian by splitting the $k$-local Hamiltonian $H = \sum_{P: |P| \leq k} \alpha_P P$ as follows,
\begin{equation} \label{eq:H-decomp}
    H = \alpha_I I + \sum_{\kappa = 1}^k H_\kappa, \,\,\quad H_\kappa \triangleq \sum_{P: |P| = \kappa} \alpha_P P.
\end{equation}
We choose $\kappa^* \in \{1, \ldots, k\}$ to be the $\kappa$ that maximizes $\sum_{P: |P| = \kappa} |\alpha_P|^r$, where $r = 2 d_e / (d_e + 1)$.
This step can be performed in time $\mathcal{O}( \mathsf{nnz}(H) k)$.

In the second step, the algorithm samples $(\kappa^* - 1) n$ Haar-random single-qubit pure states,
\begin{equation}
    \ket{\psi_{(s, j)}} \in \mathbb{C}^{2}, \,\,\quad \forall s \in \{1, \ldots, \kappa^* - 1\}, j \in \{1, \ldots, n\}.
\end{equation}
This step can be performed in time $\mathcal{O}(nk)$.

The third step is a local optimization on each qubit based on $\ket{\psi_{(s, j)}}$.
For each qubit $i$ and Pauli matrix $p \in \{X, Y, Z\}$, we define an $(n-1)$-qubit homogeneous $(\kappa^*-1)$-local Hermitian operator,
\begin{equation} \label{eq:Hkappaip}
    H_{\kappa^*, i, p} \triangleq \sum_{\substack{P \in \{I, X, Y, Z\}^{\otimes n}:\\ |P| = \kappa^*, P_i = p}} \alpha_P \left(\bigotimes_{j \neq i} P_j\right),
\end{equation}
For each qubit $i$ and $p \in \{X, Y, Z\}$, the algorithm computes the real value given as follows,
\begin{equation} \label{eq:beta-ip}
    \beta_{i, p} \triangleq \E_{\sigma\in\{\pm 1\}^{\kappa^* - 1}}\left[ \sigma_1\cdots\sigma_{\kappa^* - 1} \Tr\left[H_{\kappa^*, i, p} \bigotimes_{j \neq i} \left[\frac{I}{2} + \frac{1}{\kappa^* - 1} \sum^{\kappa^* - 1}_{s=1} \sigma_s \left(\ketbra{\psi_{(s, j)}}{\psi_{(s, j)}} - \frac{I}{2}\right)\right] \right] \right].
\end{equation}
Then for each qubit $j$, we consider a single-qubit local optimization
\begin{equation} \label{eq:psi-local-opt}
    \ket{\psi_{{(\kappa^*, j)}}} \triangleq \argmax_{\ket{\phi}: \, \mbox{\scriptsize 1-qubit state}} \bra{\phi} \left( \sum_{p \in \{X, Y, Z\}} \beta_{j, p}p \right) \ket{\phi} = \frac{I + n_X X + n_Y Y + n_Z Z}{2},
\end{equation}
where $n_p = \beta_{j, p} / \sqrt{\sum_{q} \beta_{j, q}^2}$ for $p \in \{X, Y, Z\}$.
After the optimization, the algorithm samples random numbers $\sigma_s \in \{ \pm 1 \}, \forall s \in \{1, \ldots, \kappa^*\}$ to define a one-dimensional parameterized family of $n$-qubit product states,
\begin{equation}
    \rho\left(t; \ket{\psi_{(\cdot, \cdot)}}, \sigma\right) \triangleq \bigotimes_{j = 1}^n \left(\frac{I}{2} + \frac{t}{\kappa^*} \sum^{\kappa^*}_{s=1} \sigma_s \left(\ketbra{\psi_{(s, j)}}{\psi_{(s, j)}} - \frac{I}{2}\right)\right), \,\,\quad \forall t \in [-1, 1].
\end{equation}
We will denote this by $\rho(t)$ when $\ket{\psi_{(\cdot,\cdot)}}, \sigma$ are clear from context.
This concludes the third step.
The third step can be performed in time $\mathcal{O}(\mathsf{nnz}(H) 2^k)$.

The fourth step performs a polynomial optimization over the one-dimensional family,
\begin{equation}
    \max_{t \in [-1, 1]} \left| \Tr\Big( H \rho\left(t; \ket{\psi_{(\cdot, \cdot)}}, \sigma\right) \Big)
    - \alpha_I \right|.
\end{equation}
The function $f(t) = \Tr( H \rho(t) )$ is a polynomial of degree at most $k$.
We can compute the function $f(t)$ efficiently in time $\mathcal{O}( \mathsf{nnz}(H) k)$ as $\rho(t)$ is a product state.
The optimization can thus be performed efficiently by sweeping through all possible values of $t$ on a sufficiently fine grid.
Let $t^*$ be the optimal $t$.

The final step considers the sampling of a random pure state $\ket{\psi} = \ket{\psi_1} \otimes \ldots \otimes \ket{\psi_n}$ from the distribution that corresponds to the mixed state $\rho\left(t^*; \ket{\psi_{(\cdot, \cdot)}}, \sigma\right)$.
If $\Tr( H \rho\left(t^*; \ket{\psi_{(\cdot, \cdot)}}, \sigma\right) ) - \alpha_I > 0$, then the random product state $\ket{\psi}$ is a maximizing state satisfying Eq.~\eqref{eq:max-state-opt}.
Otherwise, the random product state $\ket{\psi}$ is a minimizing state satisfying Eq.~\eqref{eq:min-state-opt}.
This step can be performed in time $\mathcal{O}(n)$.

\subsection{Proof of Theorem~\ref{thm:main}}

The first step of the algorithm considers splitting the $k$-local Hamiltonian $H$ into homogeneous $\kappa$-local Hamiltonians $H_\kappa$ defined below.
In particular, a homogeneous $\kappa^*$-local $H_{\kappa^*}$ is chosen.

\begin{definition}[Homogeneous $k$-local]
    A Hermitian operator $H$ is homogeneous $k$-local if $H = \sum_{P: |P| = k} \alpha_P P$.
\end{definition}

\noindent The second step is a random sampling that generates a single-qubit pure state $\ket{\psi_{(s, j)}}$ for each qubit $j$ and each copy $s \in \{1, \ldots, \kappa^* - 1\}$.
The third step is the most important part of the proof.
We will devote Section \ref{sec:polarization}, \ref{sec:khintchine}, and \ref{sec:charac-local} to establish the first inequality given below (Corollary~\ref{cor:final-char}).
\begin{align}
    \E_{\ket{\psi_{(\cdot, \cdot)}}} \E_{\sigma \in \{\pm 1\}^{\kappa^*}} \Big| \Tr\left( H_{\kappa^*} \rho\left(t = 1; \ket{\psi_{(\cdot, \cdot)}}, \sigma\right) \right) \Big|
    &\geq \frac{\sqrt{2 (k!)}}{ c_e^{1 / (2 d_e)} k^{k + 1.5} \sqrt{6}^{k}} \left( \sum_{P: |P| = \kappa^*} |\alpha_P|^r \right)^{1/r} \\
    &\geq \frac{\sqrt{2 (k!)}}{ c_e^{1 / (2 d_e)} k^{k + 1.5 + 1/r} \sqrt{6}^{k}} \left( \sum_{P \neq I} |\alpha_P|^r \right)^{1/r}.
\end{align}
The second inequality follows from $k \sum_{P: |P| = \kappa^*} |\alpha_P|^r \geq \sum_{\kappa = 1}^k \sum_{P: |P| = \kappa} |\alpha_P|^r = \sum_{P \neq I} |\alpha_P|^r$.
For the fourth step, the analysis of polynomial optimization given in Section \ref{sec:homo2inhomo} (Corollary~\ref{cor:poly-opt}) can be combined with the above inequality to obtain
\begin{equation}
    \E_{\ket{\psi_{(\cdot, \cdot)}}} \E_{\sigma \in \{\pm 1\}^{\kappa^*}} \Big| \Tr\left( H \rho\left(t^*; \ket{\psi_{(\cdot, \cdot)}}, \sigma\right) \right) - \alpha_I \Big|
    \geq \frac{\sqrt{2 (k!)}}{ c_e^{1 / (2 d_e)} k^{k + 1.5 + 1/r} (\sqrt{6} (1 + \sqrt{2}))^{k}} \left( \sum_{P \neq I} |\alpha_P|^r \right)^{1/r}.
\end{equation}
For the final step of the algorithm, using $\E_{\ket{\psi}} \ketbra{\psi}{\psi} = \rho(t^*; \rho_{(s, j)}, \sigma_s)$ and convexity, we have
\begin{align}
    \E_{\ket{\psi_{(\cdot, \cdot)}}} \E_{\sigma \in \{\pm 1\}^{\kappa^*}} \E_{\ket{\psi}} \Big| \bra{\psi} H \ket{\psi} - \alpha_I \Big| &\geq \E_{\ket{\psi_{(\cdot, \cdot)}}} \E_{\sigma \in \{\pm 1\}^{\kappa^*}} \Big| \Tr\left( H \E_{\ket{\psi}} \ketbra{\psi}{\psi} \right) - \alpha_I \Big|\\
    &\geq \frac{\sqrt{2 (k!)}}{ c_e^{1 / (2 d_e)} k^{k + 1.5 + 1/r} (\sqrt{6} + 2 \sqrt{3})^{k}} \left( \sum_{P \neq I} |\alpha_P|^r \right)^{1/r}.
\end{align}
The theorem follows by noting that $\E_{\ket{\phi}: \mathrm{Haar}} \big[ \bra{\phi} H \ket{\phi} \big] = \alpha_I$.


\subsubsection{Polarization}
\label{sec:polarization}

We justify the definition of $\beta_{i, p}$ using polarization.
Given an $n$-qubit homogeneous $k$-local observable $O = \sum_{P: |P| = k} \alpha_P P$, consider the following $nk$-qubit observable.
First, we will index the set $[nk]$ using ordered tuples $(s, i)$ where $s\in[k]$ and $i\in[n]$.
For every Pauli operator $P$ on $n$ qubits with $|P| = k$, suppose that it acts nontrivially on qubits $i_1 < \cdots < i_k$ via Pauli matrices $P_{i_1},\ldots,P_{i_k}$.
Then for any permutation $\pi\in\calS_k$, consider the $nk$-qubit observable $\pol_{\pi}(P)$ which acts on the $(\pi(s), i_s)$-th qubit via $P_{i_s}$ for all $s \in [k]$.
Then define
\begin{equation}
    \pol(P) \coloneqq \frac{1}{k!} \sum_{\pi\in \calS_k} \pol_{\pi}(P).
\end{equation}
We can extend $\pol(\cdot)$ linearly and define $\pol(O) \triangleq \sum_P \alpha_P \pol(P)$. We refer to $\pol(O)$ as the \emph{polarization} of $O$.
The squared Frobenius norm of $O$ and $\mathsf{pol}(O)$ are related by
\begin{equation} \label{eq:pol-frobenius}
     \Tr(O^2) = k! \Tr(\mathsf{pol}(O)^2).
\end{equation}
We prove the following operator analogue of the classical polarization identity:





\begin{lemma}[Polarization identity] \label{lem:polar-id}
    For any $nk$-qubit product state $\rho = \bigotimes_{s \in [k]} \left[ \bigotimes_{i\in [n]} \rho_{(s, i)} \right]$ and any $n$-qubit homogeneous $k$-local observable $O$ and any $t \in \mathbb{R}$, we have the following identity
    \begin{equation}
        t^k \Tr(\pol(O) \rho) = \frac{k^k}{k!} \E_{\sigma\in\{\pm 1\}^k}\left[\sigma_1\cdots\sigma_k \cdot \Tr\left( O \bigotimes_{i\in[n]} \left\{\frac{I}{2} + \frac{t}{k} \sum^k_{s=1} \sigma_s \left( \rho_{(s, i)} - \frac{I}{2} \right) \right\}\right)\right], \label{eq:polar}
    \end{equation}
    where the expectation is with respect to the uniform measure on $\{\pm 1\}^k$.
\end{lemma}

\begin{proof}
    Let $O = \sum_{P: |P| = k} \alpha_P P$. By the multinomial theorem, we can expand the right-hand side to get
    \begin{equation}
        \frac{t^k}{k!} \sum_{P: |P| = k} \alpha_P \E_\sigma\left[\sigma_1\cdots\sigma_k \sum_{0 \le s_1,\ldots,s_n \le k} \Tr\left(P \bigotimes^n_{i=1} \left\{\frac{I}{2} \cdot \mathds{1}[s_i = 0] + \sigma_{s_i} \left( \rho_{s_i, i} - \frac{I}{2} \right)\cdot \mathds{1}[s_i > 0]\right\}\right)\right]. \label{eq:expand}
    \end{equation}
    For a given Pauli operator $P$, note that the only terms in the inner summation that are nonzero are given by $(s_1,\ldots,s_n)$ satisfying that if $s_i > 0$, then $P$ acts nontrivially on the $i$-th qubit, because otherwise $\Tr(\rho_{s_i, i} - I/2) = 0$ and the corresponding summand vanishes.
    Furthermore, for $(s_1,\ldots,s_n)$ satisfying this property, if $\{1,\ldots,k\}$ do not each appear exactly once, then
    \begin{align}
        &\sigma_1\cdots\sigma_k \cdot  \bigotimes^n_{i=1} \left\{\frac{I}{2} \cdot \mathds{1}[s_i = 0] + \sigma_{s_i} \left( \rho_{s_i, i} - \frac{I}{2} \right)\cdot \mathds{1}[s_i > 0]\right\} \\
        &= \sigma^{c_1}_1 \cdots \sigma^{c_k}_k \cdot \bigotimes^n_{i=1} \left\{\frac{I}{2} \cdot \mathds{1}[s_i = 0] + \left( \rho_{s_i, i} - \frac{I}{2} \right)\cdot \mathds{1}[s_i > 0]\right\}
    \end{align}
    for $0\le c_1,\ldots,c_k\le k$ such that $c_s = 1$ for some $s\in[k]$. In this case, the expectation of this term with respect to $\sigma$ vanishes. Altogether, we conclude that for $P$ which acts via $P_1,\ldots,P_k$ on qubits $1\le i_1< \cdots < i_k\le n$ and via identity elsewhere, the corresponding expectation over $\sigma$ in Eq.~\eqref{eq:expand} is given by
    \begin{equation}
        \sum_{\pi\in\calS_k} \Tr\left(\bigotimes^k_{s=1} P_j\left(\rho_{\pi(s), i_s} - \frac{I}{2} \right)\right) = \sum_{\pi\in\calS_k} \Tr\left(\bigotimes^k_{s=1} P_s \rho_{\pi(s), i_s}\right) = \sum_{\pi}\Tr(\pol_\pi(P)\rho),
    \end{equation}
    from which the lemma follows.
\end{proof}

Using the polarization identity, we can obtain the following corollary, which shows that $\beta_{i, p}$ is defined to be proportional to the expection of the polarization $\mathsf{pol}(H_{\kappa^*, i, p})$ of the homogeneous $\kappa^*$-local observable $H_{\kappa^*, i, p}$ on the tensor product of $n(\kappa^* - 1)$ single-qubit Haar-random states.
We will later study the expectation value of the polarized observable on random product states.

\begin{corollary} \label{cor:beta-pol}
From the definitions given in Section~\ref{sec:desc-algo}, we have
\begin{equation}
    \Tr\left(\mathsf{pol}(H_{\kappa^*, i, p}) \bigotimes_{s \in [\kappa^*-1], i \in [n]} \ketbra{\psi_{(s, j)}}{\psi_{(s, j)}} \right) = \frac{(\kappa^* - 1)^{\kappa^* - 1}}{(\kappa^* - 1)!} \beta_{i, p}.
\end{equation}
\end{corollary}
\begin{proof}
    The claim follows from the polarization identity in Lemma~\ref{lem:polar-id} and the definition of $\beta_{i, p}$ in Eq.~\eqref{eq:beta-ip}.
\end{proof}

\subsubsection{Khintchine inequality for polarized observables}
\label{sec:khintchine}

We recall the following basic result in high-dimensional probability.

\begin{lemma}[Standard Khintchine inequality \cite{Haagerup1981}] \label{lem:khintchine}
Consider $\varepsilon_1, \ldots, \varepsilon_n$ to be i.i.d. random variables with $P(\varepsilon_i = \pm 1) = 1/2$. For any $a_1, \ldots, a_n \in \mathbb{R}$, we have
\begin{equation}
    \frac{1}{\sqrt{2}} \left( \sum_{i=1}^n a_i^2 \right)^{1/2} \leq \E_{\varepsilon_1, \ldots \varepsilon_n} \left| \sum_{i=1}^n a_i \varepsilon_i \right| \leq \left( \sum_{i=1}^n a_i^2 \right)^{1/2}.
\end{equation}
\end{lemma}

\noindent We prove an analogue of the Khintchine inequality when we replace the random $\pm 1$ variables with random product states and replace $a_1, \ldots, a_n$ with a homogeneous $1$-local observable.

\begin{lemma}[Khintchine inequality for homogeneous $1$-local observables] \label{lem:khintchine-obs-1}
    Let $n \geq 1$.
    Consider $\ket{\psi} = \bigotimes_{i=1}^n \ket{\psi_i}$ where $\ket{\psi_i}$ is a single-qubit Haar-random pure state.
    For any homogeneous $1$-local $n$-qubit observable $O$,
    \begin{equation}
        \frac{1}{\sqrt{6}} \sqrt{\Tr(O^2) / 2^n} \leq \E_{\ket{\psi}}\left[ \left| \bra{\psi} O \ket{\psi} \right| \right] \leq \frac{1}{\sqrt{3}} \sqrt{\Tr(O^2) / 2^n}.
    \end{equation}
\end{lemma}
\begin{proof}
    A homogeneous $1$-local observable $O$ is $\sum_{i = 1}^{n} \sum_{j=1}^3 \alpha_{ij} P^j_i,$ where $P^j_i$ is the Pauli matrix $\sigma_j \in \{X, Y, Z\}$ on the $i$-th qubit.
    Given $n$ single-qubit unitaries $U_1, \ldots, U_n$, we consider $O$ under the rotated Pauli basis
    \begin{equation}
        O = \sum_{i = 1}^{n} \sum_{j=1}^3 \alpha^U_{ij} U_i^\dagger P^j_i U_i.
    \end{equation}
    Using the orthogonality of Pauli matrices,  we have
    \begin{equation} \label{eq:Frob-rotated}
    \sqrt{\Tr(O^2) / 2^n} = \left( \sum_{i=1}^n \sum_{j=1}^3 (a^U_{ij})^2 \right)^{1/2}
    \end{equation}
    under any rotated Pauli basis.
    We will utilize the rotated Pauli basis to establish the claimed results.

    A single-qubit Haar-random pure state $\ket{\psi_i}$ can be sampled as follows.
    First, we sample a random single-qubit unitary $U_i$.
    Then, we consider $\ket{\psi_i}$ to be sampled uniformly from the set of $8$ pure states,
    \begin{equation}
        \Upsilon^{U_i} = \left\{ \frac{I + \frac{1}{\sqrt{3}}(s^X_i U_i X U_i^\dagger + s^Y_i U_i Y U_i^\dagger + s^Z_i U_i Z U_i^\dagger)}{2} \, \Bigg| \, s^X_i, s^Y_i, s^Z_i \in \{ \pm 1\}  \right\}.
    \end{equation}
    Using this sampling formulation and the rotated Pauli basis representation for $O$, we have
    \begin{align}
        \E_{\ket{\psi}}\left[ \left| \bra{\psi} O \ket{\psi} \right| \right] &= \E_{U_i} \E_{\ket{\psi_i} \sim \Upsilon^{U_i}} \left| \sum_{i = 1}^{n} \sum_{j=1}^3 \alpha^U_{ij} \Tr\left( U_i^\dagger P^j_i U_i \ketbra{\psi_i}{\psi_i} \right) \right|\\
        &= \frac{1}{\sqrt{3}} \E_{U_i} \E_{s^X_i, s^Y_i, s^Z_i \sim \{\pm 1\}} \left| \sum_{i = 1}^{n} \alpha^U_{i1} s^X_i + \alpha^U_{i2} s^Y_i + \alpha^U_{i3} s^Z_i \right|.
    \end{align}
    Using the standard Khintchine inequality given in Lemma~\ref{lem:khintchine}, we have
    \begin{equation}
        \frac{1}{\sqrt{2}} \left( \sum_{i=1}^n \sum_{j=1}^3 (a^U_{ij})^2 \right)^{1/2} \leq \E_{s^X_i, s^Y_i, s^Z_i \sim \{\pm 1\}} \left| \sum_{i = 1}^{n} \alpha^U_{i1} s^X_i + \alpha^U_{i2} s^Y_i + \alpha^U_{i3} s^Z_i \right| \leq \left( \sum_{i=1}^n \sum_{j=1}^3 (a^U_{ij})^2 \right)^{1/2}.
    \end{equation}
    Using Eq.~\eqref{eq:Frob-rotated}, we can obtain
    \begin{equation}
        \frac{1}{\sqrt{6}} \E_{U_i} \sqrt{\Tr(O^2) / 2^n} \leq \E_{\ket{\psi}}\left[ \left| \bra{\psi} O \ket{\psi} \right| \right] \leq \frac{1}{\sqrt{3}} \E_{U_i} \sqrt{\Tr(O^2) / 2^n},
    \end{equation}
    which implies the claimed result.
\end{proof}

We prove the left half of Khintchine inequality for polarized observables. The right half can be shown using a similar proof, but we are only going to use the left half stated below.

\begin{lemma}[Khintchine inequality for polarized observables]\label{lem:khintchine-obs}
    Given $n, k > 0$. Consider an $nk$-qubit observable $O = \pol(O')$, which is the polarization of an $n$-qubit  homogeneous $k$-local observable $O'$. Consider $\ket{\psi} = \bigotimes_{s \in [k], i \in [n]} \ket{\psi_{(s, i)}}$ where $\ket{\psi_{(s, i)}}$ is a single-qubit Haar-random pure state. We have
    \begin{equation}
        \left(\frac{1}{\sqrt{6}}\right)^k \sqrt{\Tr(O^2) / 2^n} \leq \E_{\ket{\psi}}\left[ \left| \bra{\psi} O \ket{\psi} \right| \right].
    \end{equation}
\end{lemma}
\begin{proof}
    For $\ell \in [3n]$, define $P^{(\ell)}$ to be an $n$-qubit observable equal to the Pauli matrix $\sigma_{1+(\ell \, \mathrm{mod} \, 3)} \in \{X, Y, Z\}$ acting on the $\lceil \ell / 3\rceil$-th qubit.
    From the definition of polarization, we can represent $O$ as
    \begin{equation}
        O = \sum_{\ell_1, \ldots, \ell_k \in [3n]} \alpha_{\ell_1, \ldots, \ell_k} P^{(\ell_1)} \otimes \ldots \otimes P^{(\ell_k)}.
    \end{equation}
    For arbitrary coefficients $\alpha_{\ell_1, \ldots, \ell_k} \in \mathbb{R}$,
    we prove the following claim by induction on $k$,
    \begin{align} \label{eq:claim-polar-K}
        \left(\frac{1}{\sqrt{6}}\right)^k \left(\sum_{\ell_1, \ldots, \ell_k \in [3n]} \alpha^2_{\ell_1, \ldots, \ell_k} \right)^{1/2} \leq \E_{\ket{\psi}}\left[ \left| \bra{\psi} \sum_{\ell_1, \ldots, \ell_k \in [3n]} \alpha_{\ell_1, \ldots, \ell_k} P^{(\ell_1)} \otimes \ldots \otimes P^{(\ell_k)} \ket{\psi} \right| \right].
    \end{align}
    It is not hard to see that the left-hand side of Eq.~\eqref{eq:claim-polar-K} is $\left(\frac{1}{\sqrt{6}}\right)^k \sqrt{\Tr(O^2) / 2^n}$ and the right-hand side of Eq.~\eqref{eq:claim-polar-K} is $\E_{\ket{\psi}}\left[ \left| \bra{\psi} O \ket{\psi} \right| \right]$.
    Hence, the lemma follows from Eq.~\eqref{eq:claim-polar-K}.

    We now prove the base case and the inductive step.
    The base case of $k=1$ follows from Khintchine inequality for homogeneous $1$-local observables given in Lemma~\ref{lem:khintchine-obs-1}.
    Assume by induction hypothesis that the claim holds for $k-1$.
    By denoting $\ket{\psi^{(k)}}$ to be a product of $n$ Haar-random single-qubit states,
    we can then apply Khintchine inequality for homogeneous $1$-local observables (Lemma~\ref{lem:khintchine-obs-1}) to obtain
    \begin{align}
        &\left(\frac{1}{\sqrt{6}}\right)^k \left(\sum_{\ell_1, \ldots, \ell_k \in [3n]} \alpha^2_{\ell_1, \ldots, \ell_k} \right)^{1/2}\\
        &= \left(\sum_{\ell_1, \ldots, \ell_{k-1} \in [3n]} \left( \left( \sum_{\ell_k \in [3n]} \alpha^2_{\ell_1, \ldots, \ell_k} \right)^{1/2} \right)^2 \right)^{1/2} \\
        &\leq \left(\sum_{\ell_1, \ldots, \ell_{k-1} \in [3n]} \left( \E_{\ket{\psi^{(k)}}} \left| \bra{\psi^{(k)}} \sum_{\ell_k \in [3n]} \alpha_{\ell_1, \ldots, \ell_k} P^{(\ell_k)} \ket{\psi^{(k)}} \right| \right)^2 \right)^{1/2}.
    \end{align}
    We can then apply Minkowski's integral inequality to upper bound the above and yield
    \begin{align}
        &\left(\frac{1}{\sqrt{6}}\right)^k \left(\sum_{\ell_1, \ldots, \ell_k \in [3n]} \alpha^2_{\ell_1, \ldots, \ell_k} \right)^{1/2}\\
        &\leq \E_{\ket{\psi^{(k)}}} \left(\sum_{\ell_1, \ldots, \ell_{k-1} \in [3n]} \left( \bra{\psi^{(k)}} \sum_{\ell_k \in [3n]} \alpha_{\ell_1, \ldots, \ell_k} P^{(\ell_k)} \ket{\psi^{(k)}} \right)^2 \right)^{1/2} \\
        &\leq \E_{\ket{\psi^{(k)}}} \E_{\ket{\psi^{(1, \ldots, k-1)}}} \left| \bra{\psi^{(1, \ldots, k-1)}}\bra{\psi^{(k)}} \sum_{\ell_1, \ldots, \ell_k \in [3n]} \alpha_{\ell_1, \ldots, \ell_k} P^{(\ell_1)} \otimes \ldots \otimes P^{(\ell_k)} \ket{\psi^{(1, \ldots, k-1)}} \ket{\psi^{(k)}} \right|.
    \end{align}
    The last line considers $\bra{\psi^{(k)}} \sum_{\ell_k \in [3n]} \alpha_{\ell_1, \ldots, \ell_k} P^{(\ell_k)} \ket{\psi^{(k)}}$ to be a scalar indexed by $\ell_1, \ldots, \ell_{k-1}$ and uses the induction hypothesis. We have thus established the induction step.
    The claim in Eq.~\eqref{eq:claim-polar-K} follows.
\end{proof}

Khintchine inequality for polarized observable allows us to show that the average magnitude of $\mathsf{pol}(H_{\kappa^*, i, p})$ for the tensor product of single-qubit Haar-random states is at least as large as the Frobenius norm of $H_{\kappa^*, i, p}$ up to a constant depending on $\kappa^*$.
Using the definitions from the design of the approximate optimization algorithm, we can obtain the following corollary.

\begin{corollary} \label{cor:polH-random}
From the definitions given in Section~\ref{sec:desc-algo}, we have
\begin{equation}
    \E_{\ket{\psi_{(\cdot, \cdot)}}}
    \left| \Tr\left(\mathsf{pol}(H_{\kappa^*, i, p}) \bigotimes_{s \in [\kappa^*-1], i \in [n]} \ketbra{\psi_{(s, j)}}{\psi_{(s, j)}} \right) \right| \geq \left( \frac{1}{\sqrt{6}} \right)^{\kappa^* - 1} \sqrt{\frac{\Tr(H_{\kappa^*, i, p}^2)}{2^n (\kappa^* - 1)!}}.
\end{equation}
\end{corollary}
\begin{proof}
    The claim follows immediately from Lemma~\ref{lem:khintchine-obs} and Eq.~\eqref{eq:pol-frobenius}.
\end{proof}

\subsubsection{Characterization of the locally optimized random state}
\label{sec:charac-local}

Recall that $\rho\left(1; \ket{\psi_{(\cdot, \cdot)}}, \sigma\right)$ is created by sampling random product states and performing local single-qubit optimizations.
The locally optimized random state satisfies the following inequality.

\begin{lemma}[Characterization of $\rho(t)$ for $t = 1$] \label{lem:key-char}
From the definitions given in Section~\ref{sec:desc-algo}, we have
\begin{equation}
    \E_{\ket{\psi_{(\cdot, \cdot)}}} \E_{\sigma \in \{\pm 1\}^{\kappa^*}} \Big| \Tr\left( H_{\kappa^*} \rho\left(1; \ket{\psi_{(\cdot, \cdot)}}, \sigma\right) \right) \Big|
    \geq \frac{\sqrt{2 (\kappa^*!)}}{ (\kappa^*)^{\kappa^* + 1.5} \sqrt{6}^{\kappa^*}} \sum_{i \in [n], p \in \{X, Y, Z\}} \sqrt{\sum_{\substack{P \in \{I, X, Y, Z\}^{\otimes n}:\\ |P| = \kappa^*, P_i = p}} \alpha_P^2}.
\end{equation}
\end{lemma}
\begin{proof}
From the polarization identity given in Lemma~\ref{lem:polar-id}, we have
\begin{equation} \label{eq:Hkip-1}
    \frac{(\kappa^*)^{\kappa^*}}{\kappa^* !} \E_{\sigma \in \{\pm 1\}^{\kappa^*}} \left[ \sigma_1 \cdots \sigma_{\kappa^*} \Tr\left( H_{\kappa^*} \rho\left(1; \ket{\psi_{(\cdot, \cdot)}}, \sigma \right) \right) \right] = \Tr\left(\mathsf{pol}(H_{\kappa^*}) \bigotimes_{s \in [\kappa^*], j \in [n]} \ketbra{\psi_{(s, j)}}{\psi_{(s, j)}} \right).
\end{equation}
Next, using the definition of $H_{\kappa^*, i, p}$ in Eq.~\eqref{eq:Hkappaip}, we have
\begin{equation} \label{eq:Hkip-2}
    \mathsf{pol}(H_{\kappa^*}) = \left(\frac{1}{\kappa^*}\right)^2 \sum_{i \in [n]} \sum_{p \in \{X, Y, Z\}} \mathsf{pol}(H_{\kappa^*, i, p}) \otimes (I^{\otimes i-1} \otimes p \otimes I^{\otimes n-i}).
\end{equation}
We can see this by considering the case when $H_{\kappa^*}$ is a single Pauli observable $P \in \{I, X, Y, Z\}^{\otimes n}$ with $|P| = \kappa^*$, and then extending linearly to any homogeneous $\kappa^*$-local Hamiltonian $H_{\kappa^*}$.
Eq.~\eqref{eq:Hkip-1}~and~\eqref{eq:Hkip-2} give
\begin{align}
    &\frac{(\kappa^*)^{\kappa^*}}{\kappa^* !} \E_{\sigma \in \{\pm 1\}^{\kappa^*}} \left[ \sigma_1 \cdots \sigma_{\kappa^*} \Tr\left( H_{\kappa^*} \rho\left(1; \ket{\psi_{(\cdot, \cdot)}}, \sigma \right) \right) \right] \\
    &=
    \frac{1}{(\kappa^*)^2} \sum_{i \in [n], p \in \{X, Y, Z\}} \bra{\psi_{(\kappa^*, i)}} p \ket{\psi_{(\kappa^*, i)}} \, \Tr\left(\mathsf{pol}(H_{\kappa^*, i, p}) \bigotimes_{s \in [\kappa^*-1], j \in [n]} \ketbra{\psi_{(s, j)}}{\psi_{(s, j)}} \right).
\end{align}
From Corollary~\ref{cor:beta-pol}, we can rewrite the right hand side as
\begin{equation}
    \frac{1}{(\kappa^*)^2} \frac{(\kappa^* - 1)^{\kappa^* - 1}}{(\kappa^* - 1)!} \sum_{i \in [n]} \bra{\psi_{(\kappa^*, i)}} \left(\sum_{p \in \{X, Y, Z\}} \beta_{i, p} p \right) \ket{\psi_{(\kappa^*, i)}}.
\end{equation}
From the local optimization of $\ket{\psi_{(\kappa^*, i)}}$ given in Eq.~\eqref{eq:psi-local-opt}, we have that for every $i\in[n]$,
\begin{equation}
    \bra{\psi_{(\kappa^*, i)}} \left(\sum_{p \in \{X, Y, Z\}} \beta_{i, p} p \right) \ket{\psi_{(\kappa^*, i)}} =  \sqrt{\sum_{p \in \{X, Y, Z\}} \beta_{i, p}^2} \geq \frac{1}{\sqrt{3}}  \sum_{p \in \{X, Y, Z\}} \left| \beta_{i, p} \right|.
\end{equation}
Using Corollary~\ref{cor:beta-pol} yields the following lower bound,
\begin{align}
    &\frac{(\kappa^*)^{\kappa^*}}{\kappa^* !} \E_{\ket{\psi_{(\cdot, \cdot)}}} \E_{\sigma \in \{\pm 1\}^{\kappa^*}} \left[ \sigma_1 \cdots \sigma_{\kappa^*} \Tr\left( H_{\kappa^*} \rho\left( 1; \ket{\psi_{(\cdot, \cdot)}}, \sigma \right) \right) \right] \\
    &\geq \frac{1}{\sqrt{3} (\kappa^*)^2} \sum_{i \in [n], p \in \{X, Y, Z\}} \E_{\ket{\psi_{(\cdot, \cdot)}}} \left| \Tr\left(\mathsf{pol}(H_{\kappa^*, i, p}) \bigotimes_{s \in [\kappa^*-1], j \in [n]} \ketbra{\psi_{(s, j)}}{\psi_{(s, j)}} \right) \right|.
\end{align}
From Corollary~\ref{cor:polH-random}, we can further obtain
\begin{align}
    &\frac{(\kappa^*)^{\kappa^*}}{\kappa^* !} \E_{\ket{\psi_{(\cdot, \cdot)}}} \E_{\sigma \in \{\pm 1\}^{\kappa^*}} \left[ \sigma_1 \cdots \sigma_{\kappa^*} \Tr\left( H_{\kappa^*} \rho\left( 1; \ket{\psi_{(\cdot, \cdot)}}, \sigma \right) \right) \right] \\
    &\geq \frac{1}{\sqrt{3} (\kappa^*)^2} \sum_{i \in [n], p \in \{X, Y, Z\}} \left( \frac{1}{\sqrt{6}} \right)^{\kappa^* - 1} \sqrt{\frac{\Tr(H_{\kappa^*, i, p}^2)}{2^n (\kappa^* - 1)!}}.
\end{align}
The definition of $H_{\kappa^*, i, p}$, the above inequality, and the following inequality
\begin{equation}
    \E_{\ket{\psi_{(\cdot, \cdot)}}} \E_{\sigma \in \{\pm 1\}^{\kappa^*}} \Big| \Tr\left( H_{\kappa^*} \rho\left(1; \ket{\psi_{(\cdot, \cdot)}}, \sigma\right) \right) \Big| \geq \E_{\ket{\psi_{(\cdot, \cdot)}}} \E_{\sigma \in \{\pm 1\}^{\kappa^*}} \left[ \sigma_1 \cdots \sigma_{\kappa^*} \Tr\left( H_{\kappa^*} \rho\left( 1; \ket{\psi_{(\cdot, \cdot)}}, \sigma \right) \right) \right],
\end{equation}
can be used to establish the claim.
\end{proof}

Given the expansion property, we are going to use the following implication, which considers an arbitrary ordering $\pi$ of the $n$ qubits.
The inequality allows us to control the growth for the number of Pauli observables that act on qubits before the $i$-th qubit under the ordering $\pi$.
The precise statement is given below.

\begin{lemma}[A characterization of expansion] \label{lem:expansion-char}
    Given an $n$-qubit Hamiltonian $H = \sum_{P} \alpha_P P$ with expansion coefficient $c_e$ and expansion dimension $d_e$.
    Consider any permutation $\pi \in S_{n}$ over $n$ qubits.
    For any $i \in [n]$,
    \begin{equation}
        \sum_{P \in \{I, X, Y, Z\}^{\otimes n}} \indicator[\alpha_P \neq 0] \indicator[P_{\pi(i)} \neq I] \indicator[P_{\pi(j)} = I, \forall j > i] \leq c_e i^{d_e - 1},
    \end{equation}
\end{lemma}
\begin{proof}
    Given a permutation $\pi \in S_{n}$ over $n$ qubits and an $i \in [n]$. We separately consider two cases: (1) $i < d_e$ and (2) $i \geq d_e$. For the first case, let $\Upsilon = \{\pi(1), \ldots, \pi(d_e)\}$, we have
    \begin{align}
        &\sum_{P \in \{I, X, Y, Z\}^{\otimes n}} \indicator[\alpha_P \neq 0] \indicator[P_{\pi(i)} \neq I] \indicator[P_{\pi(j)} = I, \forall j > i] \\
        &\leq \sum_{P \in \{I, X, Y, Z\}^{\otimes n}} \,\, \indicator\Big[\alpha_P \neq 0 \,\, \mathrm{and} \,\, \big( \mathsf{dom}(P) \subseteq \Upsilon \big)  \, \Big] \leq c_e.
    \end{align}
    The second inequality follows from the definition of the expansion coefficient $c_e$.
    For the second case, we consider all subset $\Upsilon \subseteq \pi([i]) \triangleq \{\pi(1), \pi(2), \ldots, \pi(i)\}$
    with $|\Upsilon| = d_e - 1$ and $\pi(i) \in \Upsilon$,
    \begin{align}
        &\sum_{P \in \{I, X, Y, Z\}^{\otimes n}} \indicator[\alpha_P \neq 0] \indicator[P_{\pi(i)} \neq I] \indicator[P_{\pi(j)} = I, \forall j > i] \\
        &\leq \sum_{P \in \{I, X, Y, Z\}^{\otimes n}} \,\, \sum_{\substack{\Upsilon \subseteq \pi([i]),\\ |\Upsilon| = d_e, \pi(i) \in \Upsilon}} \indicator\Big[\alpha_P \neq 0 \,\, \mathrm{and} \,\, \big( \mathsf{dom}(P) \subseteq \Upsilon \,\, \mathrm{or} \,\, \Upsilon \subseteq \mathsf{dom}(P) \big)  \, \Big] \\
        &\leq \sum_{\substack{\Upsilon \subseteq \pi([i]),\\ |\Upsilon| = d_e, \pi(i) \in \Upsilon}} c_e \leq c_e (i-1)^{d_e - 1} \leq c_e i^{d_e - 1}.
    \end{align}
    The second inequality again follows from the definition of $c_e$.
\end{proof}

Using the above implication of the expansion property, we can obtain the following inequality relating two norms.
Basically, we can use the limit on the growth of the number of Pauli observables to turn the sum of $\ell_2$-norm into an $\ell_r$-norm, where $r$ depends on the expansion dimension $d_e$.

\begin{lemma}[Norm inequality using expansion property] \label{lem:key-char-expansion}
    Given an $n$-qubit Hamiltonian $H = \sum_{P} \alpha_P P$ with an expansion coefficient $c_e$ and expansion dimension $d_e$.
    Let $r = 2 d_e / (d_e + 1)$.
    For any $\kappa^* \geq 1$, we have
    \begin{equation}
        \sum_{i \in [n] } \left( \sum_{\substack{P \in \{I, X, Y, Z\}^{\otimes n}:\\ |P| = \kappa^*, P_i \neq I}} \alpha_P^2\right)^{1/2} \geq \frac{1}{c_e^{1 / (2 d_e)}} \left( \sum_{\substack{P \in \{I, X, Y, Z\}^{\otimes n}:\\ |P| = \kappa^*}} \left| \alpha_P \right|^r \right)^{1/r}.
    \end{equation}
\end{lemma}
\begin{proof}
    We begin by considering a permutation $\pi$ over $n$ qubits, such that
    \begin{equation}
        \sqrt{\sum_{\substack{P \in \{I, X, Y, Z\}^{\otimes n}:\\ |P| = \kappa^*, P_{\pi(i)} \neq I}} \alpha_P^2} \leq \sqrt{\sum_{\substack{P \in \{I, X, Y, Z\}^{\otimes n}:\\ |P| = \kappa^*, P_{\pi(j)} \neq I}} \alpha_P^2}, \quad \forall i < j \in [n].
    \end{equation}
    The permutation $\pi$ can be obtained by sorting the $n$ qubits.
    The above ensures that for all $i \in [n]$,
    \begin{equation}\label{eq:pidef}
        i \sqrt{\sum_{\substack{P \in \{I, X, Y, Z\}^{\otimes n}:\\ |P| = \kappa^*, P_{\pi(i)} \neq I}} \alpha_P^2} \leq \sum_{j \in [n]} \sqrt{\sum_{\substack{P \in \{I, X, Y, Z\}^{\otimes n}:\\ |P| = \kappa^*, P_{\pi(j)} \neq I}} \alpha_P^2}.
    \end{equation}
    By going through the $n$ qubits based on the permutation $\pi$, we have the following identity,
    \begin{equation}
        \sum_{P: |P| = \kappa^*} \left| \alpha_P \right|^r = \sum_{i = 1}^n \sum_{p \in \{X, Y, Z\}} \sum_{\substack{P\in \{I, X, Y, Z\}^{\otimes n}:\\ |P| = \kappa^*, P_{\pi(i)} = p}} \left| \alpha_P \right|^r \indicator[\alpha_P \neq 0] \indicator[P_{\pi(j)} = I, \forall j > i].
    \end{equation}
    Holder's inequality and $1/(d_e + 1) = 1 - r/2$ allows us to obtain the following upper bound on $\sum_{P: |P| = \kappa^*} \left| \alpha_P \right|^r$,
    \begin{equation}
        \sum_{i = 1}^n \left( \sum_{\substack{P\in \{I, X, Y, Z\}^{\otimes n}:\\ |P| = \kappa^*, P_{\pi(i)} \neq I}} \alpha_P^2 \right)^{r/2} \left( \sum_{\substack{P\in \{I, X, Y, Z\}^{\otimes n}:\\ |P| = \kappa^*}} \indicator[\alpha_P \neq 0] \indicator[P_{\pi(i)} \neq I] \indicator[P_{\pi(j)} = I, \forall j > i] \right)^{1 / (d_e + 1)}.
    \end{equation}
    We can then use Lemma~\ref{lem:expansion-char} to obtain
    \begin{equation}
        \sum_{P: |P| = \kappa^*} \left| \alpha_P \right|^r \leq \sum_{i = 1}^n \left( \sum_{\substack{P\in \{I, X, Y, Z\}^{\otimes n}:\\ |P| = \kappa^*, P_{\pi(i)} \neq I}} \alpha_P^2 \right)^{r/2} \left( c_e i^{d_e - 1} \right)^{1 / (d_e + 1)}.
    \end{equation}
    Using $r - 1 = (d_e - 1) / (d_e + 1) \geq 0$, we have
    \begin{equation}
        \sum_{P: |P| = \kappa^*} \left| \alpha_P \right|^r \leq
        c_e^{1 / (d_e + 1)}
        \sum_{i = 1}^n \left( i \sqrt{\sum_{\substack{P\in \{I, X, Y, Z\}^{\otimes n}:\\ |P| = \kappa^*, P_{\pi(i)} \neq I}} \alpha_P^2} \right)^{r-1} \sqrt{ \sum_{\substack{P\in \{I, X, Y, Z\}^{\otimes n}:\\ |P| = \kappa^*, P_{\pi(i)} \neq I}} \alpha_P^2}
    \end{equation}
    The choice of $\pi$ ensures Eq.~\eqref{eq:pidef}, which gives rise to
    \begin{equation}
        \sum_{P: |P| = \kappa^*} \left| \alpha_P \right|^r \leq c_e^{1 / (d_e + 1)} \left( \sum_{i \in [n] } \left( \sum_{\substack{P \in \{I, X, Y, Z\}^{\otimes n}:\\ |P| = \kappa^*, P_i \neq I}} \alpha_P^2\right)^{1/2} \right)^r.
    \end{equation}
    The claim follows from $1 / (r(d_e + 1)) = 1 / (2 d_e)$.
\end{proof}

Together, we can obtain the $\ell_r$-norm lower bound for the expectation value of the homogeneous $\kappa^*$-local Hamiltonian $H_{\kappa^*}$ on the constructed product state $\rho\left(1; \ket{\psi_{(\cdot, \cdot)}}, \sigma\right)$.

\begin{corollary} \label{cor:final-char}
    From the definitions given in Section~\ref{sec:desc-algo}, we have
    \begin{align}
        \E_{\ket{\psi_{(\cdot, \cdot)}}} \E_{\sigma \in \{\pm 1\}^{\kappa^*}} \Big| \Tr\left( H_{\kappa^*} \rho\left(1; \ket{\psi_{(\cdot, \cdot)}}, \sigma\right) \right) \Big|
        &\geq \frac{\sqrt{2 (\kappa^*!)}}{ c_e^{1 / (2 d_e)} (\kappa^*)^{\kappa^* + 1.5} \sqrt{6}^{\kappa^*}} \left( \sum_{\substack{P \in \{I, X, Y, Z\}^{\otimes n}:\\ |P| = \kappa^*}} \left| \alpha_P \right|^r \right)^{1/r} \\
        &\geq \frac{\sqrt{2 (k!)}}{ c_e^{1 / (2 d_e)} k^{k + 1.5} \sqrt{6}^{k}} \left( \sum_{\substack{P \in \{I, X, Y, Z\}^{\otimes n}:\\ |P| = \kappa^*}} \left| \alpha_P \right|^r \right)^{1/r}.
    \end{align}
\end{corollary}
\begin{proof}
    From Lemma~\ref{lem:key-char}, we have
    \begin{equation}
        \E_{\ket{\psi_{(\cdot, \cdot)}}} \E_{\sigma \in \{\pm 1\}^{\kappa^*}} \Big| \Tr\left( H_{\kappa^*} \rho\left(1; \ket{\psi_{(\cdot, \cdot)}}, \sigma\right) \right) \Big|
        \geq \frac{\sqrt{2 (\kappa^*!)}}{ (\kappa^*)^{\kappa^* + 1.5} \sqrt{6}^{\kappa^*}} \sum_{i \in [n], p \in \{X, Y, Z\}} \sqrt{\sum_{\substack{P \in \{I, X, Y, Z\}^{\otimes n}:\\ |P| = \kappa^*, P_i = p}} \alpha_P^2}.
    \end{equation}
    By the elementary inequality $\sqrt{x} + \sqrt{y} + \sqrt{z} \ge \sqrt{x+y+z}$ for nonnegative $x,y,z$,
    \begin{equation}
        \sum_{i \in [n], p \in \{X, Y, Z\}} \sqrt{\sum_{\substack{P \in \{I, X, Y, Z\}^{\otimes n}:\\ |P| = \kappa^*, P_i = p}} \alpha_P^2} \geq  \sum_{i \in [n]} \sqrt{ \sum_{p \in \{X, Y, Z\}} \sum_{\substack{P \in \{I, X, Y, Z\}^{\otimes n}:\\ |P| = \kappa^*, P_i = p}} \alpha_P^2}.
    \end{equation}
    Combining with Lemma~\ref{lem:key-char-expansion} and the fact that $k \geq \kappa^*$ yields the stated result.
\end{proof}

\subsubsection{Homogeneous to inhomogeneous through polynomial optimization}
\label{sec:homo2inhomo}

We need the following basic result from real analysis:
\begin{lemma}[Markov brothers' inequality, see e.g. p. 248 of \cite{borwein1995polynomials}] \label{lem:markov-bro}
    For any real polynomial $p(t) = \sum^k_{\kappa=1} a_\kappa x^\kappa$,
    \begin{equation}
        |a_\kappa| \le (1+\sqrt{2})^k \sup_{|t|\le 1} |p(t)|
    \end{equation}
    for all $1 \le \kappa \le k$.
\end{lemma}

Using Markov brothers' inequality, we can show that performing the one-dimensional polynomial optimization over $t$ achieves a good advantage over $\alpha_I = \E_{\ket{\psi}:\mathrm{Haar}} \bra{\psi} H \ket{\psi}$.

\begin{corollary} \label{cor:poly-opt}
    From the definitions given in Section~\ref{sec:desc-algo}, we have
    \begin{equation}
        \Big| \Tr\left( H \rho\left(t^*; \ket{\psi_{(\cdot, \cdot)}}, \sigma\right) \right) - \alpha_I \Big|
        \geq \frac{1}{(1 + \sqrt{2})^k} \Big| \Tr\left( H_{\kappa^*} \rho\left(1; \ket{\psi_{(\cdot, \cdot)}}, \sigma\right) \right) \Big|.
    \end{equation}
\end{corollary}
\begin{proof}
    Recall that $H = \alpha_I I + \sum_{\kappa = 1}^k H_{\kappa}$ from Eq.~\eqref{eq:H-decomp}.
    We can use the polarization identity given in Lemma~\ref{lem:polar-id} to see that the function $f(t) = \Tr\left( H \rho\left(t ; \ket{\psi_{(\cdot, \cdot)}}, \sigma\right) \right)$ is a polynomial,
    \begin{equation}
        \Tr\left( H \rho\left(t ; \ket{\psi_{(\cdot, \cdot)}}, \sigma\right) \right) = \alpha_I + \sum_{\kappa=1}^k \Tr\left( H_{\kappa} \rho\left(1; \ket{\psi_{(\cdot, \cdot)}}, \sigma\right) \right) t^{\kappa}.
    \end{equation}
    Recall that $t^*$ is chosen based on the optimization
    \begin{equation}
    \max_{t \in [-1, 1]} \left| \Tr\Big( H \rho\left(t; \ket{\psi_{(\cdot, \cdot)}}, \sigma\right) \Big)
    - \alpha_I \right|.
    \end{equation}
    By considering Lemma~\ref{lem:markov-bro} with $a_\kappa = \Tr\left( H_{\kappa} \rho\left(1; \ket{\psi_{(\cdot, \cdot)}}, \sigma\right) \right)$, we have
    \begin{equation}
        (1 + \sqrt{2})^k \left| \Tr\Big( H \rho\left(t^*; \ket{\psi_{(\cdot, \cdot)}}, \sigma\right) \Big)
    - \alpha_I \right| \geq \left| \Tr\left( H_{\kappa} \rho\left(1; \ket{\psi_{(\cdot, \cdot)}}, \sigma\right) \right) \right|.
    \end{equation}
    This concludes the proof of this corollary.
\end{proof}

\section{Norm inequalities from approximate optimization algorithm}
\label{sec:norm-ineq-Pauli}

The approximate optimization algorithm described in the previous section is not used directly in the ML algorithm, but used to derive norm inequalities, i.e., inequalities relating different norms over Hermitian operators.
An important norm that we will use in the ML algorithms is the Pauli-$p$ norm defined below.
The Pauli-$p$ norm is equivalent to the vector-$p$ norm on the Pauli coefficient of an observable $H$.

\begin{definition}[Pauli-$p$ norm]
    Given $H = \sum_{P \in \{I, X, Y, Z\}^{\otimes n}} \alpha_P P$ and $p \geq 1$.
    The Pauli-$p$ norm of $H$ is
    \begin{equation}
        \norm{H}_{\mathrm{Pauli}, p} = \left(\sum_{P} |\alpha_{P}|^p \right)^{1/p}.
    \end{equation}
\end{definition}

Recall that the spectral norm $\norm{H} = \max_{\ket{\psi}} |\bra{\psi} H \ket{\psi}| = \max_{\rho} | \Tr(H \rho) |$.
In this section, we will use the approximate optimization algorithm to derive several norm inequalities relating the Pauli-$p$ norm $\norm{\cdot}_{\mathrm{Pauli}, p}$ to the spectral norm $\norm{\cdot}$ for common classes of observables.

We begin with a well-known fact that equates the Frobenius norm and the Pauli-$2$ norm.
This proposition follows directly from the orthonormality of the Pauli observables $\{I, X, Y, Z\}^{\otimes n}$.

\begin{proposition}[Frobenius norm] \label{prop:basic-Pauli2}
    Given any $n$-qubit Hermitian operator $H$. We have
    \begin{equation}
        \frac{1}{\sqrt{2^n}} \norm{H}_F = \norm{H}_{\mathrm{Pauli}, 2} \leq \norm{H}.
    \end{equation}
\end{proposition}
\begin{proof}
    Let $n$ be the number of qubits $H$ act on and $\lambda_1, \ldots, \lambda_{2^n}$ be the eigenvalues of $O$.
    From the fact that $\Tr(P Q) = 2^n \delta_{P=Q}$, we have
    \begin{equation}
        \norm{H}_F^2 = \Tr(H^2) = \sum_{P} |\alpha_P|^2 2^n = 2^n \norm{H}_{\mathrm{Pauli}, 2}^2.
    \end{equation}
    Since $\norm{H}_F^2 = \sum_{i=1}^{2^n} |\lambda_i|^2 \leq 2^n \max_i |\lambda_i|^2 = 2^n \norm{H}^2_\infty$, we have $\sum_{P} |\alpha_P|^2 = \norm{H}_F^2 / 2^n \leq \norm{H}^2_\infty$.
\end{proof}

We now utilize Theorem~\ref{thm:main} to obtain the following useful norm inequality.

\begin{theorem}[Norm inequality from Theorem~\ref{thm:main}]\label{thm:main-norm}
    Given an $n$-qubit $k$-local Hamiltonian $H$ with expansion coefficient/dimension $c_e, d_e$.
    Let $r = 2 d_e / (d_e + 1) \in [1, 2)$. We have
    \begin{equation}
        \frac{1}{3} C(c_e, d_e, k) \norm{H}_{\mathrm{Pauli}, r} \leq \norm{H},
    \end{equation}
    where $C(c_e, d_e, k) = \frac{\sqrt{2(k!)}}{c_e^{1/(2d_e)} k^{k+1.5+1/r} (\sqrt{6} + 2 \sqrt{3} )^k}$ is the same as Theorem~\ref{thm:main}.
\end{theorem}
\begin{proof}
    Consider the Pauli representation $H = \sum_{P: |P| \leq k} \alpha_P P$.
    If we consider $\rho = I / 2^n$, then we have
    \begin{equation}
        \norm{H} \geq \left| \Tr(H) / 2^n \right| \geq \left| \E_{\ket{\phi}: \mathrm{Haar}} \big[ \bra{\phi} H \ket{\phi} \big] \right| = \left| \alpha_I \right|.
    \end{equation}
    If we consider the random product state $\ket{\psi}$ from Theorem~\ref{thm:main}, then we have
    \begin{equation}
        \E_{\ket{\psi}} \left| \bra{\psi} H \ket{\psi} - \E_{\ket{\phi}: \mathrm{Haar}} \big[ \bra{\phi} H \ket{\phi} \big] \right| \geq C(c_e, d_e, k) \left(\sum_{P \neq I} |\alpha_P|^r\right)^{1/r}.
    \end{equation}
    Using $\E_{\ket{\phi}: \mathrm{Haar}} \big[ \bra{\phi} H \ket{\phi} \big] = \alpha_I$ and $\E_{\ket{\psi}} \left| \bra{\psi} H \ket{\psi} - \alpha_I \right| \leq \E_{\ket{\psi}} \left| \bra{\psi} H \ket{\psi} \right| + | \alpha_I | $, we have
    \begin{equation}
        \norm{H} \geq \E_{\ket{\psi}} \left| \bra{\psi} H \ket{\psi} \right| \geq C(c_e, d_e, k) \left(\sum_{P \neq I} |\alpha_P|^r\right)^{1/r} - \left| \alpha_I \right|.
    \end{equation}
    Next, we utilize the following inequality
    \begin{equation}
    \max(x_1, c x_2 - x_1) \geq \frac{c}{c + 2} (x_1 + x_2), \forall x_1, x_2, c \geq 0,
    \end{equation}
    which can be shown by considering the two cases: $x_1 \geq (c / 2) x_2$ and $x_1 < (c / 2) x_2$,
    as well as the lower bounds on $\norm{H}$ to show that
    \begin{equation}
        \norm{H} \geq \frac{C(c_e, d_e, k)}{C(c_e, d_e, k) + 2} \left( \left| \alpha_I \right| +  \left(\sum_{P \neq I} |\alpha_P|^r\right)^{1/r} \right) \geq \frac{C(c_e, d_e, k)}{3} \left( \left| \alpha_I \right| +  \left(\sum_{P \neq I} |\alpha_P|^r\right)^{1/r} \right).
    \end{equation}
    The second inequality uses $k, c_e, d_e \geq 1$, which implies $C(c_e, d_e, k) \in [0, 1]$.
    Finally, the inequality
    \begin{equation}
        \left| \alpha_I \right| +  \left(\sum_{P \neq I} |\alpha_P|^r\right)^{1/r} \geq  \left(\sum_{P} |\alpha_P|^r\right)^{1/r},
    \end{equation}
    can be used to establish the claim.
\end{proof}

Using Fact~\ref{fact:exp-klocal} and Fact~\ref{fact:exp-bounded} that characterize the expansion property for general $k$-local Hamiltonians and bounded degree $k$-local Hamiltonians (i.e., each qubit is acted on by at most $d$ of the $k$-qubit observables), we can establish the following corollaries.

\begin{corollary}[Norm inequality for $k$-local Hamiltonian]\label{cor:any-norm}
    Given an $n$-qubit $k$-local Hamiltonian $H$. We have
    \begin{equation}
        \frac{1}{3} C(k) \norm{H}_{\mathrm{Pauli}, \frac{2k}{k + 1}} \leq \norm{H},
    \end{equation}
    where $C(k) = \frac{\sqrt{2(k!)}}{2 k^{k+1.5+ (k+1)/(2k)} (\sqrt{6} + 2 \sqrt{3} )^k}$ is the same as Corollary~\ref{cor:opt-any}.
\end{corollary}

\begin{corollary}[Norm inequality for bounded-degree Hamiltonian]\label{cor:bounded-norm}
    Given an $n$-qubit $k$-local Hamiltonian $H$ with a bounded degree $d$. We have
    \begin{equation}
        \frac{1}{3} C(k, d) \norm{H}_{\mathrm{Pauli}, 1} \leq \norm{H},
    \end{equation}
    where $C(k, d) = \frac{\sqrt{2(k!)}}{\sqrt{d} k^{k+2.5} (2\sqrt{6} + 4 \sqrt{3} )^k}$.
\end{corollary}

\section{Sample-optimal algorithms for predicting bounded-degree observables}
\label{sec:samp-opt-bd-local-H}

In this section, we consider one of the most basic learning problems in quantum information theory: predicting properties of an unknown $n$-qubit state $\rho$.
This has been studied extensively in the literature on shadow tomography \cite{aaronson2018shadow, aaronson2019gentle} and classical shadows \cite{huang2020predicting}.

\subsection{Review of classical shadow formalism}

We recall the following definition and theorem from classical shadow tomography~\cite{huang2020predicting} based on randomized Pauli measurements.
Each randomized Pauli measurement is performed on a single copy of $\rho$ and measures each qubit of $\rho$ in a random Pauli basis ($X, Y$, or $Z$).

\begin{definition}[Shadow norm from randomized Pauli measurements]
    Given an $n$-qubit observable $O$. Let $\mathcal{U}$ be the distribution over the tensor product of $n$ single-qubit random Clifford unitary, and $\mathcal{M}^{-1}_P = \bigotimes_{i=1}^n \mathcal{M}^{-1}_1$ with $\mathcal{M}^{-1}_1(A) = 3 A - \Tr(A) I$. The shadow norm of $O$ is defined as
    \begin{equation}
        \norm{O}_{\mathrm{shadow}} = \max_{\sigma: \mathrm{state}} \left( \E_{U \sim \mathcal{U}} \sum_{b \in \{0, 1\}^n} \bra{b} U \sigma U^\dagger \ket{b} \bra{b} U \mathcal{M}^{-1}_P(O) U^\dagger \ket{b}^2 \right)^{1/2}.
    \end{equation}
\end{definition}

\begin{theorem}[Classical shadow tomography using randomized Pauli measurements \cite{huang2020predicting}] \label{thm:classical-shadow}
    Given an unknown $n$-qubit state $\rho$ and $M$ observables $O_1, \ldots, O_M$ with $B_{\mathrm{shadow}} = \max_{i \in [M]} \norm{O_i}_{\mathrm{shadow}}$.
    After $N$ randomized Pauli measurements on copies of $\rho$ satisfying
    \begin{equation}
        N = \mathcal{O}\left( \frac{\log(M) B_{\mathrm{shadow}}^2}{\epsilon^2} \right),
    \end{equation}
    we can estimate $\Tr(O_i \rho)$ to $\epsilon$ error for all $i \in [M]$ with high probability.
\end{theorem}

We can see that the sample complexity for predicting many properties of an unknown quantum state $\rho$ depends on the shadow norm $\norm{\cdot}_{\mathrm{shadow}}$.
The larger $\norm{\cdot}_{\mathrm{shadow}}$ is, the more experiments is needed to estimate properties of $\rho$ accurately.
From the original classical shadow paper \cite{huang2020predicting}, we can obtain the following shadow norm bounds for Pauli observables and for few-body observables.

\begin{lemma}[Shadow norm for Pauli observables \cite{huang2020predicting}] \label{lem:shadownorm-Pauli}
    For any $P \in \{I, X, Y, Z\}^{\otimes n}$, we have
    \begin{equation}
        \norm{P}_{\mathrm{shadow}} = 3^{|P| / 2}.
    \end{equation}
\end{lemma}

\begin{lemma}[Shadow norm for few-body observables \cite{huang2020predicting}]
    For any observable $O$ that acts nontrivially on at most $k$ qubits, we have
    \begin{equation}
        \norm{O}_{\mathrm{shadow}} \leq 2^{k} \norm{O}.
    \end{equation}
\end{lemma}

\noindent Combining the above lemmas and Theorem~\ref{thm:classical-shadow}, we can see that Pauli observables and few-body observables can both be predicted efficiently under very few number of randomized Pauli measurements.

\subsection{Upper bound for predicting bounded-degree observables}

Consider an $n$-qubit observable $O$ given as a sum of $k$-qubit observables $O = \sum_{j} O_j$, where each qubit is acted on by at most $d$ of these $k$-qubit observables $_Oj$.
We focus on $k = \mathcal{O}(1)$ and $d = \mathcal{O}(1)$, and refer to such an observable as a bounded-degree observable.
These bounded-degree observables arise frequently in quantum many-body physics and quantum information. For example, the Hamiltonian in a quantum spin system can often be described by a geometrically-local Hamiltonian, which is an instance of bounded-degree observables.
For these observables, the shadow norm is related to the Pauli-$1$ norm of the observable,
\begin{equation}
    \norm{O}_{\mathrm{shadow}} \leq \sum_{P: |P| \leq k} |\alpha_P| \norm{P}_{\mathrm{shadow}} \leq 3^{k/2} \sum_{P: |P| \leq k} |\alpha_P| = 3^{k/2} \norm{O}_{\mathrm{pauli}, 1}.
\end{equation}
If we consider the norm inequality between $\ell_1$-norm and $\ell_2$-norm and use the standard result relating Frobenius norm and spectral norm (Proposition~\ref{prop:basic-Pauli2}), we would obtain the following upper bound on shadow norm.
\begin{equation}
    \norm{O}_{\mathrm{shadow}} \leq 3^{k/2} \norm{O}_{\mathrm{pauli}, 1} \leq (2 \sqrt{3})^k \sqrt{n d} \norm{O}_{\mathrm{pauli}, 2} = \mathcal{O}\left( \sqrt{n} \norm{O} \right).
\end{equation}
Using Theorem~\ref{thm:classical-shadow}, this shadow norm bound would give rise to a number of measurements scaling as
\begin{equation} \label{eq:sample-comp-nlogM}
    N = \mathcal{O}\left( \frac{n \log(M) B^2_\infty}{\epsilon^2} \right),
\end{equation}
where $B_\infty = \max_{i \in [M]} \norm{O_i}_{\infty}$ is an upper bound on the spectral norm $\norm{\cdot}$.
Due to the linear dependence on the number $n$ of qubits in the unknown quantum state, this scaling is not ideal. Furthermore, we will later show that this scaling is actually far from optimal.

To improve the sample complexity, we will use the improved approximate optimization algorithm presented in Appendix~\ref{sec:optimize-klocal}, and the corresponding norm inequality presented in Appendix~\ref{sec:norm-ineq-Pauli}.
Using the norm inequality relating Pauli-$1$ norm and the spectral norm (Corollary~\ref{cor:bounded-norm}), we can obtain the following shadow norm bound.

\begin{lemma}[Shadow norm for bounded-degree observables] \label{lem:shadownorm-bounded}
    Given $k, d = \mathcal{O}(1)$ and an $n$-qubit observable $O$ that is a sum of $k$-qubit observables, where each qubit is acted on by at most $d$ of these $k$-qubit observables.
    \begin{equation}
         \norm{O}_{\mathrm{shadow}} \leq C \norm{O},
    \end{equation}
    for some constant $C > 0$.
\end{lemma}

Combining the above lemma with Theorem~\ref{thm:classical-shadow} allows us to establish the following theorem.
Comparing to Eq.~\eqref{eq:sample-comp-nlogM}, the following theorem uses $n$ times fewer measurements.

\begin{theorem}[Classical shadow tomography for bounded-degree observables] \label{cor:shadow-bounded}
    Given an unknown $n$-qubit state~$\rho$ and $M$ observables $O_1, \ldots, O_M$ with $B_\infty = \max_{i} \norm{O_i}_{\infty}$.
    Suppose each observable $O_i$ is a sum of few-body observables $O_i = \sum_{j} O_{ij}$, where every qubit is acted on by a constant number of the few-body observables $O_{ij}$.
    After $N$ randomized Pauli measurements on copies of $\rho$ with
    \begin{equation}
        N = \mathcal{O}\left( \frac{\log\big( \min( M, n) \big) B_\infty^2}{\epsilon^2} \right),
    \end{equation}
    we can estimate $\Tr(O_i \rho)$ to $\epsilon$ error for all $i \in [M]$ with high probability.
\end{theorem}
\begin{proof}
    The upper bound of $N = \mathcal{O}\left( \log(M) \max_{i \in [M]} \norm{O_i}_{\infty}^2 / \epsilon^2 \right)$ follows immediately from Theorem~\ref{thm:classical-shadow} and Lemma~\ref{lem:shadownorm-bounded}.
    We can also establish an upper bound of $N = \mathcal{O}\left( \log(n) \max_{i \in [M]} \norm{O_i}_{\infty}^2 / \epsilon^2 \right)$.
    To see this, consider the task of predicting all $k$-qubit Pauli observables $P \in \{I, X, Y, Z\}^{\otimes n}$ with $|P| \leq k$.
    There are at most $\mathcal{O}(n^k)$ such Pauli observables.
    To predict all of the $k$-qubit Pauli observables to $\epsilon'$ error under the unknown state $\rho$, we can combine Theorem~\ref{thm:classical-shadow} and Lemma~\ref{lem:shadownorm-Pauli} to see that we only need
    \begin{equation}
        N = \mathcal{O}\left( \log(n) \max_{i \in [M]} \norm{O_i}_{\infty}^2 / (\epsilon')^2 \right)
    \end{equation}
    randomized Pauli measurements.
    Now, given any observable $O_i = \sum_P \alpha_P P$ that is a sum of few-body observables $O_i = \sum_{j} O_{ij}$, where every qubit is acted on by a constant number of the few-body observables $O_{ij}$, we can predict $\Tr(O_i \rho)$ using the following identity
    \begin{equation}
        \Tr(O_i \rho) = \sum_{P: |P| \leq k} \alpha_P \Tr(P \rho),
    \end{equation}
    which incurs a prediction error of at most $\sum_{P} |\alpha_P| \epsilon'$.
    Using the norm inequality in Corollary~\ref{cor:bounded-norm}, we have
    \begin{equation}
         \norm{O_i}_{\mathrm{Pauli}, 1} = \sum_{P} |\alpha_P| \leq C \norm{O_i},
    \end{equation}
    for a constant $C$. Hence, by setting $\epsilon' = \epsilon / C$, we can predict $O_i$ to $\epsilon$ error.
    Thus we can also establish an upper bound of $N = \mathcal{O}\left( \log(n) \max_{i \in [M]} \norm{O_i}_{\infty}^2 / \epsilon^2 \right)$.
    The claim follows by considering the corresponding prediction algorithm (use the standard classical shadow when $M < n$, and use the above algorithm when $M \geq n$).
\end{proof}

\subsection{Optimality of Theorem~\ref{cor:shadow-bounded}}

Here we prove the following lower bound on the sample complexity of shadow tomography for bounded-degree observables demonstrating that Theorem~\ref{cor:shadow-bounded} is optimal.
The optimality holds even when we considered collective measurement procedure on many copies of $\rho$.
This is in stark contrast to other sets of observables, such as the collection of high-weight Pauli observables, where single-copy measurements (e.g., classical shadow tomography) require exponentially more copies than collective measurements.

\begin{theorem}[Lower bound for predicting bounded-degree observables] \label{thm:shadow_lbd}
    Consider the following task.
    Given any unknown $n$-qubit state $\rho$ and any $M$ observables $O_1,\ldots,O_M$ with $B_\infty = \max_{i} \norm{O_i}$.
    Each observable $O_i$ is a sum of few-body observables $O_i = \sum_{j} O_{ij}$, where every qubit is acted on by a constant number of the few-body observables $O_{ij}$.
    We would like to estimate $\Tr(O_i\rho)$ to $\epsilon$ error for all $i\in[M]$ with high probability by performing arbitrary collective measurements on $N$ copies of $\rho$.
    The number of copies needs to be at least
    \begin{equation}
        N = \Omega\left( \frac{\log\big( \min( M, n) \big) B_\infty^2}{\epsilon^2} \right),
    \end{equation}
    for any algorithm to succeed in this task.
\end{theorem}

To show Theorem~\ref{thm:shadow_lbd}, we show a lower bound for the following \emph{distinguishing task}, from which the lower bound for shadow tomography will follow readily. Given $i\in[n]$, let $P_i$ denote the $n$-body Pauli operator that acts as $Z$ on the $i$-th qubit and trivially elsewhere, and define the mixed state
\begin{equation}
    \rho^i \triangleq \frac{1}{2^n}\left(I + \frac{\epsilon}{B_\infty} \cdot P_i \right).
\end{equation}
We will show a lower bound for distinguishing whether $\rho$ is maximally mixed or of the form $\rho^i$ for some $i$.

\begin{lemma}[Lower bound for a distinguishing task] \label{lem:rhoit_lbd}
    Let $0 \le \epsilon \le 1$ and $\delta \ge 2\epsilon$. Let $\mathcal{A}$ be an algorithm that, given access to $N$ copies of a mixed state $\rho$ which is either the maximally mixed state or $\rho^{i}$ for some $i\in[\min(M,n)]$, correctly determines whether or not $\rho$ is maximally mixed with probability at least $3/4$. Then $N = \Omega(\log(\min(M, n)) B_\infty^2 /\epsilon^2)$.
\end{lemma}

\begin{proof}[Proof of Theorem~\ref{thm:shadow_lbd}]
    Let $\calA$ be an algorithm that solves the task in Theorem~\ref{thm:shadow_lbd} to error $\epsilon/3$. We can use this to give an algorithm for the task in Lemma~\ref{lem:rhoit_lbd}: applying $\calA$ to the following $\min(M,n)$ observables,
    \begin{equation}
        O_1 \triangleq B_\infty P_1, \quad \ldots, \quad O_{\min(M,n)} \triangleq B_\infty P_{\min(M,n)},
    \end{equation}
    we can produce $\epsilon/3$-accurate estimates for $\Tr(\rho P_j)$ for all $j\in[\min(M,n)]$. Note that if $\rho$ is maximally mixed, $\Tr(\rho O_j) = 0$ for all $j$, whereas if $\rho = \rho^i$, then $\Tr(\rho O_j) = \epsilon\cdot \mathds{1}[i = j]$. In particular, by checking whether there is a $j$ for which $\Tr(\rho P_j) > 2\epsilon/3$, we can determine whether $\rho$ is maximally mixed or equal to some $\rho^i$. The lower bound in Lemma~\ref{lem:rhoit_lbd} thus implies the lower bound in Theorem~\ref{thm:shadow_lbd}.
\end{proof}

For convenience, define $n' \triangleq \min(M,n)$. Note that for any $i\in[n]$, $(\rho^i)^{\otimes N}$ is diagonal, so we can assume without loss of generality that $\mathcal{A}$ simply makes $N$ independent measurements in the computational basis. Proving Lemma~\ref{lem:rhoit_lbd} thus amounts to showing a lower bound for a classical distribution testing task.

Note that the distribution $\pi^{i}$ over outcomes of a single measurement of $\rho^{i}$ in the computational basis places
\begin{equation}
    \frac{1 + (-1)^{x_i} \epsilon}{2^n}
\end{equation} mass on each string $x\in\{0,1\}^n$. The distribution $\pi$ over outcomes of a single measurement of the maximally mixed state in the computational basis is uniform over all strings $x\in\{0,1\}^n$. The following basic result in binary hypothesis testing lets us reduce proving Lemma~\ref{lem:rhoit_lbd} to upper bounding
\begin{equation}
    d_{\rm TV}\left(\E_{i}[(\pi^{i})^{\otimes N}], \pi^{\otimes N}\right). \label{eq:mixture_tv}
\end{equation}

\begin{lemma}[Le Cam's two-point method \cite{LeCam2pt}] \label{lem:lecam}
    Let $p_0,p_1$ be distributions over a domain $\Omega$ for which there exists a distribution $D$ such that $d_{\rm TV}(p_0,p_1) < 1/3$. Then there is no algorithm $\calA$ that maps elements of $\Omega$ to $\{0,1\}$ for which $\Pr_{x\sim p_i}[\calA(x) = i]\ge 2/3$ for both $i = 0,1$.
\end{lemma}

\begin{proof}[Proof of Lemma~\ref{lem:rhoit_lbd}]
    To bound the expression in Eq.~\eqref{eq:mixture_tv}, it suffices to bound the chi-squared divergence $\chi^2(\E_{i}[(\pi^{i})^{\otimes N}] \| \pi^{\otimes N})$ because for any distributions $p,q$, we have $d_{\rm TV}(p,q) \le 2\sqrt{\chi^2(p\|q)}$.
    For convenience, let us define the likelihood ratio perturbation
    \begin{equation}
        \eta^{i}(x) \triangleq \frac{\mathrm{d} \pi^{i}}{\mathrm{d} \pi}(x) - 1 = (-1)^{x_i} \epsilon
    \end{equation}
    and observe that for any $i,j\in[n]$,
    \begin{equation}
        \E_{x\sim\pi}[\eta^{i}(x) \cdot \eta^{j}(x)] = \epsilon^2 \cdot \mathds{1}[i=j].
    \end{equation}
    Also given strings $x^1,\ldots,x^N\in\{0,1\}^n$ and $S\subseteq[N]$, denote
    \begin{equation}
        \eta^{i}(x^S) \triangleq \prod_{j\in S} \eta^{i}(x_j).
    \end{equation}
    We then have the standard calculation, see e.g. \cite[Lemma 22.1]{wu2017lecture}:
    \begin{align}
        1 + \chi^2\left(\E_{i\sim[n']}[(\pi^{i})^{\otimes N}] \| \pi^{\otimes N} \right) &= \E_{x^1,\ldots,x^N\sim\pi^{\otimes N}}\left[\E_{i\sim[n']}\left[\prod^N_{j=1} (1 + \eta^{i,t}(x^j))\right]^2\right] \\
        &= \E_{i,i'\sim[n']}\left[\E_{x^1,\ldots,x^N\sim\pi^{\otimes N}} \left[\sum_{S,T\subseteq[N]} \eta^{i}(x^S)\eta^{i'}(x^T)\right]\right] \\
        &= \E_{i,i'\sim[n']}\left[\E_{x^1,\ldots,x^N\sim\pi^{\otimes N}} \left[\sum_{S\subseteq[N]} \eta^i(x^S)\eta^{i'}(x^S)\right]\right] \\
        &= \E_{i,i'\sim[n']}\left[\E_{x^1,\ldots,x^N\sim\pi^{\otimes N}} \left[\prod^N_{j=1} (1 + \eta^{i}(x^j)\eta^{i'}(x^j))\right]\right] \\
        &= \E_{i,i'\sim[n']}\left[(1 + \E_{x\sim\pi}[\eta^{i}(x)\eta^{i'}(x)])^N\right] \\
        &= \frac{1}{n'}(1 + \epsilon^2)^N + \frac{n'-1}{n'}
    \end{align}
    We conclude that
    \begin{equation}
        \chi^2(\E_{i\sim[n']}[(\pi^{i})^{\otimes N}] \| \pi^{\otimes N}) \le \frac{1}{n'}((1 + \epsilon^2)^N - 1),
    \end{equation}
    so for $N = c\log(n')/\epsilon^2$ for sufficiently small constant $c > 0$, this quantity is less than $1/3$. By applying Lemma~\ref{lem:lecam} to $p_0 = \pi^{\otimes N}$ and $p_1 = \E_{i\sim[n']}[(\pi^i)^{\otimes N}]$, we obtain the claimed lower bound.
\end{proof}

\section{Learning to predict an unknown observable}
\label{sec:learn-unk-obs}

We begin with a definition of invariance for distribution over quantum states.

\begin{definition}[Invariance under a unitary]
    A probability distribution $\mathcal{D}$ over quantum states is invariant under a unitary $U$ if the probability density remains unchanged after the action of $U$, i.e.,
    \begin{equation}
        f_{\mathcal{D}}(\rho) = f_{\mathcal{D}}(U \rho U^\dagger)
    \end{equation}
    for any state $\rho$.
\end{definition}

In this section, we will utilize the norm inequalities in Section~\ref{sec:norm-ineq-Pauli} to give a learning algorithm that achieves the following guarantee.
The learning algorithm can learn any unknown $n$-qubit observable $O^{\mathrm{(unk)}}$ even if the scale $\norm{O^{\mathrm{(unk)}}}$ is unknown.
The mean squared error $ \E_{\rho \sim \mathcal{D}} \left| h(\rho) - \Tr\left(O^{\mathrm{(unk)}} \rho\right) \right|^2$ scales quadratically with the scale of the unknown observable $O^{\mathrm{(unk)}}$.
We can see that the sample complexity $N$ has a quasi-polynomial dependence on the error $\epsilon, \epsilon'$ relative to the scale of the unknown observable $O^{\mathrm{(unk)}}$, and only depends on the system size $n$ and the failure probability $\delta$ logarithmically.

\begin{theorem}[Learning to predict an unknown observable] \label{thm:main-learning}
    Given $n, \epsilon, \epsilon', \delta > 0$.
    Consider any unknown $n$-qubit observable $O^{\mathrm{(unk)}} = \sum_{P} \alpha_P P$ and any unknown $n$-qubit state distribution $\mathcal{D}$ that is invariant under single-qubit $H$ and $S$ gates.
    Given training data $\{ \rho_\ell, \Tr(O^{\mathrm{(unk)}} \rho_\ell)) \}_{\ell=1}^N$ of size
    \begin{equation}
        N = \log\left(\frac{n}{\delta}\right) \, \min\left( 2^{\mathcal{O}\left(\log(\frac{1}{\epsilon}) \left( \log\log(\frac{1}{\epsilon}) +  \log(\frac{1}{\epsilon'})\right)\right)}, 2^{\mathcal{O}(\log(\frac{1}{\epsilon}) \log(n))} \right).
    \end{equation}
    Let $k = \lceil \log_{1.5}(1/\epsilon) \rceil$, $O^{\mathrm{(low)}} = \sum_{|P| \leq k} \alpha_P P$ be the low-degree approximation of $O^{\mathrm{(unk)}}$, and $r = \tfrac{2k}{k + 1} \in [1, 2)$.
    The algorithm can learn a function $h(\rho) = \max(-\hat{\Theta}, \min(\hat{\Theta}, \Tr(\hat{O} \rho)))$ for an observable $\hat{O}$ and a real number~$\hat{\Theta}$ that achieves a prediction error
    \begin{equation}
        \E_{\rho \sim \mathcal{D}} \left| h(\rho) - \Tr\left(O^{\mathrm{(unk)}} \rho\right) \right|^2 \leq
        \left( \epsilon + \epsilon'  \left[ \, 1 + \left(\tfrac{\norm{O^{\mathrm{(low)}}}}{\norm{O^{\mathrm{(unk)}}}}\right)^r \, \right] \right) \norm{O^{\mathrm{(unk)}}}^2 \label{eq:main-learn-bound}
    \end{equation}
    with probability at least $1 - \delta$.
\end{theorem}

\subsection{Low-degree approximation under mean squared error}
\label{sec:low-deg-MSE}

In order to characterize the mean squared error $\E_{\rho \sim \mathcal{D}} \Tr(O_1 \rho) - \Tr(O_2) \rho$ between two observables $O_1, O_2$, we need the following definition of a modified purity for quantum states.

\begin{definition}[Non-identity purity] \label{def:nonI-purity}
    Given a $k$-qubit state $\rho$.
    The non-identity purity of $\rho$ is
    \begin{equation}
        \gamma^{\star}(\rho) \triangleq \frac{1}{2^k} \sum_{Q \in \{X, Y, Z\}^{\otimes k}} \Tr(Q \rho)^2.
    \end{equation}
    Non-identity purity is bounded by purity, $\gamma^{\star}(\rho) \leq \gamma(\rho) = \Tr(\rho^2) = \frac{1}{2^k} \sum_{Q \in \{I, X, Y, Z\}^{\otimes k}} \Tr(Q \rho)^2$.
\end{definition}

\begin{lemma}[Mean squared error] \label{lem:mse-D}
    Given two $n$-qubit observables $O_1, O_2$ with
    \begin{equation}
    O_1 - O_2 = \sum_{P \in \{I, X, Y, Z\}^{\otimes n}} \Delta\alpha_P P,
    \end{equation}
    and a distribution $\mathcal{D}$ over quantum states that is invariant under single-qubit $H$ and $S$ gates.
    We have
    \begin{equation}
        \E_{\rho \sim \mathcal{D}} \left| \Tr(O_1 \rho) - \Tr(O_2 \rho) \right|^2 = \sum_{P \in \{I, X, Y, Z\}^{\otimes n}} \E_{\rho \sim \mathcal{D}} \left[ \gamma^\star\left(\rho_{\mathsf{dom}(P)}\right) \right] \left(\frac{2}{3}\right)^{|P|} \left| \Delta\alpha_P \right|^2.
    \end{equation}
\end{lemma}
\begin{proof}
    Consider $U_1, \ldots, U_n$ to be independent random single-qubit Clifford unitaries.
    Because $\mathcal{D}$ is invariant under single-qubit Hadamard and phase gates, $\mathcal{D}$ is invariant under any tensor product of single-qubit Clifford unitaries.
    This implies that the distribution of the random state $\rho$ is the same as the distribution of the random state $(U_1 \otimes \ldots \otimes U_n) \rho (U_1 \otimes \ldots \otimes U_n)^\dagger$.
    Using this fact, we expand the mean squared error as
    \begin{align}
        &\E_{\rho \sim \mathcal{D}} \left| \Tr(O_1 \rho) - \Tr(O_2 \rho) \right|^2\\
        &= \E_{\rho \sim \mathcal{D}} \E_{U_1, \ldots, U_n} \sum_{P, Q \in \{I, X, Y, Z\}^{\otimes n}} \Delta \alpha_P \Delta \alpha_Q \Tr\left(\left(\bigotimes_{i=1}^n U_i^\dagger P_i U_i \right) \otimes \left(\bigotimes_{i=1}^n U_i^\dagger Q_i U_i \right) (\rho \otimes \rho) \right).
    \end{align}
    Using the unitary $2$-design property of random Clifford unitary and $ \mathrm{SWAP} = \frac{1}{2} \sum_{P \in \{I, X, Y, Z\}} P \otimes P$, we have
    \begin{equation} \label{eq:2design-Pauli}
       \E_{U_i} \big[ U_i^\dagger P_i U_i \otimes U_i^\dagger Q_i U_i \big] =
       \begin{cases}
            I \otimes I, & P_i = Q_i = I, \\
            \frac{1}{3}\left(X \otimes X + Y \otimes Y + Z \otimes Z \right), & P_i = Q_i \neq I, \\
            0, & P_i \neq Q_i.
       \end{cases}
    \end{equation}
    We can now write the target value as
    \begin{equation} \label{eq:exp-Orho}
        \E_{\rho \sim \mathcal{D}} \left| \Tr(O_1 \rho) - \Tr(O_2 \rho) \right|^2 = \E_{\rho \sim \mathcal{D}} \sum_{P \in \{I, X, Y, Z\}^{\otimes n}} \frac{1}{3^{|P|}} \left|\Delta \alpha_P\right|^2 \sum_{Q \in \{X, Y, Z\}^{\otimes |P|}} \Tr(Q \rho_{\mathsf{dom}(P)})^2.
    \end{equation}
    The claim follows from Definition~\ref{def:nonI-purity} on non-identity purity $\gamma^\star$.
\end{proof}

The following lemma tells us that the mean absolute error can be upper bounded by the root mean squared error.
Hence, both the mean absolute error and the mean squared error are characterized by the $\ell_2$ distance between the Pauli coefficients (as well as the average non-identity purity).
Due to the following relation, we will focus on the mean squared error throughout the text.

\begin{lemma}[Mean absolute error] \label{lem:mae-D}
    Given two $n$-qubit observables $O_1, O_2$ with
    \begin{equation}
    O_1 - O_2 = \sum_{P \in \{I, X, Y, Z\}^{\otimes n}} \Delta\alpha_P P,
    \end{equation}
    and a distribution $\mathcal{D}$ over quantum states that is invariant under single-qubit $H$ and $S$ gates.
    We have
    \begin{equation}
        \E_{\rho \sim \mathcal{D}} \left| \Tr(O_1 \rho) - \Tr(O_2 \rho) \right| \leq \left(\sum_{P \in \{I, X, Y, Z\}^{\otimes n}} \E_{\rho \sim \mathcal{D}} \left[ \gamma^\star\left(\rho_{\mathsf{dom}(P)}\right) \right] \left(\frac{2}{3}\right)^{|P|} \left| \Delta\alpha_P \right|^2\right)^{1/2}.
    \end{equation}
\end{lemma}
\begin{proof}
Jensen's inequality gives
\begin{equation}
    \E_{\rho \sim \mathcal{D}} \left| \Tr(O_1 \rho) - \Tr(O_2 \rho) \right| \leq \left( \E_{\rho \sim \mathcal{D}} \left| \Tr(O_1 \rho) - \Tr(O_2 \rho) \right|^2 \right)^{1/2}.
\end{equation}
Combining with Lemma~\ref{lem:mse-D} yields the stated result.
\end{proof}

From Lemma~\ref{lem:mse-D}, we can construct a low-degree approximation by removing all high-weight Pauli terms for any observable $O$.
The approximation error decays exponentially with the weight of the Pauli terms.

\begin{corollary}[Low-degree approximation] \label{cor:low-deg-approx}
    Given an $n$-qubit observable $O = \sum_{P \in \{I, X, Y, Z\}^{\otimes n}} \alpha_P P$ and a distribution $\mathcal{D}$ over quantum states that is invariant under single-qubit $H$ and $S$ gates.
    For $k > 0$, consider $O^{(k)} = \sum_{P: |P| < k} \alpha_P P$. We have
    \begin{equation}
        \E_{\rho \sim \mathcal{D}} \left| \Tr(O \rho) - \Tr(O^{(k)} \rho) \right|^2 \leq \left(\frac{2}{3}\right)^k \norm{O}^2.
    \end{equation}
\end{corollary}
\begin{proof}
    Using Lemma~\ref{lem:mse-D} and the fact that $\gamma^\star(\varrho) \leq \gamma(\varrho) \leq 1$ for any state $\varrho$, we have
    \begin{equation}
        \E_{\rho \sim \mathcal{D}} \left| \Tr(O \rho) - \Tr(O^{(k)} \rho) \right|^2 \leq \sum_{P: |P| \geq k} \left(\frac{2}{3}\right)^{|P|} |\alpha_P|^2 \leq \left(\frac{2}{3}\right)^k \sum_{P} |\alpha_P|^2.
    \end{equation}
    The norm inequality given in Prop.~\ref{prop:basic-Pauli2} establishes the claim.
\end{proof}

\subsection{Tools for extracting and filtering Pauli coefficients}
\label{sec:Pauli-tools}

In order to learn the low-degree approximation of an arbitrary observable $O$, we need to be able to extract the relevant $\alpha_P$.
Furthermore, we will impose criteria for filtering out uninfluential Pauli observables $P$ to prevent them from increasing the noise and leading to a higher prediction error.

\subsubsection{Extracting Pauli coefficient}

\begin{lemma}[Extracting Pauli coefficient] \label{lem:extract-Pauli}
    Given an $n$-qubit observable $O = \sum_{P \in \{I, X, Y, Z\}^{\otimes n}} \alpha_P P$ and a distribution $\mathcal{D}$ over quantum states that is invariant under single-qubit $H$ and $S$ gates.
    For any Pauli observable $P \in \{I, X, Y, Z\}^{\otimes n}$,
    we have
    \begin{equation}
    \E_{\rho \sim \mathcal{D}} \Tr(O \rho) \Tr(P \rho) = \left(\frac{2}{3}\right)^{|P|} \alpha_P \E_{\rho \sim \mathcal{D}} \gamma^*(\rho_{\mathsf{dom}(P)}).
    \end{equation}
\end{lemma}
\begin{proof}
    Using the invariance of $\mathcal{D}$, we have
    \begin{equation}
        \E_{\rho \sim \mathcal{D}} \Tr(O \rho) \Tr(P \rho) = \E_{\rho \sim \mathcal{D}} \E_{U_1, \ldots, U_n} \sum_{Q \in \{I, X, Y, Z\}^{\otimes n}} \alpha_Q \Tr\left(\left(\bigotimes_{i=1}^n U_i^\dagger P_i U_i \right) \otimes \left(\bigotimes_{i=1}^n U_i^\dagger Q_i U_i \right) (\rho \otimes \rho) \right).
    \end{equation}
    Using Eq.~\eqref{eq:2design-Pauli}, we can rewrite the above expression as
    \begin{equation}
        \E_{\rho \sim \mathcal{D}} \Tr(O \rho) \Tr(P \rho) = \E_{\rho \sim \mathcal{D}} \frac{1}{3^{|P|}} \alpha_P \sum_{Q \in \{X, Y, Z\}^{\otimes |P|}} \Tr(Q \rho_{\mathsf{dom}(P)})^2.
    \end{equation}
    The claim follows from the definition of the non-identity purity $\gamma^*$.
\end{proof}

For each Pauli observable $P \in \{I, X, Y, Z\}^{\otimes n}$, define the quantity we can extract using the lemma to be
\begin{equation}
    x_P = \left(\frac{2}{3}\right)^{|P|} \alpha_P \E_{\rho \sim \mathcal{D}} \gamma^*(\rho_{\mathsf{dom}(P)}).
\end{equation}
We can obtain an estimate $\hat{x}_P$ for $x_P$ by averaging $\Tr(O \rho) \Tr(P \rho)$ over the training data.
However, to obtain an estimate $\hat{\alpha}_P$ for $\alpha_P$, we need to divide $\hat{x}$ by $\left(\frac{2}{3}\right)^{|P|} \E_{\rho \sim \mathcal{D}} \gamma^*(\rho_{\mathsf{dom}(P)})$.
The error in the estimate $\hat{\alpha}_P$ could be arbitrarily large if $\left(\frac{2}{3}\right)^{|P|} \E_{\rho \sim \mathcal{D}} \gamma^*(\rho_{\mathsf{dom}(P)})$ is close to zero.
Hence, we present a filter in Section~\ref{sec:filter-smallweight} to handle this issue.
In addition to this filter, the norm inequalities given in Section~\ref{sec:norm-ineq-Pauli} show that most $\alpha_P$ would be close to zero.
Hence, when $\alpha_P$ is small, we could simply set them to zero to avoid noise build-up.
This gives rise to the second filtering layer given in Section~\ref{sec:filter-uninfluentialPauli}.


\subsubsection{Filtering small weight factor}
\label{sec:filter-smallweight}

The first filter sets the estimate $\hat{\alpha}_P$ to be zero when the average non-identity purity $\E_{\rho \sim \mathcal{D}} \gamma^*(\rho_{\mathsf{dom}(P)})$ is close to zero.
We define the weight factor for a Pauli observable $P$ to be
\begin{equation}
    \beta_P = \left(\frac{2}{3}\right)^{|P|} \E_{\rho \sim \mathcal{D}} \gamma^*(\rho_{\mathsf{dom}(P)}).
\end{equation}
The weight factor $\beta_P$ depends on the distribution $\mathcal{D}$, which may be unknown.
Hence, we can only obtain an estimate $\hat{\beta}_P$ for $\beta_P$ by utilizing the training data.
Recall from Lemma~\ref{lem:extract-Pauli}, we can only obtain an estimate $\hat{x}_P$ for $x_P = \alpha_P \beta_P$.
The mean squared error (Lemma~\ref{lem:mse-D}) shows that the contribution from error in $\hat{\alpha}_P$ is
\begin{equation}
    \beta_P \left| \hat{\alpha}_P - \alpha_P \right|^2.
\end{equation}
The presence of $\beta_P$ in the mean squared error is very useful since it counteracts the fact that we cannot estimate $\hat{\alpha}_P$ accurately when $\beta_P$ is close to zero.
The following lemma shows that estimates for $\beta_P$ and $x_P$ are sufficient to perform filtering and achieve a small mean squared error.

\begin{lemma}[Filtering small weight factor] \label{lem:filter-smallweight}
    Given $\tilde{\epsilon}, \eta > 0$.
    Consider $\alpha \in [-\eta, \eta]$, and $\beta \in [0, 1]$.
    Let $x = \alpha \beta \in [-\eta, \eta]$.
    Given estimates $\hat{x}$ and $\hat{\beta}$ with $|\hat{x} - x| < \eta\tilde{\epsilon}$ and $|\hat{\beta} - \beta| < \tilde{\epsilon}$.
    If we define the estimate
    \begin{equation}
        \hat{\alpha} =
        \begin{cases}
            0, & \hat{\beta} \leq 2 \tilde{\epsilon},\\
            \hat{x} / \hat{\beta}, & \hat{\beta} > 2 \tilde{\epsilon},
        \end{cases}
    \end{equation}
    then we have $\beta |\hat{\alpha} - \alpha|^2 \leq 3 \eta^2 \tilde{\epsilon}$.
\end{lemma}
\begin{proof}
    Consider the first case of $\hat{\beta} \leq 2 \tilde{\epsilon}$.
    We have
    \begin{equation}
        \beta |\hat{\alpha} - \alpha|^2 = \beta \alpha^2 \leq \eta^2 \beta \leq \eta^2 \hat{\beta} + \eta^2 \tilde{\epsilon} \leq 3 \eta^2 \tilde{\epsilon}.
    \end{equation}
    For the second case of $\hat{\beta} > 2 \tilde{\epsilon}$, we have $\beta > \tilde{\epsilon}$. By applying triangle inequality, we have
    \begin{equation}
        \left| \sqrt{\beta} \hat{\alpha} - \sqrt{\beta} \alpha \right| \leq \frac{\sqrt{\beta}}{\hat{\beta}} \left| \hat{x} - x \right| + \left| \sqrt{\beta} x \right| \left| \frac{1}{\hat{\beta}} - \frac{1}{\beta} \right|.
    \end{equation}
    The first term can be bounded as $\frac{\sqrt{\beta}}{\hat{\beta}} \left| \hat{x} - x \right| \leq \eta \frac{\sqrt{\beta}}{\hat{\beta}} \tilde{\epsilon}$.
    The second term can be bounded by the same expression
    \begin{equation}
        \left| \sqrt{\beta} x \right| \left| \frac{1}{\hat{\beta}} - \frac{1}{\beta} \right| = \beta^{3/2} |\alpha| \frac{|\hat{\beta} - \beta|}{\hat{\beta} \beta} \leq \eta \frac{\sqrt{\beta}}{\hat{\beta}} \tilde{\epsilon}.
    \end{equation}
    Using the fact that $\sqrt{z + \tilde{\epsilon}} / z$ is monotonically decreasing for $z > 0$, we have
    \begin{equation}
        \frac{\sqrt{\beta}}{\hat{\beta}} \tilde{\epsilon} \leq \frac{\sqrt{\hat{\beta} + \tilde{\epsilon}}}{\hat{\beta}} \tilde{\epsilon} \leq \sqrt{\frac{3}{4} \tilde{\epsilon}}.
    \end{equation}
    Together, $\left| \sqrt{\beta} \hat{\alpha} - \sqrt{\beta} \alpha \right|^2 \leq 3 \eta^2 \tilde{\epsilon}$ and the claim is established.
\end{proof}

\subsubsection{Filtering uninfluential Pauli observables}
\label{sec:filter-uninfluentialPauli}

Consider a set $S \subseteq \{I, X, Y, Z\}^{\otimes n}$ that contains the Pauli observables of interest.
For example, we will later consider $S$ to be the set of all few-body Pauli observables.
Using the norm inequalities given in Appendix~\ref{sec:norm-ineq-Pauli}, we can filter out more $\alpha_P$ to achieve an improved mean squared error.
Below is the filtering lemma that combines both the filtering of Pauli observables with a small weight factor (Lemma~\ref{lem:filter-smallweight}) and the filtering of those with a small contribution (characterized by $|x_P| / \beta_P^{1/2}$).

\begin{lemma}[Filtering lemma] \label{lem:filtering-full}
    Given $\tilde{\epsilon}, \eta > 0,$ and a set $S \subseteq \{I, X, Y, Z\}^{\otimes n}$.
    Consider $\alpha_P \in [-\eta, \eta]$, $\beta_P \in [0, 1]$, $x_P = \alpha_P \beta_P \in [-\eta, \eta]$ for all $P \in S$.
    Suppose there exists $A > 0$ and $1 \leq r < 2$, such that
    \begin{equation}
        \sum_{P \in S} |\alpha_P|^r \leq A^r.
    \end{equation}
    Given $\hat{x}_P$ and $\hat{\beta}_P$ with $|\hat{x}_P - x_P| < \eta \tilde{\epsilon}$ and $|\hat{\beta}_P - \beta_P| < \tilde{\epsilon}$ for all $P \in S$.
    If we define
    \begin{equation} \label{eq:alphahat-def}
        \hat{\alpha}_P =
        \begin{cases}
            0, & \hat{\beta}_P \leq 2 \tilde{\epsilon},\\
            0, & \hat{\beta}_P > 2 \tilde{\epsilon}, \, |\hat{x}_P| / \hat{\beta}_P^{1/2} \leq 2 \eta \sqrt{\tilde{\epsilon}},\\
            \hat{x}_P / \hat{\beta}_P, & \hat{\beta}_P > 2 \tilde{\epsilon}, \, |\hat{x}_P| / \hat{\beta}_P^{1/2} > 2 \eta \sqrt{\tilde{\epsilon}},
        \end{cases}
    \end{equation}
    then we have $\sum_{P \in S} \beta_P |\hat{\alpha}_P - \alpha_P|^2 \leq 6 A^r \eta^{2-r} \tilde{\epsilon}^{1 - (r/2)}$.
    We also have $\beta_P |\hat{\alpha}_P - \alpha_P|^2 \leq 9 \eta^2 \tilde{\epsilon}, \forall P \in S$.
\end{lemma}
\begin{proof}
    We first define $S^{u} \subseteq S$ to be the set of Pauli observables $P$ with $\hat{\beta}_P > 2 \tilde{\epsilon},  |\hat{x}_P| / \hat{\beta}_P^{1/2} > 2 \eta \sqrt{\tilde{\epsilon}}$.
    The set $S^{u}$ contains all the unfiltered Pauli observables. We define $S^{f}$ to be $S \setminus S^u$, which contains all the filtered Pauli observables.
    We separate the contribution of $S^u$ and $S^f$ in the mean squared error $\sum_{P \in S} \beta_P |\hat{\alpha}_P - \alpha_P|^2$,
    \begin{equation}
        \sum_{P \in S} \beta_P |\hat{\alpha}_P - \alpha_P|^2 = \sum_{P \in S^u} \beta_P |\hat{\alpha}_P - \alpha_P|^2 + \sum_{P \in S^f} \beta_P |\hat{\alpha}_P - \alpha_P|^2.
    \end{equation}
    A key quantity for the analysis is $\beta_P^{1/2} \alpha_P = x_P / \beta_P^{1/2}$.
    For Pauli $P$ with $\hat{\beta}_P \leq 2 \tilde{\epsilon}$, we have
    \begin{equation} \label{eq:smallweightP-bound}
        |\beta_P^{1/2} \alpha_P| \leq \eta \sqrt{\hat{\beta}_P + \tilde{\epsilon}} \leq \eta \sqrt{3 \tilde{\epsilon}}.
    \end{equation}
    For Pauli $P$ with $\hat{\beta}_P > 2 \tilde{\epsilon}$, we have
    \begin{equation} \label{eq:filter-1beta}
        \left| \frac{\hat{x}_P}{\hat{\beta}_P^{1/2}} - \frac{x_P}{\beta_P^{1/2}} \right| \leq \frac{1}{\hat{\beta}_P^{1/2}} \left| \hat{x}_P - x_P \right| + |x_P| \left| \frac{1}{\hat{\beta}_P^{1/2}} - \frac{1}{\beta_P^{1/2}} \right| \leq \eta \sqrt{\frac{\tilde{\epsilon}}{2}} + \eta \left| \frac{\beta_P}{\hat{\beta}_P^{1/2}} - \beta_P^{1/2} \right| \leq \eta \sqrt{\tilde{\epsilon}}.
    \end{equation}
    The last inequality uses the fact that $\beta_P > \tilde{\epsilon}$, $\hat{\beta}_P / \beta_P > 2$, and hence
    \begin{equation} \label{eq:filter-2beta}
        \left| \frac{\beta_P}{\hat{\beta}_P^{1/2}} - \beta_P^{1/2} \right| = \frac{\left| \hat{\beta}_P - \beta_P \right|}{\hat{\beta}_P^{1/2} \left(1 + \left(\frac{\hat{\beta}_P}{\beta_P}\right)^{1/2} \right)} \leq \frac{\sqrt{\tilde{\epsilon}}}{2 + \sqrt{2}}.
    \end{equation}
    We are now ready to analyze the contributions of $S^u$ and $S^f$.

    For the unfiltered Pauli observables (those in set $S^u$), we can use Lemma~\ref{lem:filter-smallweight} to obtain
    \begin{equation}
        \sum_{P \in S^u} \beta_P |\hat{\alpha}_P - \alpha_P|^2 \leq 3 \eta^2 \tilde{\epsilon} |S^u|.
    \end{equation}
    Eq.~\eqref{eq:filter-1beta} shows that for Pauli observable $P$ with $\hat{\beta}_P > 2 \tilde{\epsilon}$ and $|\hat{x}_P| / \hat{\beta}_P^{1/2} > 2 \eta \sqrt{\tilde{\epsilon}}$, we have $|x_P| / \beta_P^{1/2} > 2 \eta \sqrt{\tilde{\epsilon}} - \eta \sqrt{\tilde{\epsilon}}.$
    We will use this fact to bound the size of the set $|S^u|$,
    \begin{equation}
        |S^u| \leq \sum_{P \in S^u} \frac{(|x_P| / \beta_P^{1/2})^r}{\left(2 \eta \sqrt{\tilde{\epsilon}} - \eta \sqrt{\tilde{\epsilon}}\right)^r} = \frac{1}{\eta^r \tilde{\epsilon}^{r/2}}
        \sum_{P \in S^u}| \alpha_P|^r \beta_P^{r/2} \leq \frac{1}{\eta^r\tilde{\epsilon}^{r/2}}
        \sum_{P \in S} |\alpha_P|^r = \frac{A^r}{\eta^r \tilde{\epsilon}^{r/2}}.
    \end{equation}
    Together, we have the following upper bound,
    \begin{equation}
        \sum_{P \in S^u} \beta_P |\hat{\alpha}_P - \alpha_P|^2 \leq 3 \eta^{2-r} A^r \tilde{\epsilon}^{1 - (r/2)}.
    \end{equation}
    For the filtered Pauli observables (those in set $S^f$), we have
    \begin{equation}
        \sum_{P \in S^f} \beta_P |\hat{\alpha}_P - \alpha_P|^2 = \sum_{P \in S^f} \left|\beta_P^{1/2} \alpha_P\right|^r \left|\beta_P^{1/2} \alpha_P\right|^{2-r}.
    \end{equation}
    There are two types of Pauli observables in $S^f$.
    \begin{enumerate}
        \item For $P$ with $\hat{\beta}_P \leq 2 \tilde{\epsilon}$, we have $\left|\beta_P^{1/2} \alpha_P\right| \leq \eta \sqrt{3 \tilde{\epsilon}}$ from Eq.~\eqref{eq:smallweightP-bound}.
        \item For $P$ with $\hat{\beta}_P > 2 \tilde{\epsilon}$ and $ |\hat{x}_P| / \hat{\beta}_P^{1/2} \leq \eta 2 \sqrt{\tilde{\epsilon}}$, we have $\left|\beta_P^{1/2} \alpha_P\right| = |x_P| / \beta_P^{1/2} \leq 2 \eta \sqrt{\tilde{\epsilon}} + \eta \sqrt{\tilde{\epsilon}}$ from Eq.~\eqref{eq:filter-1beta}.
    \end{enumerate}
    Together, we have the following upper bound,
    \begin{equation}
        \sum_{P \in S^f} \beta_P |\hat{\alpha}_P - \alpha_P|^2 \leq (3 \eta \sqrt{\tilde{\epsilon}})^{2-r} \sum_{P \in S^f} \beta_P^{r/2} \left| \alpha_P\right|^r \leq A^r (3 \eta \sqrt{\tilde{\epsilon}})^{2-r} \leq 3 A^r \eta^{2-r} \tilde{\epsilon}^{1 - (r/2)}.
    \end{equation}
    Combining the contribution of $S^u$ and $S^f$ yields
    \begin{equation}
        \sum_{P \in S} \beta_P |\hat{\alpha}_P - \alpha_P|^2 \leq 6 A^r \eta^{2-r} \tilde{\epsilon}^{1 - (r/2)}.
    \end{equation}
    Thus we have established the first statement of the lemma.

    We now focus on the second statement of the lemma. For Pauli observable $P$ that satisfies the first and the third cases of Eq.~\eqref{eq:alphahat-def}, we can use Lemma~\ref{lem:filter-smallweight} to obtain $\beta_P |\hat{\alpha}_P - \alpha_P|^2 \leq 3 \eta^2 \tilde{\epsilon} < 9 \eta^2 \tilde{\epsilon}$.
    For the second case of Eq.~\eqref{eq:alphahat-def}, we can use Eq.~\eqref{eq:filter-1beta} to see that
    \begin{equation}
        \beta_P |\hat{\alpha}_P - \alpha_P|^2 = \left( \frac{x_P}{\beta_P^{1/2}} \right)^2 \leq \left( \frac{|\hat{x}_P|}{\hat{\beta}_P^{1/2}} + \left| \frac{\hat{x}_P}{\hat{\beta}_P^{1/2}} - \frac{x_P}{\beta_P^{1/2}} \right| \right)^2 \leq 9 \eta^2 \tilde{\epsilon}.
    \end{equation}
    Hence, for all $P \in S$, we have $\beta_P |\hat{\alpha}_P - \alpha_P|^2 \leq 9 \eta^2 \tilde{\epsilon}$.
\end{proof}

\subsection{Learning algorithm}
\label{sec:learning-algo}

In this section, we present a learning algorithm satisfying the guarantee given in Theorem~\ref{thm:main-learning}.
Consider the full training data $\{ \rho_\ell, y_\ell = \Tr(O^{\mathrm{(unk)}} \rho_\ell)) \}_{\ell=1}^N$ of size $N$.
The learning algorithm splits the full data into a smaller training set of size $N_{\mathrm{tr}}$ and a validation set of size $N_{\mathrm{val}}$ with $N = N_{\mathrm{tr}} + N_{\mathrm{val}}$.
The training set is used to extract Pauli coefficients and perform filtering with a hyperparameter $\eta$.
The validation set is used to choose the best hyperparameter $\eta$.
We can set $N_{\mathrm{tr}} = (4/5) N$ and $N_{\mathrm{val}} = (1/5) N$.

We consider two slightly different learning algorithms for the sample complexity scaling of
\begin{equation} \label{eq:samp-learning-obs-dual}
N = \log\left(\frac{n}{\delta}\right) \, 2^{\mathcal{O}\left(\log(\frac{1}{\epsilon}) \left( \log\log(\frac{1}{\epsilon}) +  \log(\frac{1}{\epsilon'})\right)\right)}
 \quad \mbox{and} \quad
N = \log\left(\frac{n}{\delta}\right) \,  2^{\mathcal{O}(\log(\frac{1}{\epsilon}) \log(n))}.
\end{equation}
We can simply look at which sample complexity is smaller and select the corresponding learning algorithm.

We begin with the learning algorithm for achieve the sample complexity on the left of Eq.~\eqref{eq:samp-learning-obs-dual}.
First, the algorithm computes the sample maximum over the training set,
\begin{equation}
    \hat{\Theta} = \max_{\ell \in \{1, \ldots, N_{\mathrm{tr}} \} } \left| y_\ell \right| = \max_{\ell \in \{1, \ldots, N_{\mathrm{tr}} \} } \left| \Tr(O^{\mathrm{(unk)}} \rho_\ell)) \right| \leq \norm{O^{\mathrm{(unk)}}}.
\end{equation}
to obtain a scale for the function value.
Let $C(k)$ be the constant from Corollary~\ref{cor:any-norm}.
We define
\begin{equation}
    \tilde{\epsilon} \triangleq \left( \frac{\epsilon'}{12} \right)^{k + 1} \left( \frac{C(k)}{3} \right)^{2k}. \label{eq:tildeeps_def}
\end{equation}
Next, we consider the following grid of hyperparameters,
\begin{equation} \label{eq:eta-grid}
    \eta \in \left\{ 2^0 \hat{\Theta}, 2^1\hat{\Theta}, 2^2 \hat{\Theta} \ldots, 2^R \hat{\Theta} \right\},
\end{equation}
where $R = \log_2 \lceil 1 / \tilde{\epsilon} \rceil$.
For each hyperparameter $\eta$, the learning algorithm runs the following.
The learning algorithm considers every Pauli observable $P \in \{I, X, Y, Z\}^{\otimes n}$ with $|P| \leq \log_{1.5}(1 / \epsilon)$.
We define the set that contains the Pauli observables of interest,
\begin{equation}
    S = \left\{ P: |P| \leq \log_{1.5}(1 / \epsilon) \right\},
\end{equation}
and $k = \lceil \log_{1.5}(1 / \epsilon) \rceil$.
For each $P \in S$, the algorithm computes
\begin{align}
    \hat{x}_P &= \frac{1}{N_{\mathrm{tr}}} \sum_{\ell=1}^{N_{\mathrm{tr}}} \Tr(P \rho_\ell) y_\ell,\\
    \hat{\beta}_P &= \frac{1}{N_{\mathrm{tr}}} \sum_{\ell=1}^{N_{\mathrm{tr}}} \Tr(P \rho_\ell) \Tr(P \rho_\ell),
\end{align}
using the training set $\{(\rho_\ell, y_\ell = \Tr(O^{\mathrm{unk} \rho_\ell}))\}_{\ell=1}^{N_{\mathrm{tr}}}$.
By definition of $\hat{x}_P$ and $\hat{\Theta}$, we have
\begin{equation} \label{eq:hatx-bound}
    \left|\hat{x}_P\right| \leq \hat{\Theta}, \quad \forall P \in S.
\end{equation}
Then, for each $P \in S$, the algorithm computes
\begin{equation}
    \hat{\alpha}_P(\eta) =
    \begin{cases}
        0, & \hat{\beta}_P \leq 2 \tilde{\epsilon},\\
        0, & \hat{\beta}_P > 2 \tilde{\epsilon}, \, |\hat{x}_P| / \hat{\beta}_P^{1/2} \leq 2 \eta \sqrt{\tilde{\epsilon}},\\
        \hat{x}_P / \hat{\beta}_P, & \hat{\beta}_P > 2 \tilde{\epsilon}, \, |\hat{x}_P| / \hat{\beta}_P^{1/2} > 2 \eta \sqrt{\tilde{\epsilon}},
    \end{cases}
\end{equation}
The algorithm considers the function $h(\rho; \eta) = \max(-\hat{\Theta}, \min(\hat{\Theta}, \Tr(\hat{O}(\eta) \rho)))$, where the observable $\hat{O}(\eta)$ is defined as follows,
\begin{equation}
    \hat{O}(\eta) = \sum_{P \in S} \hat{\alpha}_P(\eta) P.
\end{equation}
The best $\eta$ is selected using the validation set,
\begin{equation}
    \eta^* = \argmin_{\eta \in \{2^0 \hat{\Theta}, \ldots, 2^R \hat{\Theta}\}} \frac{1}{N_{\mathrm{val}}}\sum_{\ell = N_{\mathrm{tr}} + 1}^{N_{\mathrm{tr}} + N_{\mathrm{val}}} \left| h(\rho_\ell; \eta) - y_\ell \right|^2.
\end{equation}
The learning algorithm outputs $h(\rho; \eta^*)$ as the learned function.

We now present the learning algorithm for achieving the sample complexity on the right of Eq.~\eqref{eq:samp-learning-obs-dual}.
We define the set that contains the Pauli observables of interest,
\begin{equation}
    S' = \left\{ P: |P| \leq \log_{1.5}(2 / \epsilon) \right\},
\end{equation}
and $k' = \lceil \log_{1.5}(2 / \epsilon) \rceil$.
For each $P \in S'$, the algorithm computes
\begin{align}
    \hat{x}_P' &= \frac{1}{N} \sum_{\ell=1}^{N} \Tr(P \rho_\ell) y_\ell,\\
    \hat{\beta}_P' &= \frac{1}{N} \sum_{\ell=1}^{N} \Tr(P \rho_\ell) \Tr(P \rho_\ell),
\end{align}
using the full dataset $\{(\rho_\ell, y_\ell = \Tr(O^{\mathrm{unk} \rho_\ell}))\}_{\ell=1}^{N}$.
The algorithm uses the following hyperparameter
\begin{equation} \label{eq:t-eps-prime}
    \tilde{\epsilon}' \triangleq \frac{\epsilon}{6 n^{k'}}.
\end{equation}
Then, for each $P \in S'$, the algorithm computes
\begin{equation}
    \hat{\alpha}_P' =
    \begin{cases}
        0, & \hat{\beta}_P' \leq 2 \tilde{\epsilon}',\\
        \hat{x}_P' / \hat{\beta}_P', & \hat{\beta}_P' > 2 \tilde{\epsilon}'.
    \end{cases}
\end{equation}
The algorithm outputs the function $h'(\rho) = \Tr(\hat{O}' \rho)$, where the observable $\hat{O}'$ is defined as $\hat{O}' = \sum_{P \in S'} \hat{\alpha}_P' P.$

Here, we assume that $\Tr(P \rho_\ell)$ can be obtained from the training data.
However, for each $\Tr(P \rho_\ell)$, we only need to be able to obtain an unbiased estimator for $\Tr(P \rho_\ell)$ and for $\Tr(P \rho_\ell)^2$.
Recall that an unbiased estimator for $a$ is a random variable with expectation value equal to $a$.
For example, an unbiased estimator for $\Tr(P \rho_\ell)^2$ can be obtained by performing two quantum measurements on two individual copies of $\rho_\ell$ using the observable $P$ and multiplying the results, or by utilizing classical shadow formalism \cite{huang2020predicting} and randomized measurement \cite{elben2022randomized}.

\subsection{Rigorous performance guarantee}

In this section, we prove that the learning algorithm presented in the last section satisfies Theorem~\ref{thm:main-learning}.
We separate the proof for achieving the sample complexity on the left and right of Eq.~\eqref{eq:samp-learning-obs-dual}.

The proof for the sample complexity stated on the left of Eq.~\eqref{eq:samp-learning-obs-dual} consists of three parts: (1) a characterization of the prediction error, (2) the existence of a good hyperparamter $\eta^{\triangle}$ that achieves a small prediction error, (3) the hyperparameter $\eta^*$ found through grid search on the validation set has a small prediction error.

The proof for the sample complexity stated on the right of Eq.~\eqref{eq:samp-learning-obs-dual} is simpler and is given at the end.

\subsubsection{Characterization of the prediction error}

We begin with a lemma about the sample maximum.

\begin{lemma}[Sample maximum]
    Given $1 > \epsilon, \delta > 0$.
    Consider an arbitrary real-valued random variable $X$. Let $X_1, \ldots, X_N$ be $N$ independent samples of $X$ with $N = \lceil \log(1 / \delta) / \epsilon\rceil$ and let $\hat{\Theta} = \max_i X_i$. Then
    \begin{equation}
        \Pr\left[ X \leq \hat{\Theta} \right] \geq 1 - \epsilon.
    \end{equation}
    with probability at least $1 - \delta$.
\end{lemma}
\begin{proof}
Recall that the cumulative distribution function is defined as $F(\theta) = \Pr\left[ X \leq \theta \right]$.
We define the approximate maximum as follows,
\begin{equation}
\Theta \triangleq \inf_{\theta: F(\theta) \geq 1 - \epsilon} \theta.
\end{equation}
Using the right-continuity of $F(\theta) = \Pr\left[ X \leq \theta \right]$, we have
\begin{equation}
    F(\Theta) = \Pr\left[ X \leq \Theta \right] \geq 1 - \epsilon.
\end{equation}
Furthermore, from the definition of $\Theta$, we have
\begin{equation}
    \Pr\left[ X \geq \Theta \right] \geq \epsilon.
\end{equation}
To see the above inequality, suppose that $\Pr\left[ X \geq \Theta \right] < \epsilon$.
Then from the left-continuity of $F'(\theta) = \Pr\left[ X \geq \theta \right]$, we can find $\Theta' < \Theta$, such that $\Pr\left[ X \geq \Theta' \right] \leq \epsilon$.
Thus, there exists $\Theta' < \Theta$ with $\Pr\left[ X \leq \Theta' \right] \geq 1 - \epsilon$, which is a contradiction to the definition of $\Theta$.
Together, we have
\begin{equation}
    \Pr\left[ X_i < \Theta, \forall i \in [N] \right] \leq (1 - \epsilon)^N.
\end{equation}
By choosing $N = \lceil\log(1 / \delta) / \epsilon\rceil$, we have
\begin{equation}
    \Pr\left[ \max_{i} X_i \geq \Theta \right] \geq 1 - (1 - \epsilon)^{\log(1 / \delta) / \epsilon} \geq  1 - \delta.
\end{equation}
Thus with probability at least $1 - \delta$, we have $\hat{\Theta} \geq \Theta$. Using the monotonicity of $F(\theta)$, we have
\begin{equation}
    \Pr\left[ X \leq \hat{\Theta} \right] = F(\hat{\Theta}) \geq F(\Theta) \geq 1 - \epsilon,
\end{equation}
which establishes this lemma.
\end{proof}

Using the above lemma, we can show that given a training set of size
\begin{equation}
    N_{\mathrm{tr}} \geq \frac{12 \log(3 / \delta)}{\epsilon'},
\end{equation}
the real value $\hat{\Theta} \leq \norm{O^{\mathrm{(unk)}}}$ obtained by the algorithm satisfies
\begin{equation}
    \Pr_{\rho \sim \mathcal{D}} \left[ \left| \Tr(O^{\mathrm{(unk)}} \rho) \right| \leq \hat{\Theta} \right] \geq 1 - \frac{\epsilon'}{12}
\end{equation}
with probability at least $1 - (\delta / 3)$.
Hence, with probability at least $1 - (\delta / 3)$, we have
\begin{equation}
    \E_{\rho \sim \mathcal{D}} \left| h(\rho; \eta) - \Tr(O^{\mathrm{(unk)}} \rho) \right|^2 \leq \E_{\rho \sim \mathcal{D}} \left| \Tr(\hat{O}(\eta) \rho) - \Tr(O^{\mathrm{(unk)}} \rho) \right|^2 + \frac{\epsilon'}{12} \left| \hat{\Theta} +  \norm{O^{\mathrm{(unk)}}} \right|^2.
\end{equation}
Using Lemma~\ref{lem:mse-D} on mean squared error and Corollary~\ref{cor:low-deg-approx} on low-degree approximation, we have
\begin{align}
    &\E_{\rho \sim \mathcal{D}} \left| h(\rho; \eta) - \Tr(O^{\mathrm{(unk)}} \rho) \right|^2\\
    &\leq \underbrace{(2/3)^k \norm{O^{\mathrm{(unk)}}}^2}_{\leq \norm{O^{\mathrm{(unk)}}}^2 \epsilon} + \sum_{P \in S} \E_{\rho \sim \mathcal{D}} \left[ \gamma^*(\rho_{\mathsf{dom}(P)}) \right] \left( \frac{2}{3} \right)^{|P|} \left|\hat{\alpha}_P(\eta) -\alpha_P \right|^2 + \frac{\epsilon'}{3} \norm{O^{\mathrm{(unk)}}}^2
\end{align}
with probability at least $1 - (\delta / 3)$.

Let us define the following variables,
\begin{equation}
    x_P \triangleq \E_{\rho \sim \mathcal{D}} \left[ \gamma^*(\rho_{\mathsf{dom}(P)}) \right] \left( \frac{2}{3} \right)^{|P|} \alpha_P, \quad \beta_P \triangleq \E_{\rho \sim \mathcal{D}} \left[ \gamma^*(\rho_{\mathsf{dom}(P)}) \right] \left( \frac{2}{3} \right)^{|P|}, \quad \forall P \in S.
\end{equation}
Then, with probability at least $1 - (\delta / 3)$ over the sampling of the training set, we have the following characterization of the prediction error for all $\eta > 0$,
\begin{equation} \label{eq:char-prediction-error}
    \E_{\rho \sim \mathcal{D}} \left| h(\rho; \eta) - \Tr(O^{\mathrm{(unk)}} \rho) \right|^2 \leq \epsilon \norm{O^{\mathrm{(unk)}}}^2 + \frac{\epsilon'}{3} \norm{O^{\mathrm{(unk)}}}^2 + \sum_{P \in S} \beta_P \left|\hat{\alpha}_P(\eta) -\alpha_P \right|^2.
\end{equation}
We will utilize this form to show the existence of a good hyperparameter $\eta^{\triangle}$.

\subsubsection{Existence of a good hyperparamter $\eta^{\triangle}$}

By considering the training set size to be
\begin{equation}
    N_{\mathrm{tr}} = \Omega\left( \frac{\log(1 / \delta)}{\epsilon'} + \frac{\log(|S| / \delta)}{\tilde{\epsilon}^2} \right), \label{eq:Ntrdef}
\end{equation}
we can guarantee Eq.~\eqref{eq:char-prediction-error} with probability at least $1 - (\delta / 3)$.
Furthermore, utilizing Hoeffding's inequality and union bound, we could also guarantee that
\begin{equation} \label{eq:hat-est-error}
    \left| \hat{x}_P - x_P \right| \leq \norm{O^{\mathrm{(unk)}}} \tilde{\epsilon}, \quad \left| \hat{\beta}_P - \beta_P \right| \leq \tilde{\epsilon}, \quad \forall P \in S
\end{equation}
with probability at least $1 - (\delta/3)$.
The norm inequality given in Corollary~\ref{cor:any-norm} shows that
\begin{equation}
    \sum_{P \in S} \left| \alpha_P \right|^r \leq \left(  \frac{3}{C(k)} \right)^{r} \norm{O^{\mathrm{(low)}}}^r
\end{equation}
for a constant given by
\begin{equation}
    C(k) = \frac{\sqrt{2(k!)}}{2 k^{k+1.5+ (k+1)/(2k)} (\sqrt{6} + 2 \sqrt{3} )^k}.
\end{equation}
We now condition on the event that Eq.~\eqref{eq:char-prediction-error} and Eq.~\eqref{eq:hat-est-error} both hold, which happens with probability at least $1 - (2/3) \delta$.
We are now ready to define the good hyperparameter $\eta^{\triangle}$.

Let hyperparameter $\eta^{\triangle}$ belonging to the grid in Eq.~\eqref{eq:eta-grid} be defined as follows,
\begin{equation}
    \eta^{\triangle} = 2^{\min\left(R, \left\lceil \log_2\left(\norm{O^{\mathrm{(unk)}}} / \hat{\Theta} \right) \right\rceil \right)} \hat{\Theta}.
\end{equation}
We separately consider two cases: (1) $\eta^{\triangle} = 2^R \hat{\Theta}$, (2) $\eta^{\triangle} < 2^R \hat{\Theta}$.
For the first case $\eta^{\triangle} = 2^R \hat{\Theta}$, we can use $|\hat{x}_P| \leq \hat{\Theta}$ in Eq.~\eqref{eq:hatx-bound} and the definition of $R$ to see that
\begin{equation}
    \hat{\alpha}_P(\eta^{\triangle}) = 0, \quad \forall P \in S.
\end{equation}
Since $\eta^{\triangle} = 2^R \hat{\Theta}$, we have $R \leq \left\lceil \log_2\left(\norm{O^{\mathrm{(unk)}}} / \hat{\Theta} \right) \right\rceil$.
This yields $\eta^{\triangle} \leq 2 \norm{O^{\mathrm{(unk)}}}$, which implies that
\begin{equation}
    \hat{\alpha}_P\left(2 \norm{O^{\mathrm{(unk)}}}\right) = 0, \quad \forall P \in S.
\end{equation}
Hence, the reconstructed Pauli coefficients $\hat{\alpha}_P(\cdot)$ are the same for $\eta^{\triangle}$ and $2 \norm{O^{\mathrm{(unk)}}}$.
The filtering lemma given in Lemma~\ref{lem:filtering-full} shows that
\begin{align}
    &\sum_{P \in S} \E_{\rho \sim \mathcal{D}} \left[ \gamma^*(\rho_{\mathsf{dom}(P)}) \right] \left( \frac{2}{3} \right)^{|P|} \left|\hat{\alpha}_P(\eta^{\triangle}) -\alpha_P \right|^2\\
    &= \sum_{P \in S} \E_{\rho \sim \mathcal{D}} \left[ \gamma^*(\rho_{\mathsf{dom}(P)}) \right] \left( \frac{2}{3} \right)^{|P|} \left|\hat{\alpha}_P\left(2 \norm{O^{\mathrm{(unk)}}}\right) -\alpha_P \right|^2 \\
    &\leq 12 \left(  \frac{3}{C(k)} \right)^{r} \norm{O^{\mathrm{(unk)}}}^{2-r} \norm{O^{\mathrm{(low)}}}^r \tilde{\epsilon}^{1 - (r / 2)}.
\end{align}
For the second case $\eta^{\triangle} < 2^R \hat{\Theta}$, we have the following bound on $\eta^{\triangle}$,
\begin{equation}
    \eta^{\triangle} = 2^{\left\lceil \log_2\left(\norm{O^{\mathrm{(unk)}}} / \hat{\Theta} \right) \right\rceil} \hat{\Theta} \in \left[ \norm{O^{\mathrm{(unk)}}}, 2\norm{O^{\mathrm{(unk)}}} \right].
\end{equation}
The filtering lemma given in Lemma~\ref{lem:filtering-full} shows that
\begin{align}
    \sum_{P \in S} \E_{\rho \sim \mathcal{D}} \left[ \gamma^*(\rho_{\mathsf{dom}(P)}) \right] \left( \frac{2}{3} \right)^{|P|} \left|\hat{\alpha}_P(\eta^{\triangle}) -\alpha_P \right|^2 &\leq 6 (\eta^{\triangle})^r \left(  \frac{3}{C(k)} \right)^{r} \norm{O^{\mathrm{(low)}}}^r \tilde{\epsilon}^{1 - (r / 2)} \\
    &\leq 12 \left(  \frac{3}{C(k)} \right)^{r} \norm{O^{\mathrm{(unk)}}}^{2-r} \norm{O^{\mathrm{(low)}}}^r \tilde{\epsilon}^{1 - (r / 2)}.
\end{align}
In both case (1) and case (2), using the definition $r = 2k / (k+1)$ and $\tilde{\epsilon} = \left( \frac{\epsilon'}{12} \right)^{k + 1} \left( \frac{C(k)}{3} \right)^{2k}$, we have
\begin{equation}
    \sum_{P \in S} \E_{\rho \sim \mathcal{D}} \left[ \gamma^*(\rho_{\mathsf{dom}(P)}) \right] \left( \frac{2}{3} \right)^{|P|} \left|\hat{\alpha}_P(\eta^\triangle) -\alpha_P \right|^2 \leq \epsilon' \norm{O^{\mathrm{(unk)}}}^{2-r} \norm{O^{\mathrm{(low)}}}^r.
\end{equation}
Combining with Eq.~\eqref{eq:char-prediction-error}, we have
\begin{equation} \label{eq:good-eta-tri}
    \E_{\rho \sim \mathcal{D}} \left| h(\rho; \eta^\triangle) - \Tr(O^{\mathrm{(unk)}} \rho) \right|^2 \leq \epsilon \norm{O^{\mathrm{(unk)}}}^2 + \frac{\epsilon'}{3} \norm{O^{\mathrm{(unk)}}}^2 +  \epsilon' \norm{O^{\mathrm{(unk)}}}^{2-r} \norm{O^{\mathrm{(low)}}}^r
\end{equation}
with probability at least $1 - (2/3) \delta$.

\subsubsection{The prediction performance of the hyperparameter $\eta^*$}

From the definition of $h(\rho; \eta)$, for any quantum state $\rho$, we have
\begin{equation}
    \left| h(\rho; \eta) - \Tr(O^{\mathrm{(unk)}} \rho)) \right|^2 \leq \left|\hat{\Theta} + \norm{O^{\mathrm{(unk)}}} \right|^2 \leq 4 \norm{O^{\mathrm{(unk)}}}^2
\end{equation}
Using Hoeffding's inequality and union bound, we can show that given a validation set of size
\begin{equation}
    N_{\mathrm{val}} = \Omega\left( \frac{\log(R / \delta) }{(\epsilon')^2} \right), \label{eq:Nvaldef}
\end{equation}
with probability at least $1 - (\delta / 3)$, we have
\begin{equation}
 \left| \frac{1}{N_{\mathrm{val}}}\sum_{\ell = N_{\mathrm{tr}} + 1}^{N_{\mathrm{tr}} + N_{\mathrm{val}}} \left| h(\rho_\ell; \eta) - \Tr(O^{\mathrm{(unk)}} \rho_\ell)) \right|^2 - \E_{\rho \sim \mathcal{D}}  \left| h(\rho; \eta) - \Tr(O^{\mathrm{(unk)}} \rho)) \right|^2 \right| \leq \norm{O^{\mathrm{(unk)}}}^2 \frac{\epsilon'}{3},
\end{equation}
for all $\eta \in \{2^0 \hat{\Theta}, \ldots, 2^R \hat{\Theta}\}$.
Using the definition of $\eta^*$ and $\eta^\triangle$, we have
\begin{align}
    \E_{\rho \sim \mathcal{D}}  \left| h(\rho; \eta^*) - \Tr(O^{\mathrm{(unk)}} \rho)) \right|^2 &\leq \frac{1}{N_{\mathrm{val}}}\sum_{\ell = N_{\mathrm{tr}} + 1}^{N_{\mathrm{tr}} + N_{\mathrm{val}}} \left| h(\rho_\ell; \eta^*) - \Tr(O^{\mathrm{(unk)}} \rho_\ell)) \right|^2 + \norm{O^{\mathrm{(unk)}}}^2 \frac{\epsilon'}{3}\\
    &\leq \frac{1}{N_{\mathrm{val}}}\sum_{\ell = N_{\mathrm{tr}} + 1}^{N_{\mathrm{tr}} + N_{\mathrm{val}}} \left| h(\rho_\ell; \eta^\triangle) - \Tr(O^{\mathrm{(unk)}} \rho_\ell)) \right|^2 + \norm{O^{\mathrm{(unk)}}}^2 \frac{\epsilon'}{3}\\
    &\leq \E_{\rho \sim \mathcal{D}}  \left| h(\rho; \eta^\triangle) - \Tr(O^{\mathrm{(unk)}} \rho)) \right|^2 + \norm{O^{\mathrm{(unk)}}}^2 \frac{2 \epsilon'}{3}
\end{align}
with probability at least $1 - (\delta / 3)$ over the sampling of the validation set.
Combining with Eq.~\eqref{eq:good-eta-tri} and employing union bound, we have
\begin{equation}
    \E_{\rho \sim \mathcal{D}}  \left| h(\rho; \eta^*) - \Tr(O^{\mathrm{(unk)}} \rho)) \right|^2 \leq \epsilon \norm{O^{\mathrm{(unk)}}}^2 + \epsilon' \norm{O^{\mathrm{(unk)}}}^2 + \epsilon' \norm{O^{\mathrm{(unk)}}}^{2-r} \norm{O^{\mathrm{(low)}}}^r
\end{equation}
with probability at least $1 - \delta$, as claimed in Eq.~\eqref{eq:main-learn-bound}.

Finally, by noting that $|S| = \mathcal{O}(n^k)$ and $k = \log_{1.5}(1/\epsilon)$ and recalling the definition of $\tilde{\epsilon}$ in Eq~.\eqref{eq:tildeeps_def} on the right-hand side of Eq.~\eqref{eq:Ntrdef}, we have
\begin{equation}
    \frac{\log(1 / \delta)}{\epsilon'} + \frac{\log(|S| / \delta)}{\tilde{\epsilon}^2} = \log\left(\frac{n}{\delta}\right) \left(\frac{1}{\epsilon'}\right)^{k+1} 2^{\mathcal{O}(k \log k)} = \log\left(\frac{n}{\delta}\right) 2^{\mathcal{O}\left(\log(\frac{1}{\epsilon}) \left( \log\log(\frac{1}{\epsilon}) +  \log(\frac{1}{\epsilon'})\right)\right)}.
\end{equation} So it suffices to have
\begin{equation}
    N_{\mathrm{val}} = \log\left(\frac{n}{\delta}\right) 2^{\Omega\left(\log(\frac{1}{\epsilon}) \left( \log\log(\frac{1}{\epsilon}) +  \log(\frac{1}{\epsilon'})\right)\right)}. \label{eq:Ntr_bd}
\end{equation}
Furthermore, by noting that $R = \log_2 \lceil 1 / \tilde{\epsilon} \rceil = \mathcal{O}(k \log(\epsilon') + k \log^2 k)$ in Eq.~\eqref{eq:Nvaldef}, we see that it suffices to have
\begin{equation}
    N_{\mathrm{val}} = \Omega\left( \frac{\log \log (\epsilon) + \log \log (\epsilon') + \log(1 / \delta)}{(\epsilon')^2} \right). \label{eq:Nval_bd}
\end{equation}
Recall that the full data size $N = N_{\mathrm{tr}} + N_{\mathrm{val}}$, and the quantity in Eq.~\eqref{eq:Nval_bd} is dominated by the one in Eq.~\eqref{eq:Ntr_bd}, yielding one argument in the minimum of the sample complexity claimed in Theorem~\ref{thm:main-learning}.

\subsubsection{Establishing sample complexity on the right of Eq.~\eqref{eq:samp-learning-obs-dual}}

By considering the full dataset size to be
\begin{equation}
    N = \Omega\left( \frac{\log(|S'| / \delta)}{(\tilde{\epsilon}')^2} \right),
\end{equation}
Hoeffding's inequality and union bound can be used to guarantee that
\begin{equation}
    \left| \hat{x}_P' - x_P \right| \leq \norm{O^{\mathrm{(unk)}}} \tilde{\epsilon}', \quad \left| \hat{\beta}_P' - \beta_P \right| \leq \tilde{\epsilon}', \quad \forall P \in S'
\end{equation}
with probability at least $1 - \delta$.
Using Lemma~\ref{lem:filter-smallweight} on filtering small-weight factor,
we have
\begin{equation}
    \beta_P \left| \hat{\alpha}_P' - \alpha_P \right|^2 \leq 3 \norm{O^{\mathrm{(unk)}}}^2 \tilde{\epsilon}'.
\end{equation}
Using Lemma~\ref{lem:mse-D} on mean squared error and Corollary~\ref{cor:low-deg-approx} on low-degree approximation, we have
\begin{align}
    \E_{\rho \sim \mathcal{D}} \left| \Tr(\hat{O}' \rho) - \Tr(O^{\mathrm{(unk)}} \rho) \right|^2 &\leq (2/3)^k \norm{O^{\mathrm{(unk)}}}^2 + \sum_{P \in S'} \beta_P \left|\hat{\alpha}_P' -\alpha_P \right|^2\\
    &\leq \norm{O^{\mathrm{(unk)}}}^2 \frac{\epsilon}{2} + 3 n^{k'} \norm{O^{\mathrm{(unk)}}}^2 \tilde{\epsilon}'.
\end{align}
From the definition of $\tilde{\epsilon}'$ in Eq.~\eqref{eq:t-eps-prime}, we have
\begin{equation}
    \E_{\rho \sim \mathcal{D}} \left| \Tr(\hat{O}' \rho) - \Tr(O^{\mathrm{(unk)}} \rho) \right|^2 \leq \epsilon \norm{O^{\mathrm{(unk)}}}^2.
\end{equation}
The sample complexity is
\begin{equation}
    N = \mathcal{O}\left( \frac{\log(|S'| / \delta)}{(\tilde{\epsilon}')^2} \right) = \log(n / \delta) \, 2^{\mathcal{O}(\log(1 / \epsilon) \log(n))},
\end{equation}
which completes the sample complexity claimed in Theorem~\ref{thm:main-learning}.

\section{Learning quantum evolutions from randomized experiments}
\label{sec:learning-qevo}

We recall the following definitions pertaining to classical shadows for quantum states and quantum evolutions, based on randomized Pauli measurements and random input states.

\begin{definition}[Single-qubit stabilizer state]
We define
\begin{equation}
\mathrm{stab}_1 \triangleq \{\ket{0}, \ket{1}, \ket{+}, \ket{-}, \ket{y+}, \ket{y-}\}
\end{equation}
to be the set of single-qubit stabilizer states.
\end{definition}

We define randomized Pauli measurements as follows.

\begin{definition}[Randomized Pauli measurement]
Given $n > 0$. A randomized Pauli measurement on an $n$-qubit state is given by a $6^n$-outcome POVM
\begin{equation}
    \mathcal{F}^{\mathrm{(Pauli)}} \triangleq \left\{ \frac{1}{3^n} \bigotimes_{i=1}^n \ketbra{s_i}{s_i} \right\}_{s_1, \ldots, s_n \in \mathrm{stab}_1},
\end{equation}
which corresponds to measuring every qubit under a random Pauli basis ($X, Y, Z$).
The outcome of $\mathcal{F}^{\mathrm{(Pauli)}}$ is an $n$-qubit state $\ket{\psi} = \bigotimes_{i=1}^n \ket{s_{i}},$ where $\ket{s_{i}} \in \mathrm{stab}_1$ is a single-qubit stabilizer state.
\end{definition}

In the following, we define the classical shadow of a quantum state based on randomized Pauli measurements.
Classical shadows could also be defined based on other randomized measurements \cite{huang2020predicting}.

\begin{definition}[Classical shadow of a quantum state]
Given $n, N > 0$.
Consider an $n$-qubit state $\rho$.
A size-$N$ classical shadow $S_N(\rho)$ of quantum state $\rho$ is a random set given by
\begin{equation}
    S_N(\rho) \triangleq \left\{ \ket{\psi_\ell} \right\}_{\ell = 1}^N,
\end{equation}
where $\ket{\psi_\ell} = \bigotimes_{i=1}^n \ket{s_{\ell, i}}$ is the outcome of the $\ell$-th randomized Pauli measurement on a single copy of $\rho$.
\end{definition}

We can generalize classical shadows from quantum states to quantum processes by considering random product input states and randomized Pauli measurements.
A similar generalization has been studied in \cite{levy2021classical}.

\begin{definition}[Classical shadow of a quantum process]
Given an $n$-qubit CPTP map $\mathcal{E}$.
A size-$N$ classical shadow $S_N(\mathcal{E})$ of quantum evolution $\mathcal{E}$ is a random set given by
\begin{equation}
    S_N(\mathcal{E}) \triangleq \left\{ \ket{\psi^{\mathrm{(in)}}_\ell}, \ket{\psi^{\mathrm{(out)}}_\ell} \right\}_{\ell = 1}^N,
\end{equation}
where $\ket{\psi^{\mathrm{(in)}}_\ell} = \bigotimes_{i=1}^n \ket{s^{\mathrm{(in)}}_{\ell, i}}$ is a random input state with $\ket{s^{\mathrm{(in)}}_{\ell, i}} \in \mathrm{stab}_1$ sampled uniformly, and $\ket{\psi^{\mathrm{(out)}}_\ell} = \bigotimes_{i=1}^n \ket{s^{\mathrm{(out)}}_{\ell, i}}$ is the outcome of performing randomized Pauli measurement on $\mathcal{E}(\ketbra{\psi^{\mathrm{(in)}}_\ell}{\psi^{\mathrm{(in)}}_\ell})$.
\end{definition}

After obtaining the outcome from $N$ randomized experiments, we can design a learning algorithm that learns a model of the unknown CPTP map $\mathcal{E}$, such that given an input state $\rho$ and an observable $O$, the algorithm could predict $\Tr(O \mathcal{E}(\rho))$.
The rigorous guarantee is given in the following theorem.

\begin{theorem}[Learning to predict a quantum evolution] \label{thm:main-learning-evo}
    Given $n, \epsilon, \epsilon', \delta > 0$.
    Consider any unknown $n$-qubit CPTP map $\mathcal{E}$.
    Given a classical shadow $S_N(\mathcal{E})$ of $\mathcal{E}$ obtained by $N$ randomized experiments with
    \begin{equation} \label{eq:Nsamp-evo}
        N = \log\left(\frac{n}{\delta}\right) \, \min\left( 2^{\mathcal{O}\left(\log(\frac{1}{\epsilon}) \left( \log\log(\frac{1}{\epsilon}) +  \log(\frac{1}{\epsilon'})\right)\right)}, \,\, 2^{\mathcal{O}\left( \log(1 / \epsilon) \log(n) \right)} \right).
    \end{equation}
    With probability $\geq 1 - \delta$,
    the algorithm learns a function~$h$, s.t. for any $n$-qubit state distribution $\mathcal{D}$ invariant under single-qubit $H$ and $S$ gates, and any observable $O$ given as a sum of few-body observables, where each qubit is acted on by $\mathcal{O}(1)$ of the few-body observables,
    \begin{equation}
        \E_{\rho \sim \mathcal{D}} \left| h(\rho, O) - \Tr\left( O \mathcal{E}(\rho) \right) \right|^2 \leq
        \left( \epsilon + \epsilon'  \left[ \tfrac{\norm{O^{\mathrm{(\mathrm{low})}}}}{\norm{O}} \right]^{\frac{2 \lceil \log_{1.5}(1/\epsilon) \rceil}{\lceil \log_{1.5}(1/\epsilon) \rceil + 1}} \right) \norm{O}^2. \label{eq:main-learn-bound-evo}
    \end{equation}
    Here, $O^{(\mathrm{low})}$ is the low-degree approximation of $O$ after Heisenberg evolution under $\mathcal{E}$.
\end{theorem}

The scaling given in the main text corresponds to the additional assumption that $\norm{O} \leq 1$.
By noting that $\frac{2 \lceil \log_{1.5}(1/\epsilon) \rceil}{\lceil \log_{1.5}(1/\epsilon) \rceil + 1} \in [1, 2)$, we have
\begin{equation}
    \left[ \tfrac{\norm{O^{\mathrm{(\mathrm{low})}}}}{\norm{O}} \right]^{\frac{2 \lceil \log_{1.5}(1/\epsilon) \rceil}{\lceil \log_{1.5}(1/\epsilon) \rceil + 1}} \norm{O}^2 \leq \norm{O^{\mathrm{(\mathrm{low})}}}^{\frac{2 \lceil \log_{1.5}(1/\epsilon) \rceil}{\lceil \log_{1.5}(1/\epsilon) \rceil + 1}} \leq \max\left(\norm{O^{\mathrm{(\mathrm{low})}}}^2, 1\right).
\end{equation}
Theorem~\ref{thm:main-res} follows by considering $\epsilon' \rightarrow 0$.

\subsection{Learning algorithm}

Recall that a size-$N$ classical shadow $S_N(\mathcal{E})$ of the CPTP map $\mathcal{E}$ is a set given by
\begin{equation}
    S_N(\mathcal{E}) \triangleq \left\{ \ket{\psi^{\mathrm{(in)}}_\ell} = \bigotimes_{i=1}^n \ket{s^{\mathrm{(in)}}_{\ell, i}}, \ket{\psi^{\mathrm{(out)}}_\ell} = \bigotimes_{i=1}^n \ket{s^{\mathrm{(out)}}_{\ell, i}} \right\}_{\ell = 1}^N.
\end{equation}
Given an observable $O$ that can be written as a sum of $\kappa$-qubit observables, where each qubit is acted on by at most $d$ of the $\kappa$-qubit observables with $\kappa, d = \mathcal{O}(1)$. We have
\begin{equation}
    O = \sum_{Q \in \{I, X, Y, Z\}^{\otimes n}: |Q| \leq \kappa} a_Q Q,
\end{equation}
where $\sum_{Q: |Q| \leq \kappa} \indicator[a_Q \neq 0] = \mathcal{O}(n)$.
The algorithm creates a dataset,
\begin{equation}
    \left\{ \rho_\ell = \ketbra{\psi^{\mathrm{(in)}}_\ell}{\psi^{\mathrm{(in)}}_\ell}, \quad y_\ell(O) = \sum_{Q: |Q| \leq \kappa} a_Q \Tr\left(Q \bigotimes_{i=1}^n\left(3 \ketbra{s^{\mathrm{(out)}}_{\ell, i}}{s^{\mathrm{(out)}}_{\ell, i}} - I \right) \right) \right\}_{\ell=1}^N
\end{equation}
from the classical shadow $S_N(\mathcal{E})$, which requires $\mathcal{O}(n N)$ computational time.
We also define the parameter
\begin{equation}
    \eta \triangleq \sum_{Q: |Q| \leq \kappa} |a_Q| = \norm{O}_{\mathrm{Pauli}, 1}
\end{equation}
based on the given observable $O$.

The sample complexity in Eq.~\eqref{eq:Nsamp-evo} is the minimum of two arguments.
Each of the two corresponds to a hyperparameter setting for $k$ and $\tilde{\epsilon}$.
Let $C(k)$ be the function from Corollary~\ref{cor:any-norm} and $C(k, d)$ be the function from Corollary~\ref{cor:bounded-norm}.
The first hyperparameter setting considers
\begin{equation} \label{eq:tilde-eps-def}
k = \lceil \log_{1.5}(1 / \epsilon) \rceil, \quad
\tilde{\epsilon} = \left( \frac{\epsilon'}{6 \cdot 2^k} \right)^{k + 1} \left( \frac{C(\kappa, d)}{3} \right)^{2} \left( \frac{C(k)}{3} \right)^{2k}.
\end{equation}
The second hyperparameter setting considers
\begin{equation} \label{eq:tilde-eps-def-prime}
k = \lceil \log_{1.5}(2 / \epsilon) \rceil, \quad
\tilde{\epsilon} = \frac{\epsilon}{9 \cdot 2^{k+1} \cdot n^{k}} \left( \frac{C(\kappa, d)}{3} \right)^{2}.
\end{equation}
For every Pauli observable $P \in \{I, X, Y, Z\}^{\otimes n}$ with $|P| \leq k$, the algorithm computes
\begin{align}
    \hat{x}_P(O) &= \frac{1}{N} \sum_{\ell=1}^{N} \Tr(P \rho_\ell) y_\ell(O),\\
    \hat{\beta}_P &= \left(\frac{1}{3}\right)^{|P|},\\
    \hat{\alpha}_P(O) &=
    \begin{cases}
        0, & \hat{\beta}_P \leq 2 \tilde{\epsilon},\\
        0, & \hat{\beta}_P > 2 \tilde{\epsilon}, \, |\hat{x}_P(O)| / \hat{\beta}_P^{1/2} \leq 2 \eta \sqrt{\tilde{\epsilon}},\\
        \hat{x}_P(O) / \hat{\beta}_P, & \hat{\beta}_P > 2 \tilde{\epsilon}, \, |\hat{x}_P(O)| / \hat{\beta}_P^{1/2} > 2 \eta \sqrt{\tilde{\epsilon}},
    \end{cases}
\end{align}
which requires $\mathcal{O}(k N)$ time per Pauli observable $P$.
Finally, given an $n$-qubit state $\rho$, the algorithm outputs
\begin{equation}
    h(\rho, O) \triangleq \sum_{P: |P| \leq k} \hat{\alpha}_P(O) \Tr(P\rho),
\end{equation}
which uses a computational time of $\mathcal{O}(n^k)$.

\subsection{Rigorous performance guarantee}

In this section, we prove that the learning algorithm presented in the last section satisfies Theorem~\ref{thm:main-learning-evo}.
The proof uses the tools presented in Section~\ref{sec:Pauli-tools} and is similar to the proof of Theorem~\ref{thm:main-learning}.

\subsubsection{Definitions}

For a given observable that is a sum of $\kappa$-qubit observables, where $\kappa = \mathcal{O}(1)$ and each qubit is acted on by $d = \mathcal{O}(1)$ of the $\kappa$-qubit observables, we can write
\begin{equation}
    O = \sum_{Q \in \{I, X, Y, Z\}^{\otimes n}: |Q| \leq \kappa} a_Q Q.
\end{equation}
We define a few variables based on $O$ as follows. We consider the unknown observable to be
\begin{equation}
O^{(\mathrm{unk})} \triangleq \mathcal{E}^\dagger(O) \triangleq \sum_{P \in \{I, X, Y, Z\}^{\otimes n}} \alpha_P(O) P,
\end{equation}
and the low-degree approximation of $O^{(\mathrm{unk})}$ to be
\begin{equation}
    O^{(\mathrm{low})} \triangleq \sum_{P \in \{I, X, Y, Z\}^{\otimes n}: |P| \leq k} \alpha_P(O) P.
\end{equation}
Then for all Pauli observables $P \in \{I, X, Y, Z\}^{\otimes n}$, we define
\begin{equation}
    x_P(O) \triangleq \left(\frac{1}{3}\right)^{|P|} \alpha_P(O), \quad \beta_P \triangleq \left(\frac{1}{3}\right)^{|P|}.
\end{equation}
We also define the standard $n$-qubit input state distribution $\mathcal{D}^0$
to be the uniform distribution over the tensor product of $n$ single-qubit stabilizer states.
A nice property of $\mathcal{D}^0$ is that for any state $\rho$ in the support of $\mathcal{D}^0$, the non-identity purity for a subsystem $A$ of size $L$ is
\begin{equation}
    \gamma^*(\rho_A) = \frac{1}{2^L}.
\end{equation}
Using this property and Lemma~\ref{lem:extract-Pauli} on extracting Pauli coefficients, we have the following identities
\begin{equation}
    x_P(O) = \E_{\rho \sim \mathcal{D}^0} \Tr(P \rho) \Tr\left(\mathcal{E}^\dagger(O) \rho \right), \quad \beta_P = \E_{\rho \sim \mathcal{D}^0} \Tr(P \rho)^2 = \E_{\rho \sim \mathcal{D}^0} \left[ \gamma^*(\rho_{\mathsf{dom}(P)}) \right] \left( \frac{2}{3} \right)^{|P|}.
\end{equation}
We are now ready to prove Theorem~\ref{thm:main-learning-evo}.

\subsubsection{Prediction error under standard distribution $\mathcal{D}^0$ (first set of hyperparameters)}

We begin the proof by considering the first set of hyperparameters $k, \tilde{\epsilon}$ as given in Eq.~\eqref{eq:tilde-eps-def}.
For a Pauli observable $Q \in \{I, X, Y, Z\}^{\otimes n}$ with $|Q| \leq \kappa = \mathcal{O}(1)$, we consider the random variable
\begin{equation}
    \hat{x}_P(Q) = \frac{1}{N} \sum_{\ell=1}^{N} \Tr(P \rho_\ell) y_\ell(Q) = \frac{1}{N} \sum_{\ell=1}^{N} \Tr(P \rho_\ell) \Tr\left(Q \bigotimes_{i=1}^n\left(3 \ketbra{s^{\mathrm{(out)}}_{\ell, i}}{s^{\mathrm{(out)}}_{\ell, i}} - I \right) \right).
\end{equation}
Because $|Q| = \mathcal{O}(1)$, we have $\left|\Tr\left(Q \bigotimes_{i=1}^n\left(3 \ketbra{s^{\mathrm{(out)}}_{\ell, i}}{s^{\mathrm{(out)}}_{\ell, i}} - I \right) \right)\right| = \mathcal{O}(1)$ with probability one.
By considering the size of the classical shadow $S_N(\mathcal{E})$ to be
\begin{equation}
    N = \Omega\left( \frac{\log(n^{k + \kappa} / \delta)}{\tilde{\epsilon}^2} \right), \label{eq:Ntrdef-evo}
\end{equation}
we can utilize Hoeffding's inequality and union bound to guarantee that
\begin{equation}
    \left| \hat{x}_P(Q) - x_P(Q) \right| \leq \tilde{\epsilon}, \quad \forall P, Q \in \{I, X, Y, Z\}^{\otimes n}, |P| \leq k, |Q| \leq \kappa
\end{equation}
with probability at least $1 - \delta$.
In the following proof, we will condition on the above event.

Using triangle inequality, we have
\begin{equation} \label{eq:hat-est-error-learn-evo}
    \left| \hat{x}_P(O) - x_P(O) \right| \leq \norm{O}_{\mathrm{Pauli},1}\tilde{\epsilon} = \eta \tilde{\epsilon}, \quad \left| \hat{\beta}_P - \beta_P \right| = 0, \quad \forall P: |P| \leq k.
\end{equation}
The norm inequality given in Corollary~\ref{cor:any-norm} shows that
\begin{equation}
    \sum_{P: |P| \leq k} \left| \alpha_P(O) \right|^r \leq \left(  \frac{3}{C(k)} \right)^{r} \norm{O^{\mathrm{(low)}}}^r
\end{equation}
for the constant $C(k)$ defined in \eqref{cor:opt-any}.

The filtering lemma given in Lemma~\ref{lem:filtering-full} shows that
\begin{align}
    \sum_{P: |P| \leq k} \E_{\rho \sim \mathcal{D}^0} \left[ \gamma^*(\rho_{\mathsf{dom}(P)}) \right] \left( \frac{2}{3} \right)^{|P|} \left|\hat{\alpha}_P(O) -\alpha_P(O) \right|^2 &\leq 6 \eta^{2-r} \left(  \frac{3}{C(k)} \right)^{r} \norm{O^{\mathrm{(low)}}}^r \tilde{\epsilon}^{1 - (r / 2)}.
\end{align}
From the norm inequality and the constant $C(k, d)$ given in Corollary~\ref{cor:bounded-norm}, we have
\begin{equation}
    \eta = \norm{O}_{\mathrm{Pauli}, 1} \leq \frac{3}{C(\kappa, d)} \norm{O}.
\end{equation}
Combining with the definition of $\tilde{\epsilon}$ given in Eq.~\eqref{eq:tilde-eps-def}, we have
\begin{equation} \label{eq:error-bound-D0}
    \sum_{P: |P| \leq k} \E_{\rho \sim \mathcal{D}^0} \left[ \gamma^*(\rho_{\mathsf{dom}(P)}) \right] \left( \frac{2}{3} \right)^{|P|} \left|\hat{\alpha}_P(O) -\alpha_P(O) \right|^2 \leq \left[\tfrac{\norm{O^{\mathrm{(low)}}}}{\norm{O}}\right]^r \frac{\epsilon'}{2^k} \cdot \norm{O}^2.
\end{equation}
Using Lemma~\ref{lem:mse-D} on mean squared error and Corollary~\ref{cor:low-deg-approx} on low-degree approximation, we have
\begin{align}
    &\E_{\rho \sim \mathcal{D}^0} \left| h(\rho, O) - \Tr(O^{\mathrm{(unk)}} \rho) \right|^2\\
    &\leq \underbrace{(2/3)^k \norm{O^{\mathrm{(unk)}}}^2}_{\leq \norm{O^{\mathrm{(unk)}}}^2 \epsilon} + \sum_{P: |P| \leq k} \E_{\rho \sim \mathcal{D}^0} \left[ \gamma^*(\rho_{\mathsf{dom}(P)}) \right] \left( \frac{2}{3} \right)^{|P|} \left|\hat{\alpha}_P(O) -\alpha_P(O) \right|^2.
\end{align}
Using the definition of $O^{\mathrm{(unk)}}$, we have $O^{\mathrm{(unk)}} = \mathcal{E}^\dagger(O)$ and $\norm{O^{\mathrm{(unk)}}} \leq \norm{O}$. Hence
\begin{equation}
    \E_{\rho \sim \mathcal{D}^0} \left| h(\rho, O) - \Tr(O \mathcal{E}(\rho)) \right|^2 \leq \left(\epsilon +  \frac{\epsilon'}{2^k} \left[\tfrac{\norm{O^{\mathrm{(low)}}}}{\norm{O}}\right]^r \right) \norm{O}^2,
\end{equation}
which establishes a prediction error bound for distribution $\mathcal{D}^0$.

\subsubsection{Prediction error under general distribution $\mathcal{D}$  (first set of hyperparameters)}

We now consider an arbitrary $n$-qubit state distribution $\mathcal{D}$ invariant under single-qubit $H$ and $S$ gates.
Using Lemma~\ref{lem:mse-D} on mean squared error and Corollary~\ref{cor:low-deg-approx} on low-degree approximation, we have
\begin{align}
    &\E_{\rho \sim \mathcal{D}} \left| h(\rho, O) - \Tr(O^{\mathrm{(unk)}} \rho) \right|^2\\
    &\leq \epsilon \norm{O}^2 + \sum_{P: |P| \leq k} \E_{\rho \sim \mathcal{D}} \left[ \gamma^*(\rho_{\mathsf{dom}(P)}) \right] \left( \frac{2}{3} \right)^{|P|} \left|\hat{\alpha}_P(O) -\alpha_P(O) \right|^2.
\end{align}
Recall that $\gamma^*(\rho_{\mathsf{dom}(P)}) \leq 1$, hence
\begin{equation}
    \E_{\rho \sim \mathcal{D}} \left[ \gamma^*(\rho_{\mathsf{dom}(P)}) \right] \left( \frac{2}{3} \right)^{|P|} \leq 2^k \left( \frac{1}{3} \right)^{|P|}, \quad \forall P \in \{I, X, Y, Z\}^{\otimes n}, |P| \leq k.
\end{equation}
Furthermore, we have $\E_{\rho \sim \mathcal{D}_0} \left[ \gamma^*(\rho_{\mathsf{dom}(P)}) \right] \left( \frac{2}{3} \right)^{|P|} = (1 / 3)^{|P|}$. Together, we have
\begin{equation}
    \E_{\rho \sim \mathcal{D}} \left| h(\rho, O) - \Tr(O^{\mathrm{(unk)}} \rho) \right|^2 \leq \epsilon \norm{O}^2 + 2^k \sum_{P: |P| \leq k} \E_{\rho \sim \mathcal{D}^0} \left[ \gamma^*(\rho_{\mathsf{dom}(P)}) \right] \left( \frac{2}{3} \right)^{|P|} \left|\hat{\alpha}_P(O) -\alpha_P(O) \right|^2.
\end{equation}
Combining the above with Eq.~\eqref{eq:error-bound-D0}, we have
\begin{equation}
    \E_{\rho \sim \mathcal{D}} \left| h(\rho, O) - \Tr(O \mathcal{E}(\rho)) \right|^2 \leq \left(\epsilon +  \epsilon' \left[\tfrac{\norm{O^{\mathrm{(low)}}}}{\norm{O}}\right]^r \right) \norm{O}^2,
\end{equation}
which is the prediction error under distribution $\mathcal{D}$.

\subsubsection{Putting everything together (first set of hyperparameters)}

From Eq.~\eqref{eq:tilde-eps-def}, we have set the parameter $\tilde{\epsilon}$ to be
\begin{equation}
    \tilde{\epsilon} = \left( \frac{\epsilon'}{6} \right)^{k + 1} \left( \frac{C(\kappa, d)}{3} \right)^{2} \left( \frac{C(k)}{3} \right)^{2k}.
\end{equation}
Furthermore, given the classical shadow $S_N(\mathcal{E})$ of size
\begin{equation}
    N = \mathcal{O}\left( \frac{\log(n^{k + \kappa} / \delta)}{\tilde{\epsilon}^2} \right) = \log\left(\frac{n}{\delta}\right) \, 2^{\mathcal{O}\left(\log(\frac{1}{\epsilon}) \left( \log\log(\frac{1}{\epsilon}) +  \log(\frac{1}{\epsilon'})\right)\right)},
\end{equation}
we can guarantee that with probability at least $1 - \delta$, the following holds.
For any observable $O$ that is a sum of $\kappa$-qubit observables, where $\kappa = \mathcal{O}(1)$ and each qubit is acted on by $d = \mathcal{O}(1)$ of the $\kappa$-qubit observables, and any $n$-qubit state distribution $\mathcal{D}$ invariant under single-qubit $H$ and $S$ gates, we have
\begin{equation}
    \E_{\rho \sim \mathcal{D}} \left| h(\rho, O) - \Tr(O \mathcal{E}(\rho)) \right|^2 \leq \left(\epsilon +  \epsilon' \left[\tfrac{\norm{O^{\mathrm{(low)}}}}{\norm{O}}\right]^r \right) \norm{O}^2.
\end{equation}
This establishes one of the argument for the sample complexity stated in Theorem~\ref{thm:main-learning-evo}.

\subsubsection{Prediction error under standard distribution $\mathcal{D}^0$ (second set of hyperparameters)}

In the following proof, we consider the second set of hyperparameters $k, \tilde{\epsilon}$ as given in Eq.~\eqref{eq:tilde-eps-def-prime}.
By considering the size of the classical shadow $S_N(\mathcal{E})$ to be
\begin{equation}
    N = \Omega\left( \frac{\log(n^{k + \kappa} / \delta)}{\tilde{\epsilon}^2} \right),
\end{equation}
we can utilize Hoeffding's inequality and union bound to guarantee that
\begin{equation}
    \left| \hat{x}_P(Q) - x_P(Q) \right| \leq \tilde{\epsilon}, \quad \forall P, Q \in \{I, X, Y, Z\}^{\otimes n}, |P| \leq k, |Q| \leq \kappa
\end{equation}
with probability at least $1 - \delta$.
In the following proof, we will condition on the above event.
Using triangle inequality, we have
\begin{equation}
    \left| \hat{x}_P(O) - x_P(O) \right| \leq \norm{O}_{\mathrm{Pauli},1}\tilde{\epsilon} = \eta \tilde{\epsilon}, \quad \left| \hat{\beta}_P - \beta_P \right| = 0, \quad \forall P: |P| \leq k.
\end{equation}
The filtering lemma given in Lemma~\ref{lem:filtering-full} shows that
\begin{align}
    \sum_{P: |P| \leq k} \E_{\rho \sim \mathcal{D}^0} \left[ \gamma^*(\rho_{\mathsf{dom}(P)}) \right] \left( \frac{2}{3} \right)^{|P|} \left|\hat{\alpha}_P(O) -\alpha_P(O) \right|^2 &\leq 9 \eta^2 \tilde{\epsilon}^2.
\end{align}
From the norm inequality and the function $C(k, d)$ given in Corollary~\ref{cor:bounded-norm}, we have
\begin{equation}
    \eta = \norm{O}_{\mathrm{Pauli}, 1} \leq \frac{3}{C(\kappa, d)} \norm{O}.
\end{equation}
Combining with the definition of $\tilde{\epsilon}$ given in Eq.~\eqref{eq:tilde-eps-def-prime}, we have
\begin{equation} \label{eq:error-bound-D0-prime}
    \sum_{P: |P| \leq k} \E_{\rho \sim \mathcal{D}^0} \left[ \gamma^*(\rho_{\mathsf{dom}(P)}) \right] \left( \frac{2}{3} \right)^{|P|} \left|\hat{\alpha}_P(O) -\alpha_P(O) \right|^2 \leq \frac{\epsilon}{2^{k+1}} \cdot \norm{O}^2.
\end{equation}
Using Lemma~\ref{lem:mse-D} on mean squared error and Corollary~\ref{cor:low-deg-approx} on low-degree approximation, we have
\begin{align}
    &\E_{\rho \sim \mathcal{D}^0} \left| h(\rho, O) - \Tr(O^{\mathrm{(unk)}} \rho) \right|^2\\
    &\leq (2/3)^k \norm{O^{\mathrm{(unk)}}}^2 + \sum_{P: |P| \leq k} \E_{\rho \sim \mathcal{D}^0} \left[ \gamma^*(\rho_{\mathsf{dom}(P)}) \right] \left( \frac{2}{3} \right)^{|P|} \left|\hat{\alpha}_P(O) -\alpha_P(O) \right|^2 \\
    &\leq \frac{\epsilon}{2} \norm{O^{\mathrm{(unk)}}}^2 + \sum_{P: |P| \leq k} \E_{\rho \sim \mathcal{D}^0} \left[ \gamma^*(\rho_{\mathsf{dom}(P)}) \right] \left( \frac{2}{3} \right)^{|P|} \left|\hat{\alpha}_P(O) -\alpha_P(O) \right|^2.
\end{align}
Using the definition of $O^{\mathrm{(unk)}}$, we have $O^{\mathrm{(unk)}} = \mathcal{E}^\dagger(O)$ and $\norm{O^{\mathrm{(unk)}}} \leq \norm{O}$. Hence
\begin{equation}
    \E_{\rho \sim \mathcal{D}^0} \left| h(\rho, O) - \Tr(O \mathcal{E}(\rho)) \right|^2 \leq \frac{1}{2} \left(\epsilon +  \frac{\epsilon}{2^{k}} \right) \norm{O}^2,
\end{equation}
which establishes a prediction error bound for distribution $\mathcal{D}^0$.

\subsubsection{Prediction error under general distribution $\mathcal{D}$  (second set of hyperparameters)}

We now consider an arbitrary $n$-qubit state distribution $\mathcal{D}$ invariant under single-qubit $H$ and $S$ gates.
Using Lemma~\ref{lem:mse-D} on mean squared error, Corollary~\ref{cor:low-deg-approx} on low-degree approximation, $k = \lceil \log_{1.5}(2 / \epsilon) \rceil$, the fact that $\gamma^*(\rho_{\mathsf{dom}(P)}) \leq 1$, and $\E_{\rho \sim \mathcal{D}_0} \left[ \gamma^*(\rho_{\mathsf{dom}(P)}) \right] \left( \frac{2}{3} \right)^{|P|} = (1 / 3)^{|P|}$, we have
\begin{equation}
    \E_{\rho \sim \mathcal{D}} \left| h(\rho, O) - \Tr(O^{\mathrm{(unk)}} \rho) \right|^2 \leq \frac{\epsilon}{2} \norm{O}^2 + 2^k \sum_{P: |P| \leq k} \E_{\rho \sim \mathcal{D}^0} \left[ \gamma^*(\rho_{\mathsf{dom}(P)}) \right] \left( \frac{2}{3} \right)^{|P|} \left|\hat{\alpha}_P(O) -\alpha_P(O) \right|^2.
\end{equation}
Combining the above with Eq.~\eqref{eq:error-bound-D0-prime}, we have
\begin{equation}
    \E_{\rho \sim \mathcal{D}} \left| h(\rho, O) - \Tr(O \mathcal{E}(\rho)) \right|^2 \leq \epsilon \norm{O}^2,
\end{equation}
which is the prediction error under distribution $\mathcal{D}$.

\subsubsection{Putting everything together (second set of hyperparameters)}

From Eq.~\eqref{eq:tilde-eps-def-prime}, we have set the parameter $\tilde{\epsilon}$ to be
\begin{equation}
    \tilde{\epsilon} = \frac{\epsilon}{9 \cdot 2^{k+1} \cdot n^{k}} \left( \frac{C(\kappa, d)}{3} \right)^{2}.
\end{equation}
Furthermore, given the classical shadow $S_N(\mathcal{E})$ of size
\begin{equation}
    N = \mathcal{O}\left( \frac{\log(n^{k + \kappa} / \delta)}{\tilde{\epsilon}^2} \right) = \log\left(\frac{n}{\delta}\right) \, 2^{\mathcal{O}\left(\log(\frac{1}{\epsilon}) \log(n) \right)},
\end{equation}
we can guarantee that with probability at least $1 - \delta$, the following holds.
For any observable $O$ that is a sum of $\kappa$-qubit observables, where $\kappa = \mathcal{O}(1)$ and each qubit is acted on by $d = \mathcal{O}(1)$ of the $\kappa$-qubit observables, and any $n$-qubit state distribution $\mathcal{D}$ invariant under single-qubit $H$ and $S$ gates, we have
\begin{equation}
    \E_{\rho \sim \mathcal{D}} \left| h(\rho, O) - \Tr(O \mathcal{E}(\rho)) \right|^2 \leq \epsilon \norm{O}^2.
\end{equation}
This completes the proof of Theorem~\ref{thm:main-learning-evo}.

\section{Numerical details}
\label{sec:num-details}

In the numerical experiments, we consider the two classes of Hamiltonians,
\begin{align}
    H &= \frac{1}{4} \sum_{i} (X_i X_{i+1} + Y_i Y_{i+1}) + \frac{1}{2} \sum_i h_i Z_i, &\mbox{(XY model)}\\
    H &= \frac{1}{2} \sum_{i} X_i X_{i+1}  + \frac{1}{2} \sum_i h_i Z_i, &\mbox{(Ising model)}
\end{align}
where $h_i = 0.5$ for the homogeneous Z field, and $h_i$ is sampled uniformly at random from $[-5, 5]$ for the disordered Z field.
We solve for the time-evolved properties using the Jordan-Wigner transform to map the spin chains to a free fermion model and the technique described in \cite{lieb1961two} to solve the free fermion model.

We consider the training set to be a collection of $N$ random product states $\ket{\psi_\ell}, \ell = 1, \ldots, N$ and their associated measured properties $y_\ell$ corresponding to measuring an observable $O$ after evolving under $U(t) = \exp(-i t H)$.
The measured properties are averaged over $500$ measurements. Hence $y_\ell$ is a noisy estimate of the true expectation value $\Tr(O U(t) \ketbra{\psi_\ell}{\psi_\ell} U(t)^\dagger)$.
We consider essentially the same ML algorithm as described in Section~\ref{sec:ML-algo}, but utilize a more sophisticated approach to enforce sparsity in $\hat{\alpha}_P$.
We also consider $\alpha_P$ for Pauli operator $P$ that is geometrically local.
For ease of analysis, we consider a simple strategy of setting small values to zero.
The standard approach that is often used in practice is LASSO \cite{tibshirani1996regression}.

In the numerical experiments, we perform a simple grid search for the two hyperparameters using two-fold cross-validation on the training set:
\begin{align}
    k &= 1, 2, 3, 4, \\
    a &= 2^{-15}, 2^{-14}, 2^{-13}, \ldots, 2^{-4}, 2^{-3},
\end{align}
where $k$ corresponds to the maximum number of qubits that the Pauli operators $P$ act on, and $a$ is a hyperparameter corresponding to the strength of the $\ell_1$ regularization term in LASSO.
In particular, the optimization problem of LASSO is given by
\begin{equation}
    \min_{\hat{\alpha}_P} \frac{1}{2 N} \sum_{\ell=1}^N \left| y_\ell - \sum_{P: |P| \leq k} \hat{\alpha}_P \Tr(P \ketbra{\psi_\ell}{\psi_\ell}) \right|^2 + a \sum_{P: |P| \leq k} |\hat{\alpha}_P|,
\end{equation}
where $|P|$ is the number of qubits that the Pauli observable $P$ acts nontrivially on.
We then use the values $\hat{\alpha}_P$ found by the above optimization to form a succinct approximate model
\begin{equation}
 \sum_{P: |P| \leq k} \hat{\alpha}_P P
\end{equation}
of the time-evolved observable $O(t) = U(t)^\dagger O U(t)$.
Given a new initial state $\rho$, we would predict the time-evolved property $\Tr(O(t) \rho) = \Tr(O U(t) \rho U(t)^\dagger)$ using
\begin{equation}
    \sum_{P: |P| \leq k} \hat{\alpha}_P \Tr(P \rho).
\end{equation}
In addition to the figures given in the main text, Fig.~\ref{fig:GHZ2} shows another example for predicting a highly entangled initial state.
Even though the ML model is trained with random product states, it still performs very well on a structured entangled initial state.

\begin{figure*}[t]
\centering
\includegraphics[width=1.0\textwidth]{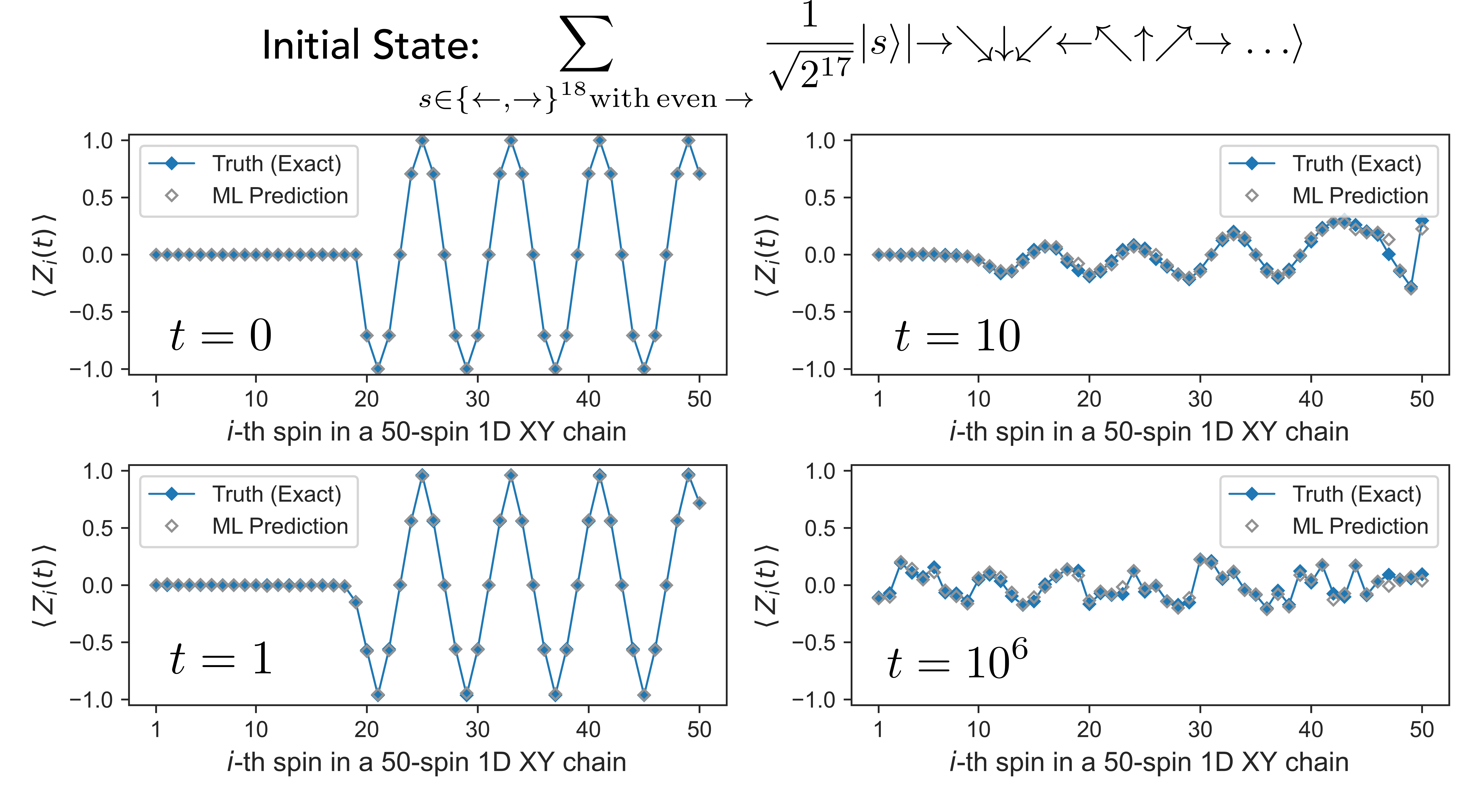}
    \caption{
    \emph{Visualization of ML model's prediction for a highly-entangled initial state $\rho = \ketbra{\psi}{\psi}$.}
    We consider the expected value of $Z_i(t) = e^{i t H} Z_i e^{-i t H}$, where $H$ corresponds to the 1D 50-spin XY chain with a homogeneous $Z$ field.
    The initial state $\ket{\psi}$ has a GHZ-like entanglement over the first $18$-spin chain and is a product state with spins rotating clockwise over the latter $32$-spin chain.
    To prepare $\ket{\psi}$ with 1D circuits, a depth of at least $\Omega(n)$ is required.
    Even though the ML model is trained only on random product states (a total of $N = 10000$), it still performs accurately in predicting the highly-entangled state over a wide range of evolution time $t$.
    \label{fig:GHZ2}}
\end{figure*}

\end{document}